\documentclass[11pt]{article}
\usepackage[textsize=footnotesize]{todonotes}
\usepackage{hyperref}

\usepackage{mathrsfs}
\usepackage{graphicx}
\usepackage{amsmath}        
\usepackage{amssymb}
\usepackage{amsthm}
\usepackage[countmax]{subfloat}  
\usepackage{fourier}
\usepackage{xcolor}
\usepackage{mathtools}
\usepackage{enumerate}
\usepackage{tabu}
\usepackage{chngcntr}
\counterwithin{figure}{section}
\usepackage{natbib}
\usepackage{url}
\usepackage{verbatim}
\usepackage{booktabs}
\usepackage{placeins}
\usepackage{multirow}
\usepackage{caption}
\usepackage[title]{appendix}
\usepackage{subcaption}

\def\E{\mathbb E}
\def\V{\text{Var}}
\def\cv{\text{Cov}}

\def\cm{\text{Cum}}
\graphicspath{{plots/}}

\usepackage{array}
\newcolumntype{$}{>{\global\let\currentrowstyle\relax}}
\newcolumntype{^}{>{\currentrowstyle}}

\usepackage{calc}
\usepackage{tikz}
\usetikzlibrary{decorations.markings}
\tikzstyle{vertex}=[circle, draw, inner sep=0pt, minimum size=6pt]
\newcommand{\vertex}{\node[vertex]}

\newtheorem{prop}{Proposition}[section]     
\newtheorem{thm}{Theorem}[section]
\newtheorem{lemma}{Lemma}[section]

\newtheorem{Remark}{Remark}[section]
\newtheorem{Corollary}{Corollary}[section]
\newtheorem{assumption}{Assumption}[section]
\newtheorem{properties}{Properties}[section]

\theoremstyle{definition}
\newtheorem*{exmp}{Example}

\newcommand{\im}{\mathrm{i}}

\newcommand{\mynegspace}{\hspace{-0.12em}}

\newcommand{\bigsnorm}[1]{\Big\rvert\mynegspace\Big\rvert\mynegspace\Big\rvert\mynegspace {#1} \Big\rvert\mynegspace \Big\rvert\mynegspace \Big\rvert}
\newcommand{\snorm}[1]{\rvert\mynegspace\rvert\mynegspace\rvert\mynegspace {#1} \rvert\mynegspace \rvert\mynegspace \rvert}

\newcommand{\opnorm}[1]{{\|}{#1}{\|}_{\infty}}

\newcommand{\Bignorm}[1]{\Bigg{\|}{#1}\Bigg{\|}}

\newcommand{\innprod}[2]{\langle #1, #2 \rangle}
\newcommand{\biginnprod}[2]{\Big\langle #1, #2 \Big\rangle}

\newcommand{\fdft}[3]{ D_{#3}^{{u_{j_{{#1}}}},{\omega_{k_{#2}}}}}
\newcommand{\fdftc}[3]{ D_{#3}^{{u_{j_{{#1}}}},{-\omega_{k_{#2}}}}}
\newcommand{\fdftl}[3]{ D_{#3}^{{u_{j_{{#1}}}},{\omega_{k_{#2}-1}}}}
\newcommand{\fdftcl}[3]{ D_{#3}^{{u_{j_{{#1}}}},{-\omega_{k_{#2}-1}}}}
\newcommand{\rnum}{\mathbb{R}}
\newcommand{\znum}{\mathbb{Z}}
\newcommand{\nnum}{\mathbb{N}}
\newcommand{\XT}[1]{X_{#1, T}}
\newcommand{\Xu}[1]{X^{(u)}_{#1}}

\newcommand{\cumbp}[2]{\kappa_{{#1};t_1,\ldots,t_{{#2}-1}}}

\newcommand{\F}{\mathcal{F}}
\newcommand{\Eps}{\mathcal{E}}
\newcommand{\A}{\mathcal{A}}

\newcommand{\cnum}{\mathbb{C}}
\newcommand{\Vn}{\hat{\mathcal{U}}}
\newcommand{\mH}{\mathcal{H}}
\allowdisplaybreaks

\newcommand\tageq{\addtocounter{equation}{1}\tag{\theequation}}
\DeclareMathOperator{\Tr}{Tr}
\DeclareMathOperator*{\argmin}{arg\,min}

\allowdisplaybreaks

\usepackage{authblk}

\begin{document}

\title{A similarity measure for second order properties of non-stationary functional time series with applications to clustering and testing}

\author{Anne van Delft}
\author{Holger Dette}
\affil{Ruhr-Universit\"at Bochum,\\ Fakult\"at f\"ur Mathematik, \\44780 Bochum, Germany}

 \maketitle
\begin{abstract}
Due to the surge of data storage techniques, the need for the development of appropriate techniques to identify patterns and to extract knowledge from the resulting enormous data sets, which can be viewed as collections of dependent functional data, is of increasing interest in many scientific areas. 
We develop a similarity measure for spectral density operators of a collection of functional time series, which is based on the aggregation of Hilbert-Schmidt differences of the individual time-varying spectral density operators. Under fairly general conditions, the asymptotic properties of the corresponding estimator are derived and asymptotic normality is established. The introduced statistic lends itself naturally to quantify  (dis)-similarity between functional time series, which we subsequently exploit in order to build a spectral clustering algorithm. Our algorithm is the first of its kind in the analysis of non-stationary (functional) time series and enables to discover particular patterns by grouping together `similar' series into clusters, thereby reducing the complexity of the analysis considerably. The algorithm is simple to implement and computationally feasible. 
As a further application we provide a simple  test for the hypothesis that the second order properties of two non-stationary functional time series coincide.

\end{abstract}

 
\section{Introduction}  \label{sec1}
\def\theequation{1.\arabic{equation}}
\setcounter{equation}{0}

The surge in data storage techniques over the past two decades has led to more and more data sets that are almost continuously recorded from their domain of definition. The development of tools to model these type of data is the main focus of functional data analysis. In functional data analysis, the variables of interest are perceived as random smooth functions that vary on a continuum $D$,  i.e., $X(\tau), \tau \in D$. While the intrinsically infinite variation of such random functions can be considered a rich source of information, extracting relevant information and to identify patterns becomes more and more a challenge. Especially when the data is collected sequentially over time and the curves exhibit serial dependence, i.e., when the data set consists of a collection of $d$ \textit{functional time series}, $\{X_{i,t}(\tau): \tau \in D\}_{t \in \mathbb{Z}, i \in 1,\ldots, d}$. This type of data arises naturally in a wide range of scientific disciplines such as  astronomy, biology, finance, meteorology, medicine or yet engineering. 
In addition, in most real-world applications, the second order characteristics of time series change gradually over time. In meteorology, the distribution of the daily records of temperature, precipitation and cloud cover for a region, viewed as three related functional surfaces, may change over time due to global climate changes. Other relevant examples appear in the study of cognitive functions such as high-resolution recordings from local field potentials, EEG and MEG or from the financial industry where implied volatility of an option as a function of moneyness changes over time.  The development of appropriate exploratory techniques that allow to discover patterns or anomalies is therefore of foremost interest for this type of data. 

The most widely used technique for this preliminary step of data exploration is known as \textit{cluster analysis}. Clustering is concerned with partitioning a data-set into a set of disjoint homogeneous groups (clusters) of realizations. Unlike supervised learning, clustering does not rely on prior knowledge of the groups or on building classifiers based upon a training set. It is therefore especially meaningful when little is known about the nature of the process and the data set is large. 

A large body of literature on clustering (and related learning techniques) of Euclidean-valued time series has been published. Depending on the goal of the application, clustering algorithms can differ in a variety of aspects such as the representation of the data, how similarities are measured and the way clusters are constructed.
 The first two aspects are known to be crucial in terms of efficiency and accuracy of the solution and this is where most research focuses on\citep[see section 5 of][for a full taxonomy of the different aspects of clustering time series]{ASW2015}. 
For instance, in parametric approaches, clusters are usually built based upon similarity of their estimated parameters \citep[see e.g.,][]{KalGP01,CoPic08} some of which take a Bayesian approach \citep[][]{BaRom07,FrKa08,JuSt10}. Nonparametric methods are often based upon comparing similarity of the estimated power spectra, which is a research topic on its own \citep[see e.g.,][]{CoDig86,Eichler08,DetPap09,DetHil12,JP2015}. This approach is taken in, i.a., \citet{KaShTa98,SaPF08,FokPro11,HolRav18}. A wavelet-based approach can be found in \citet{VlLK03}. 

Clustering and classification methods have also been extended to non-stationary time series. For example, \citet{ST04} use the framework of locally stationary time series \citep{dahlhaus1997} for clustering while \citet{CP06} use it to develop a shape-based approach discriminant analysis. Another branch of literature focuses on piecewise stationary processes using Smooth Localized Complex EXponentials (SLEX) transforms, which were introduced by \citet{OmRvSM01}, or variations thereof \citep[see e.g.,][]{HuOmSt04,Harvill2017} for clustering approaches and \citet{BoOmvSS10} for classification of multivariate series. 

In contrast to the Euclidean case, the literature on cluster analysis for functional data is not that rich. Some methods have been developed for clustering of i.i.d. functional data \citep[and references therein]{JP14}. A popular  technique is to first reduce dimension by projecting the curves onto a basis of finite dimension and then apply a standard classical clustering algorithm such as $k$-means \citep[see][]{PengMul08,Abretal03}. Alternatively, nonparametric methods have been proposed that use specific distances for functional data \citep[][]{FerVi06,Ievaetal13} and parametric (Bayesian) approaches assuming a particular probability distribution for multivariate functional data \citep[see ][among others]{JacPr14,HHS06}.

Despite of the vast amount of literature available on various  data structures, existing methods are inappropriate to cluster possibly non-stationary functional time series. A clustering technique for functional time series requires to take into account its infinite variation and hence a clustering approach must be able to capture the complex within-curves dynamics as well as the between-curve dynamics. At the same time, it needs to be efficient to apply because of the high dimensionality of the data. 
In this article, we address this problem from several perspectives. We develop a new measure to compare the second order properties of non-stationary functional time series. This measure is based on the aggregation of Hilbert-Schmidt differences of the individual time-varying spectral density operators. Under fairly general conditions, the asymptotic properties of the corresponding estimator are derived and asymptotic normality is established. 
We then use this methodology for two purposes. Firstly, we consider this measure and its estimate to develop a new spectral clustering algorithm for functional time series, which are allowed to be non-stationary and nonlinear. Our algorithm is novel in the sense that not only is it the first of its kind for exploratory analysis of functional time series but moreover because, to our knowledge, spectral clustering has also not yet been considered for Euclidean time series. The underlying principle of spectral clustering is to reformulate the problem into a graph partitioning problem (see \autoref{fig:example}) 
Geometrical properties of graphs can be conveniently described by the spectral properties of the corresponding graph Laplacian \citep[see e.g.,][]{Chung1997,CR2011,VLux07}. 
Using these properties, we can detect clusters in non-convex regions, which classical clustering techniques may not be able to find. Furthermore, it can be solved efficiently via classical linear algebra operations. 
 We will show that our introduced measure of similarity provides a meaningful basis for the adjacency matrix underlying the graph Laplacian.  Secondly, we use this measure to develop a particularly simple level $\alpha$-test for the hypothesis of equality of two time-varying spectral density operators, which uses the quantiles of the standard normal distribution.

The structure of this paper is as follows. We first introduce necessary notation and background on the type of processes considered in this paper. We then define a measure of similarity 
for a pair of functional time series
and derive a consistent estimator to construct an empirical adjacency matrix. In \autoref{sec3}, the spectral clustering algorithm is discussed in  detail  and it is shown that the algorithm based upon an empirical graph Laplacian, a transformation of the estimated adjacency matrix, is consistent.  In \autoref{sec5}, we discuss the application of hypothesis testing, whereas in \autoref{sec4}  we study the properties of our algorithm in finite samples. In \autoref{sec6}, the clustering method is illustrated by means of an application to high-resolution meteorological data. All proofs and technical assumptions for the main statements are relegated to the Appendix.

\section{A measure of similarity}  \label{sec2}
\def\theequation{2.\arabic{equation}}
\setcounter{equation}{0}

In this section we introduce a measure of similarity for functional time series, which is appropriate to use as a basis for a similarity matrix to cluster functional time series.

\begin{exmp}
We have generated  $90$ functional time series uniformly from  models I, II, III (described in detail in Section \ref{sec5}), which should be clustered according to their second order properties.  To visualize the difficulties of this task we exemplary depict $9$  series in  Figure \ref{fig:FTSplots}, where we took three from each distribution. As can be seen, the high dimensionality of the data makes a visual assessment of their second order properties almost impossible for 6 of them, while 3 of them appear more obvious to distinguish from the other 6.
  \begin{figure}[h!]
    \centering
    \includegraphics[width=1.0\textwidth]{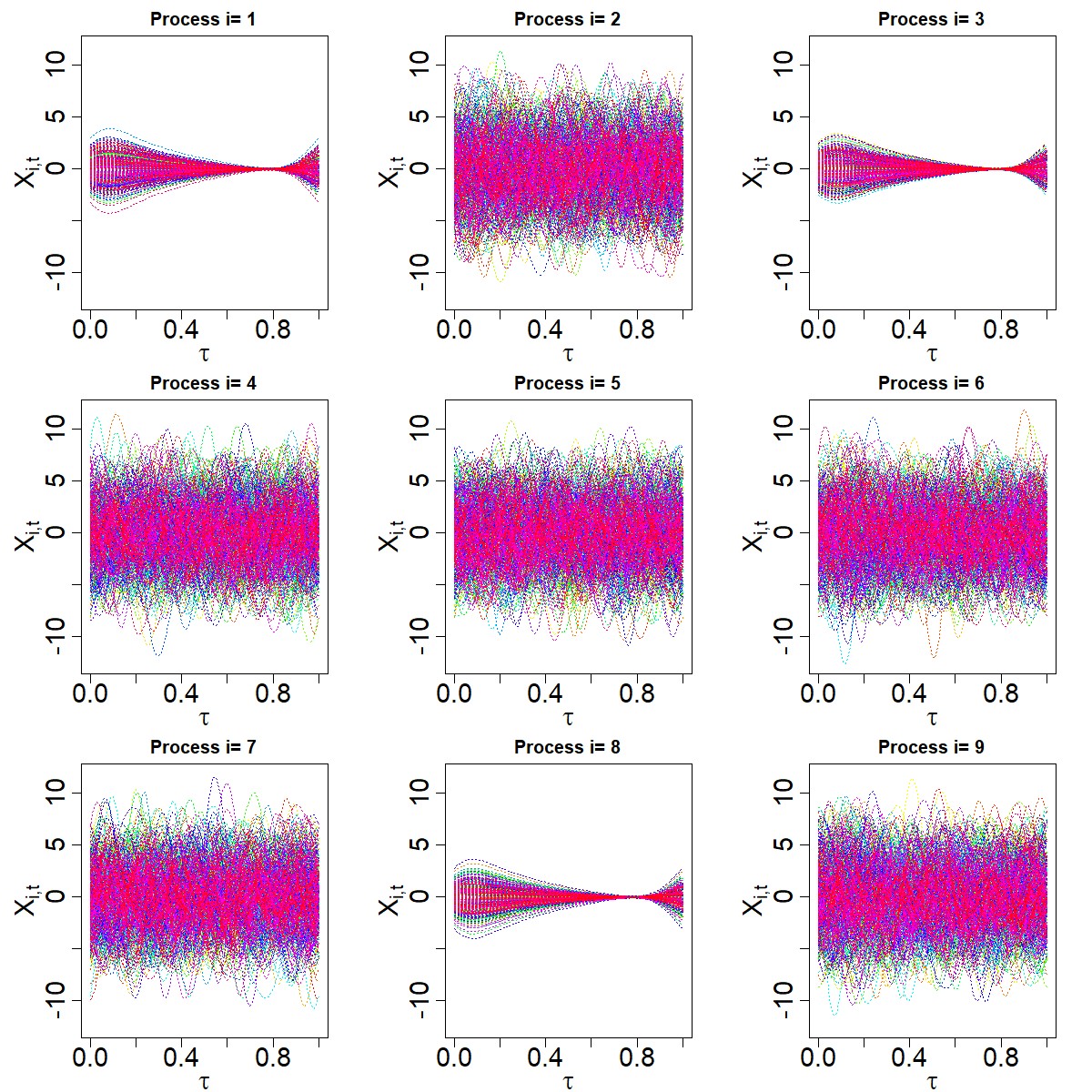}
    \caption{\it Functional time series from $3$ different distributions, $T=512$.}%
    \label{fig:FTSplots}
\end{figure} 
\end{exmp}
\setcounter{figure}{2}

\subsection{Notation}

First, let us introduce some necessary notation. For a separable Hilbert space $\mH$, we denote the inner product as $\innprod{\cdot}{\cdot}: \mH \times \mH \to \mathbb{C}$ and its induced norm by $\|\cdot\|$. The Banach space of bounded linear operators $A: \mH \to \mH$ with operator norm  $\snorm{A}_\infty=\sup_{\|x\|\le1}\|Ax\|$ shall be denoted by $\mathcal L(\mathcal H)$, the adjoint of $A \in \mathcal L({\mH})$ by $A^\dagger$ 
. $A\in\mathcal L(\mathcal H)$ is called self-adjoint if $A=A^\dagger$ and non-negative definite if $\langle Ax,x\rangle\ge0$ for each $x\in\mathcal H$. If well-defined we denote the trace of $A \in \mathcal{L}(\mH)$ by $\Tr(A)$. 
A compact operator $A\in\mathcal L(\mathcal H)$ belongs to the Schatten class of order $1\le p< \infty$, denoted by $A\in S_p(\mathcal H)$, if $\vert\kern-0.25ex\vert\kern-0.25ex\vert A\vert\kern-0.25ex\vert\kern-0.25ex\vert_p^p=\sum_{j\ge1}s_j^p(A)<\infty$, where $\{s_j(A):j\ge1\}$ are the singular values of $A$. Operators that belong to the Banach spaces $(S_1(\mathcal H),\snorm{\cdot}_1)$ or $(S_2(\mathcal H),\snorm{\cdot}_2)$ will be referred to as trace-class operators and Hilbert-schmidt operators, respectively. We remark in particular that $(S_2(\mathcal H),\snorm{\cdot}_2)$ is a Hilbert space with the inner product given by $\langle A,B\rangle_{\mathrm HS}=\Tr(A B^\dagger)=\sum_{j\ge1}\langle Ae_j,Be_j\rangle$ for each $A,B\in S_2(\mathcal H)$ and $\{e_j\}_{j \in \mathbb{N}}$ an orthonormal basis of $\mathcal{H}$. For $f,g \in \mathcal{H}$, we define the tensor product $f\otimes g: \mathcal{H}  \otimes \mathcal{H}\to \mathcal{H} $ as the bounded linear operator defined by
\[
(f \otimes g) v = \innprod{v}{g}f \quad  \forall v \in \mathcal{H}.
\]  
Without loss of generality, we consider the Hilbert space $\mathcal{H}=\L^2_{\cnum}([0,1])$ of equivalence classes of square integrable measurable functions $f:[0,1] \to\mathbb C$ with inner product $\int_0^1 f(\tau) \overline{g(\tau)} d\tau$, where the complex conjugate of $x \in \cnum$ is denoted as usual by $\overline{x}$. Since the mapping  $\mathcal{T}:\mathcal{H} \otimes\mathcal{H} \to S_2(\mathcal{H})$ defined by the linear extension of $ \mathcal{T}(f \otimes g) = f \otimes \overline{g} $ is an isometric isomorphism, it defines a Hilbert-Schmidt operator with the kernel in $\mathcal{H}$ given by $(f\otimes g)(\tau,\sigma)=f(\tau)\overline{g(\sigma)}$ for each $\tau,\sigma\in[0,1]$ in an $L^2$-sense. We refer to Appendix \autoref{ap:background} for further details and background. 

\subsection{A measure of similarity for functional time series}

The main object of this paper are $d$ zero-mean stochastic processes $X_i =\{X_{i,t,T}\}_{t=1,\dots,T; T \in \mathbb{N}}$ $i =1, \ldots, d$ that take values in the Hilbert space $L^2_{\rnum}([0,1])$. The second order dynamics are assumed to be  finite but allowed to change over time. Processes of this type fit the framework of {\em locally stationary functional time series} as defined in \cite{vde16} --which extends the concept of local stationarity \citep{dahlhaus1997} to the function space-- and contain weakly stationary functional processes as a subclass. Asymptotic properties can be described by so-called infill-asymptotics, such that, as $T \to \infty$, we obtain more and more observations at a local level. Formally, a process $\{X_{t,T}\}_{t=1,\dots,T; T \in \mathbb{N}}$ is called {\em functional locally stationary} if, for all rescaled times $u\in[0,1]$, there exists an $L^2_{\mathbb{R}}([0,1])$-valued strictly stationary process $\{X^{(u)}_t:t\in\mathbb Z\}$ such that
\begin{equation}\label{Local stationarity}
\Bigl\|\XT{t}-\Xu{t}\Bigr\|_{2}
\leq\big(\big|\tfrac{t}{T}-u\big|+\tfrac{1}{T}\big) \,P_{t,T}^{(u)}\qquad a.s.
\end{equation}
for all $1\leq t\leq T$, where $\{P_{t,T}^{(u)}\}_{t=1,\dots,T; T \in \mathbb{N}}$ is a positive real-valued process such that for some $\rho>0$ and $C<\infty$ the process satisfies $\E\big(\big|P_{t,T}^{(u)}\big|^\rho\big)<C$ for all $t$ and $T$ and uniformly in $u\in[0,1]$.  We refer to \cite{vde16} for further details.

What we will exploit throughout this paper is that the full second order dynamics of the the triangular array $\{X_{t,T}\}_{t=1,\dots,T; T \in \mathbb{N}}$ are completely and uniquely characterized by the {\em time-varying spectral density operator}
\begin{align}\label{eq:spdens}
\F_{u,\omega} = \frac{1}{2 \pi}\sum_{h\in \mathbb{Z}}
\E\big(X^{(u)}_{t+h} \otimes X^{(u)}_{t}\big)
e^{-\im  \omega h}
\end{align}
for each $u\in[0,1]$ and $\{X_{t}^{(u)}:t\in\mathbb Z\}$. In particular, for weakly stationary functional time series we can drop the dependence on local time and thus $\F_{u,\omega} \equiv \F_{\omega}$. This uniquely characterizing object of a functional time series lends itself naturally as a basis for a measure of similarity. \\

More specifically, let, ${\F}^{(i_1)}_{u,\omega}$ and ${\F}^{(i_2)}_{u,\omega}$ denote the time-varying spectral density operator of processes $\{X_{i_1,t,T}\}_{t=1,\dots,T; T \in \mathbb{N}}$ and $\{X_{i_2,t,T}\}_{t=1,\dots,T; T \in \mathbb{N}}$, respectively. As a measure of pairwise similarity between two functional time series, we consider
\begin{equation}\label{eq:dist}
\A_{i_1,i_2} = \frac{ \int_0^1 \int_{-\pi}^{\pi} \snorm{ {\F}^{(i_1)}_{u,\omega} -{\F}^{(i_2)}_{u,\omega}}_2^2 du d\omega}{\int_0^1 \int_{-\pi}^{\pi}\snorm{{\F}^{(i_1)}_{u,\omega}}^2_2+\snorm{{\F}^{(i_2)}_{u,\omega}}_2^2  du d\omega}.
\end{equation}
Clearly, if this distance is zero, then processes $\{X_{i_1,t,T}\}_{t=1,\dots,T; T \in \mathbb{N}}$ and $\{X_{i_2,t,T}\}_{t=1,\dots,T; T \in \mathbb{N}}$ must have the same second order properties and hence must belong to the same cluster. Note that \eqref{eq:dist} takes values in the interval $[0,1)$. 
The local scaling via the denominator is an essential aspect for its usage in a spectral clustering procedure. Differences in scales can lead spectral clustering to fail. While most similarity graphs are based upon a global scaling parameter of which the optimal value is difficult to determine and can highly affect the clustering performance [see \cite{VLux07}], we find that making the measure scale-invariant by accounting for local scales avoids this issue. We discuss this further in \autoref{sec4}. 

 Since ${\F}^{(i_1)}_{u,\omega}$ and ${\F}^{(i_2)}_{u,\omega}$ are unknown, we will have to estimate \eqref{eq:dist}. 
Similar to \citet{vDCD18}, we shall do this using integrated periodogram tensors. That is, we split the sample into $M$ blocks with $N$ elements inside each of these blocks so that $T=MN=M(T)N(T)$ for each  $T\in\mathbb N$, where $M\in\mathbb N$ and $N$ is an even number. $M$ and $N$ will correspond to the number of terms used in a Riemann sum approximating the integrals in \eqref{eq:dist} with respect to $du$ and $d\omega$ and therefore they have to be reasonable large. The functional discrete Fourier transform (fDFT) at time point $u$, is a random function with values in $L^2([0,1],\mathbb C)$ defined by
\begin{equation}\label{eq:fDFT}
	D_i^{u,\omega}
    :=\frac1{\sqrt{2\pi N}}\sum_{s=0}^{N-1}X_{i,\lfloor uT \rfloor - N/2 +s+1,T}e^{-\im  \omega s}.
\end{equation}
 The local periodogram tensor for the $i$-th time series can be given as
\begin{equation}\label{eq:per}
I_i^{u,\omega} := D_i^{u,\omega}\otimes {D_i^{u,\omega}},
\end{equation}
for $i=1, \dots,d$.  We base our estimator upon a linear combination of the following Hilbert-Schmidt inner products
\begin{equation}
\label{Fij}
F_{i_1i_2} = \frac{1}{T}\sum_{j=1}^M \sum_{k=1}^{\lfloor N/2 \rfloor} \langle I_{i_1}^{u_j,\omega_k},I_{i_2}^{u_j,\omega_{k-1}}\rangle_{HS}.
\end{equation}
In particular, a suitable and (symmetric) estimator for the distance \eqref{eq:dist} is given by
\begin{equation}
\label{eq:est_dist}
\hat{\A}_{i_1,i_2} := \frac{F_{i_1i_1} + F_{i_2i_2} - F_{i_1i_2} - F_{i_2i_1}}{(F_{i_1i_1} + F_{i_2i_2})},
\end{equation}
To ease notation, we provide empirical quantities with $\hat{\cdot}$. The dependence of these quantities on $T$ is therefore implicit. We find under suitable regularity conditions, which are postponed to section \autoref{sec:dist},
\begin{thm}[consistency] \label{thm:con}
Under \autoref{cumglsp} with $m=8$ and \autoref{ratesNM}, we find that 
\[
\hat{\A}_{i_1,i_2} - \A_{i_1,i_2} = O_p({T}^{-1/2}).
\]
\end{thm}

We remark that, under suitable moment conditions, it is moreover asymptotically multivariate normal and therefore lends itself for other statistical applications such as a test for equality of time-varying spectral density operators. We shall briefly discuss this application together with more details on the statistic $\hat{\A}_{i_1,i_2} $ in \autoref{sec5}. In the next section, we  define a similarity graph based upon the measure \eqref{eq:est_dist} and introduce a spectral clustering algorithm to cluster the functional time series.



\begin{exmp}[continued]
 For $6$ of the $9$ functional time series depicted in  Figure \ref{fig:FTSplots} it is difficult to visually distinguish the second order properties. This becomes even more difficult for all $90$ series. However, we can use the measure $\hat \A $ to identify similarities. A heat map of the corresponding estimates for all 90 time series are displayed in Figure \ref{fig:ahat}(a). As can be seen the empirical measure gives the pairs of time series different weights ranging from $0$ to $1$,  but it is difficult to identify any structure. For the sake comparison, we ordered the times series  assuming knowledge of the clusters in the right part of Figure \ref{fig:ahat}(b). We observe that the similarity measure $\hat \A $ make the clusters visible.
\begin{figure}[h!]
    \centering
    {\renewcommand{\arraystretch}{0}
    \begin{tabu}{c@{}c}
    \begin{subfigure}[t]{.5\columnwidth}
        \centering
        \includegraphics[width=0.75\textwidth]{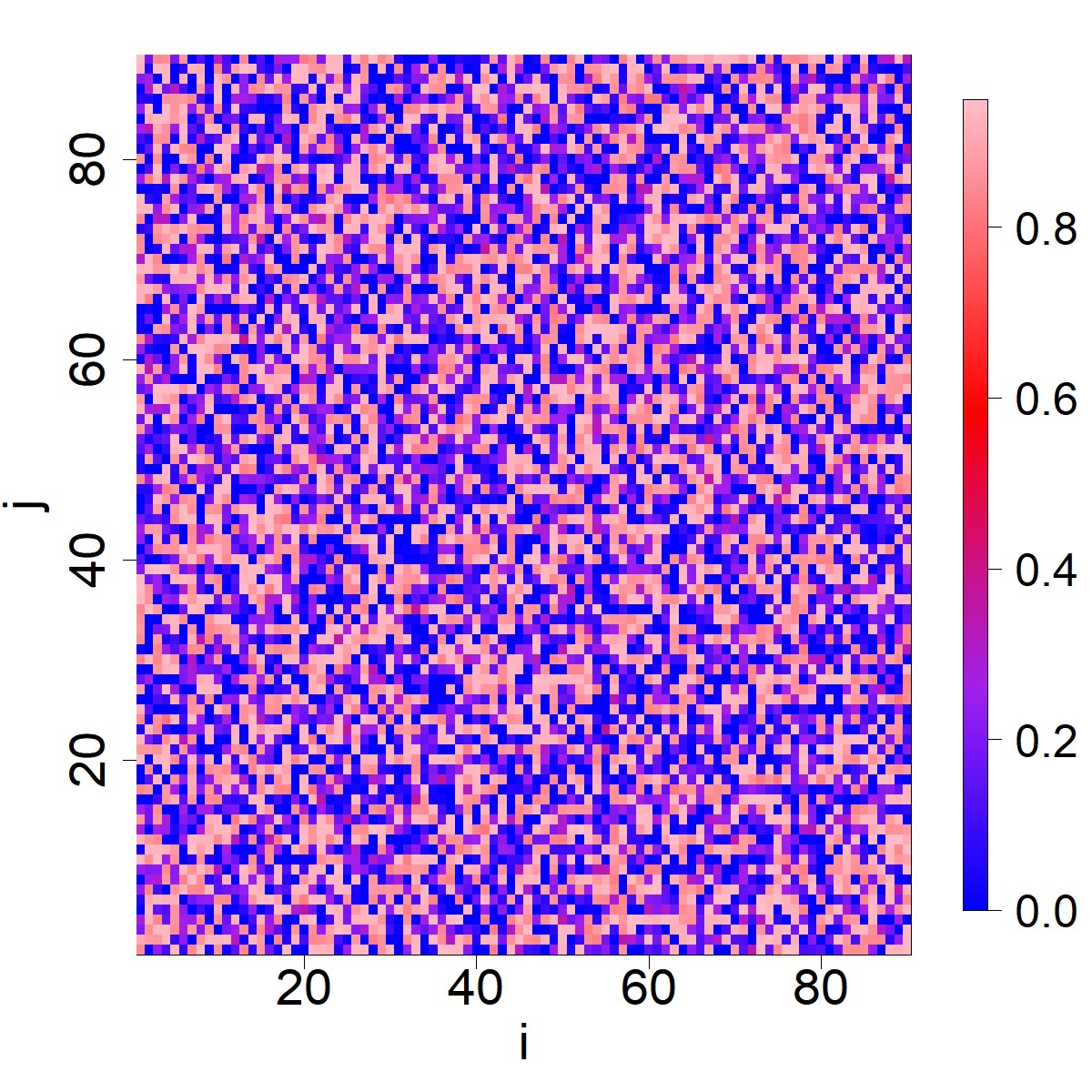}%
          \caption{\it {without information on the clusters}} 
    \end{subfigure}\hfill
    \begin{subfigure}[t]{.5\columnwidth}  
        \centering
        \includegraphics[width=0.75\textwidth]{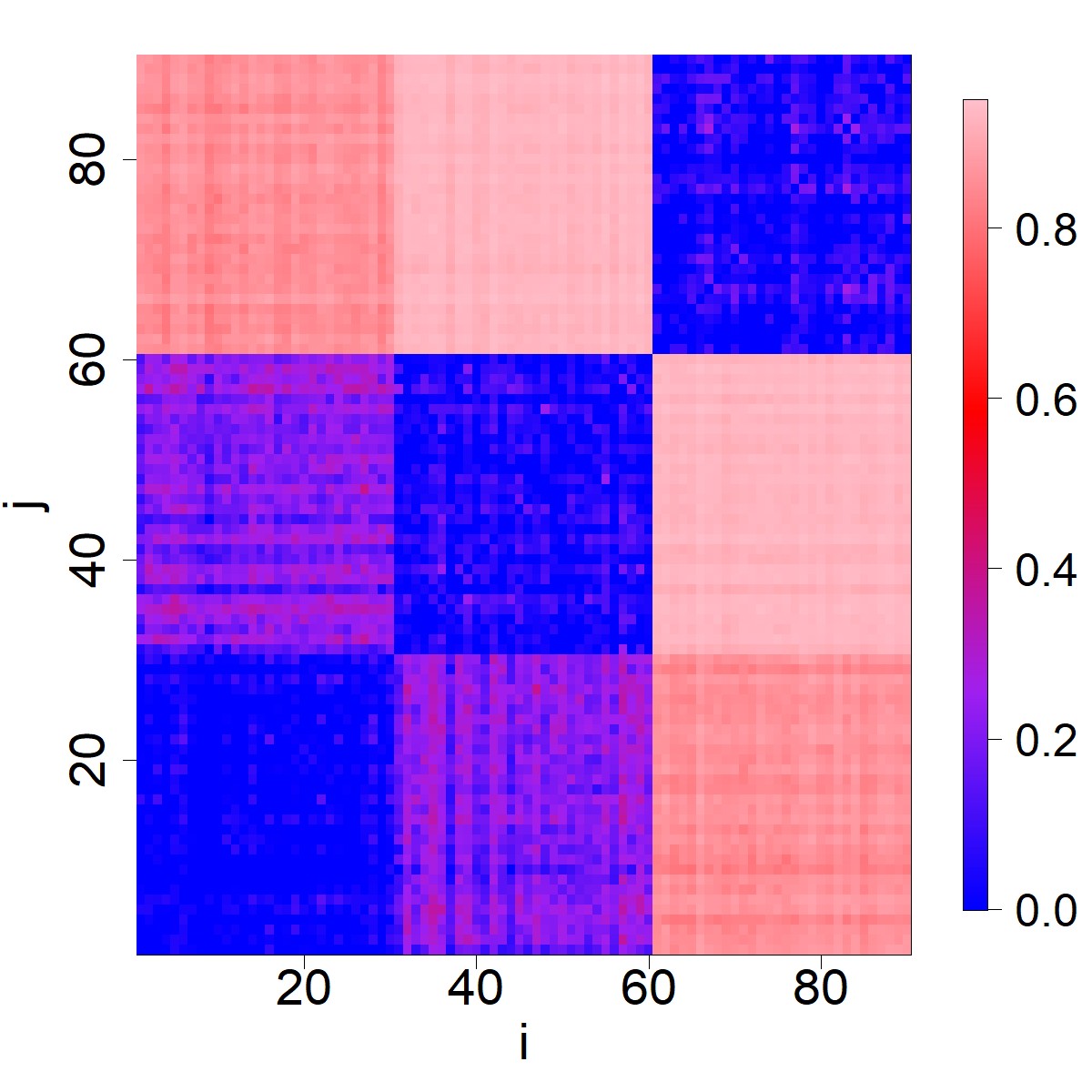}%
        \caption{\it {with information on the clusters}}
    \end{subfigure}
    \end{tabu}}
    \caption{\it  \small Heat map of the value of the estimates $\hat{\A}_{i,j}$ of ${\A}_{i,j}$  plotted for all pairs $i, j =1,\ldots,90$ in random order (left); and ordererd by cluster (right).}. 
    \label{fig:ahat}
\end{figure}
 \end{exmp}

\section{Spectral clustering of functional time series}  \label{sec3}
\def\theequation{3.\arabic{equation}}
\setcounter{equation}{0}

In this section, we develop a spectral clustering  algorithm, which 
consists of a several steps.  We start by translating the problem of clustering $d$ functional time series into $k$ clusters into a graph partitioning problem using the previously defined similarity measure. Secondly, we construct a spectral embedding using an empirical graph Laplacian, which is shown to have spectral properties that converge to those of the population graph Laplacian. We then use the embedded points to cluster the data using $k$-means and show that we can bound the number of points that are clustered converges to zero as $T \to \infty$. 

\subsection{Construction of the graph}
We construct an undirected similarity graph $G=(V,E)$, where $V$ denotes the set of vertices and $E$ denotes the set of edges. To each family of random curves $\{X_i\}, i\in [d]$ we relate a node $v_i \in V$ and describe the similarity between node $v_i$ and $v_j$ via non-negative weights on the edges. These weighs are given by the {\em empirical adjacency matrix} which is defined as the following transformation of the matrix $\hat{\A}=
 \{\hat{\A}_{i_1,i_2}: i_1, i_2 =1,\ldots,d\}$ as defined in \eqref{eq:est_dist}
\begin{align} \label{eq:W_T}
\hat{W}=e^{-\hat{\A}} \in  \mathbb{R}^{d \times d}.
\end{align}
 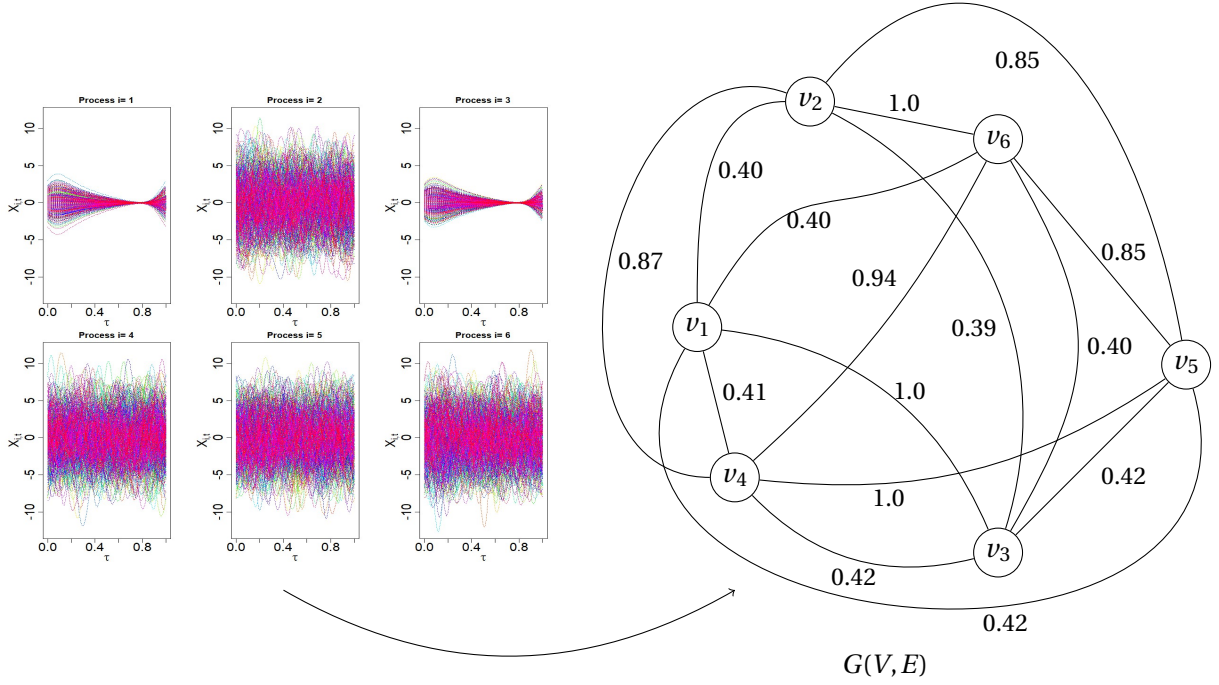
\begin{figure}[h!]%
 \vspace{-50pt}
    \centering
\hspace*{-10pt}
\begin{tikzpicture}
\coordinate (T) at (-6.5,-0.5);
\node (myfirstpic) at (-6.5,3) {\includegraphics[width=7.5cm, height=6.25cm]{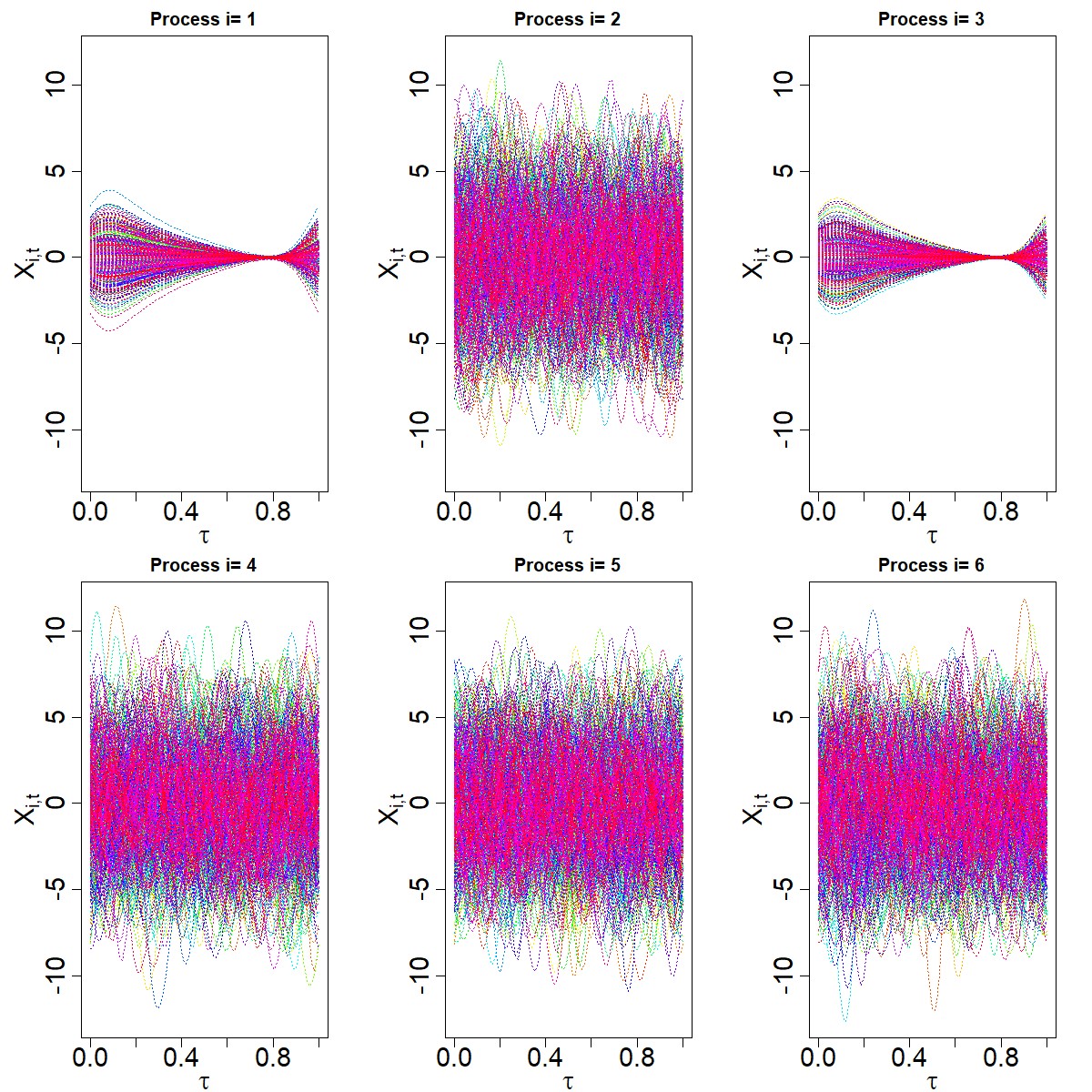}};
	\vertex[minimum size=0.65cm] (v_1) at (-1.0, 3) {$v_1$};
	\vertex[minimum size=0.65cm] (v_2) at (0.5, 6.0){$v_2$};
	\vertex[minimum size=0.65cm] (v_3) at (3, 0){$v_3$};
	\vertex[minimum size=0.65cm] (v_4) at (-0.5, 1) {$v_4$};
	\vertex[minimum size=0.65cm] (v_5) at (5.5, 2.5) {$v_5$};
	\vertex[minimum size=0.65cm] (v_6) at (3.0,5.5){$v_6$};
	\path
		(v_1) edge[out= 90, in=-180,looseness=1]  node[pos=0.5,right]{\small{0.40}}(v_2)
		(v_1) edge[bend left=30] node[pos=.5,right]{\small{1.0}}(v_3)
		(v_1) edge node[pos=.4,right]{\small{0.41}}(v_4)
		(v_1) edge  [out= 240, in=286,looseness=1.75] node[pos=0.6,below]{\small{0.42}} (v_5)
	
		(v_1) edge[bend left=40,out=30, in=180,looseness=1.5] node[pos=0.3,right]{\small{0.40}} (v_6)
	
		(v_2) edge[bend left=40] node[pos=0.55,below left]{\small{0.39}}(v_3)
		(v_2) edge [out= 160, in=180,looseness=1.25] node[pos=0.5,right]{\small{0.87}} (v_4)
		(v_2) edge  [out=50, in=100,looseness=1.5] node[pos=0.5,below left]{\small{0.85}} (v_5)
		(v_2) edge node[pos=0.5,above]{\small{1.0}} (v_6)	
		
		(v_3) edge[bend left=30] node[pos=0.5,below]{\small{0.42}} (v_4)
		(v_3)  edge node[pos=0.5,below right]{\small{0.42}} (v_5)
			(v_3) edge [out=60, in=-60,looseness=1.25] node[pos=0.5,right]{\small{0.40}} (v_6)
			(v_4) edge[bend right=20] node[pos=0.35,below left]{\small{1.0}} (v_5)
			(v_4) edge[bend right=10] node[pos=0.6,above left]{\small{0.94}} (v_6)
			(v_5)  edge node[pos=0.5,right]{\small{0.85}} (v_6)
		
	;
\coordinate (T') at (-0.5,-0.5);
\path[->] (T) edge [bend right] node[above] {} (T');
 \node at (1.5,-1.5) {$G(V,E)$};
\end{tikzpicture} %
\vspace{-30pt}
    \caption{\it Illustration of the map from the space of functional time series (left) into a graph partitioning problem (right).}%
    \label{fig:example}%
\end{figure}
\setcounter{figure}{2}
This is illustrated in \autoref{fig:example} for the first 6 processes depicted in \autoref{fig:FTSplots}. The clustering problem then becomes equivalent to partitioning the similarity graph into connected components such that points 
with pairwise high weight belong to the same component while points with low weight are put into different components.

\subsection{The spectral embedding} 
The main ingredient to our algorithm is the empirical graph Laplacian
\begin{align}\label{eq:L_T}
\hat{L}=I -\hat{D} ^{-1/2}\hat{W}\hat{D} ^{-1/2},
\end{align}
where $\hat{D}=\text{diag}\big(\{\hat{D}_i\}_{i =1}^{d}\big)$ denotes the {\it degree matrix} of $\hat W$ which carries the degree of vertex $v_i$ as its $i$-th diagonal element, i.e., $\hat{D}_{i}= \sum_{j=1}^{d} \hat{W}_{i,j}$.
This Laplacian can be viewed as a perturbed version of the unknown population Laplacian
\begin{align}\label{eq:L}
L= I - D^{-1/2} W D^{-1/2},
\end{align}
where $W=e^{-\A}$ is the population adjacency matrix and $D$ its corresponding degree matrix. 
Note that we use a 
Normalized Laplacian  as it shows a better 
performance if $d$ is large \citep[][]{vLBB08} and can also be applied to irregular graphs, i.e., graphs of which the vertices have different degrees.

It is straightforward to show that \eqref{eq:L} is symmetric and positive-definite and therefore has an eigendecomposition, say $L=U^{\top}\Lambda U$ with $\Lambda=\text{diag}\big(\{\lambda_i\}_{i=1}^d) \in \mathbb{R}_{\ge 0}^{d \times d}$. In he case where $G=(V,E)$ is an undirected weighted graph with non-negative weights, its spectrum provides information on the connectivity of the graph. More specifically, the multiplicity $k$ of the eigenvalue $0$ of $L$ equals the number of connected components $G_1, \ldots, G_k$ in the graph $G$, and the eigenspace of the eigenvalue $0$ is spanned by the vectors $D^{1/2} \mathbb{1}_{G_i} i = 1,\ldots, k$, where $\mathbb{1}_{G_i} \in \mathbb{R}^d$ denotes the indicator vector on component $G_i$  \citep[see e.g.,][]{Chung1997}. 

To understand the usefulness of these properties, suppose for a moment that our graph has exactly $k$ disconnected components where the nodes belonging to different components are infinitely far apart, i.e. have zero weight. 
The above implies that if we collect the $k$ eigenvectors that belong to the $k$ smallest eigenvectors of $L$
and we row-normalize this matrix, we obtain a matrix of indicator vectors
  \begin{align} \label{eq:popEmb}
  \mathcal{U}:=[\mathbb{1}_{G_1}, \ldots ,\mathbb{1}_{G_k} ] \in \mathbb{R}^{d \times k}.\end{align}
Per row of $\mathcal{U}$ there will be exactly one nonzero element, indicating the component to which it belongs to.
In practice,  one never encounters the ideal situation that the nullspace of the empirical graph Laplacian is perfectly spanned by (scaled) indicator vectors because the empirical similarity graph, by construction, consists of only one connected component.  Nevertheless, the information about the structure of $k$ clusters is still completely contained in the eigenvectors that belong to the smallest $k$ eigenvalues of $L$. 
In particular, these eigenvectors provide a relaxation solution to the {\em Normalized minimal cut} problem (Ncut) \citep{SM02}, which has the objective to partition the graph into $k$ disjoint components by `cutting' it at the edges of which the total sum of normalized weights (relative to the volume of the partitioning component) is minimized.  In order to exploit this information, we embed the infinite-dimensional processes $X_i$ into the space spanned by the $k$ eigenvectors that belong to the $k$ smallest eigenvalues by representing the $i$-th process by the $i$-th row of $\mathcal{U}$. The embedded points then provide a representation of $X_i$ in $\mathbb{R}^{k}$ of which the clustering properties are enhanced. As a result, a simple algorithm such as $k$-means can be applied to the embedded points to identify the clusters.

To be precise, let  ${\Vn}_{\cdot,1}, \ldots ,{\Vn}_{\cdot,k} \in \mathbb{R}^d$ denote the row-normalized $k$ eigenvectors of 
the empirical graph Laplacian $\hat{L}$ defined in \eqref{eq:L_T} corresponding
to the smallest $k$ eigenvalues. 
Note that the eigenvectors  of $\hat L$ can only be identified up to an orthogonal rotation. 
The row normalization  avoids this additional source of unidentifiability as will become clear below. It is by no means guaranteed that the embedding of our data as the rows of the matrix
  \begin{align}\label{eq:empEmb}
 {\Vn}:=[{\Vn}_{\cdot,1}, \ldots ,{\Vn}_{\cdot,k}  ] \in \mathbb{R}^{d \times k},
  \end{align}
provides a meaningful spectral embedding for clustering. For this purpose  we need to show  that 
$\hat{L}$  is `close' to $L$  measured by a suitable norm such that their spectral properties converge to their population counterparts. The approach that we take to establish consistency of the empirical Laplacian as $T$ tends to infinity is 
 based on perturbation theory comparing the matrices  
 $\hat{L}$ and $L$ \citep[see also][]{NJW02,RChY11}.
The proofs are technical and rely on several auxiliary results which are relegated to the Appendix. Using consistency of $\hat{\A}$ and the symmetry of $L$, we can show that the distance in operator norm between $\hat{L}$ and $L$ converges to zero in probability.

\begin{lemma} \label{lem:Lcons}
Under the conditions of \autoref{thm:con}
\[\forall \varepsilon>0,\quad \lim_{T \to \infty}  \mathbb{P}\big(\snorm{\hat{L}-L}_{\infty} > \varepsilon\big)=0.\]
\end{lemma}
To analyze the concentration of $\Vn$, we use a slight modification of the classical Davis-Kahan Sin $\Theta$ theorem \citep{StSun1990}. While this theorem provides an upper bound on the sinus of the principal angles between the eigenspaces $\hat{U}O$,  for some orthonormal rotation matrix $O \in \mathbb{R}^{k \times k}$, and $U$, in terms of the spectral grap $\delta$, the dimension of the space $k$ and on the size of the perturbation $\snorm{H}_{\infty}$,
 \autoref{lem:Rotbound} in the Appendix can be used to bound the Euclidean distance between the non-normalized empirical eigenvectors $\hat{U}$ and their population counterparts $U$ up to rotation. For the row-normalized matrix $\Vn$, we avoid the additional source of unidentifiability caused by the rotation matrix. We derive the following result.

\begin{lemma} \label{cor:conVn}
The matrix $\Vn$ defined in \eqref{eq:empEmb} satisfies
\begin{align*}
\|\Vn-\mathcal{U} \|_2 & \le \frac{ 4\sqrt{k}}{\sqrt{\min_{i}\|{U}_{i,\cdot}\|^2_2}} \frac{\snorm{\hat{L}-L}_\infty}{\lambda_{k+1}} 
\le 4 \sqrt{k} \sqrt{\frac{\mathcal{C_{\text{max}}}}{\min_i D_{i}}} 
\frac{
\snorm{\hat{L}-L}_\infty
}{\lambda_{k+1}} 
\end{align*}
where $\lambda_{k+1}$ is the $(k+1)$-th smallest eigenvalue of $L$ and where $\mathcal{C}_{\text{max}} =\max_{i}\sum_{i_1,i_2 \in G_i}W_{i_1,i_2}$. Hence, Under the conditions of \autoref{lem:Lcons},
\[
\forall \varepsilon>0,\quad \lim_{T \to \infty}  \mathbb{P}\big(\|\Vn-\mathcal{U} \|_2 > \varepsilon\big)=0.
\]
\end{lemma}

These results thus justify to use the rows $\Vn_{1,\cdot}, \ldots, \Vn_{d,\cdot}$ of \eqref{eq:empEmb} to embed the $d$ functional time series.  Each embedded point $\hat{\mathcal{U}}_{i,\cdot}$ then represents a process $X_i$ (or node) in $k$ dimensions, where these $k$ dimensions can be seen as the features. Based on this representation we can cluster our data using $k$-means. This step is analyzed next. We remark that our treatment of the spectral embedding by means of row-normalized eigenvectors of the graph Laplacian is therefore similar to \citet{NJW02} who first investigated the use of a symmetric Laplacian in the context of spectral clustering. 

\begin{exmp}[continued]
For our example, the eigenvalues in \autoref{fig:spectrum}(a) indicate the empirical graph Laplacian has one connected component. This is not surprising as we have a fully connected graph. The first eigenvector is therefore approximately constant, which can be seen in  \autoref{fig:spectrum}(b), while the second and third are approximately constant on the clusters. 
The second eigenvector  (\autoref{fig:spectrum}(c)) allows to separate the first 3 observations from the rest, while the third eigenvector indicates to separate the last three from the first 6 observations (\autoref{fig:spectrum}(d)).
 \begin{figure}[h!]
    \centering
    {\renewcommand{\arraystretch}{0}
    \begin{tabu}{c@{}c}
    \begin{subfigure}[t]{.25\columnwidth}
        \centering
        \includegraphics[width=\columnwidth]{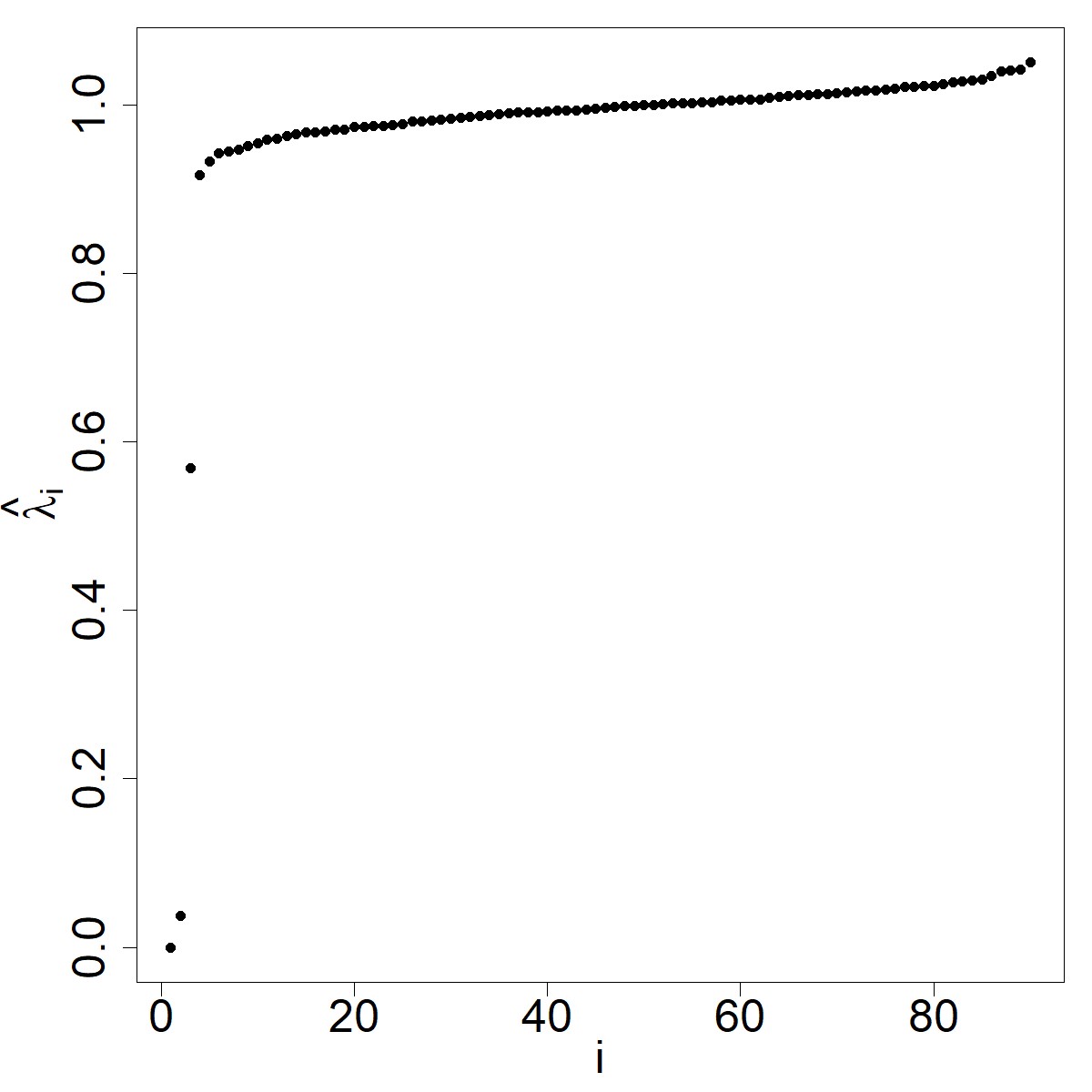}%
          \caption{\it {\small eigenvalues of $\hat{L}$}}
        \label{fig:mean and std of net14}
    \end{subfigure}
    \begin{subfigure}[t]{.25\columnwidth}  
        \centering
        \includegraphics[width=\columnwidth]{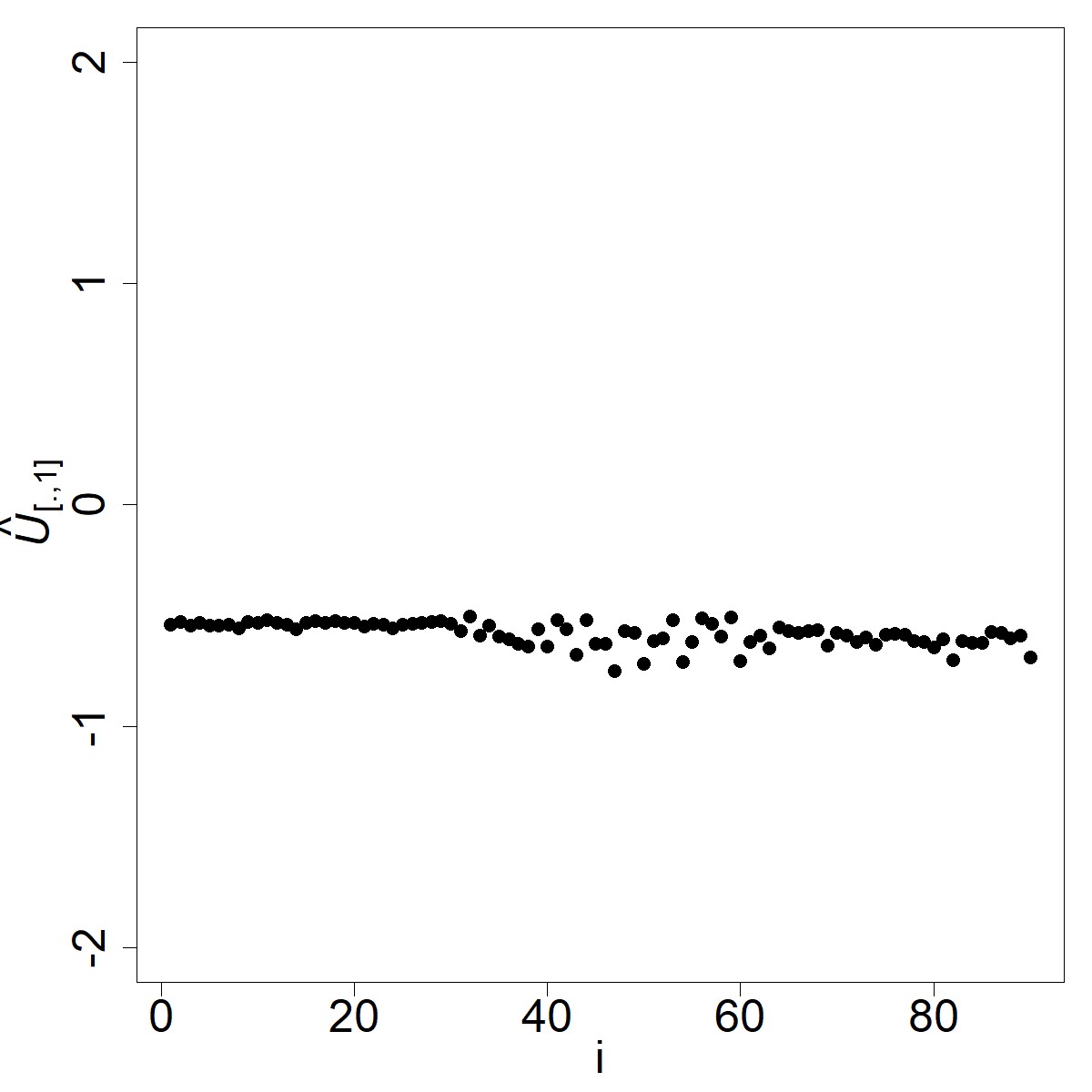}%
          \caption{\it {\small eigenvector $\hat{\mathcal{U}}_{\cdot,1}$}}
        \label{fig:mean and std of net24}
    \end{subfigure}
    \begin{subfigure}[t]{.25\columnwidth}   
        \centering 
        \includegraphics[width=\textwidth]{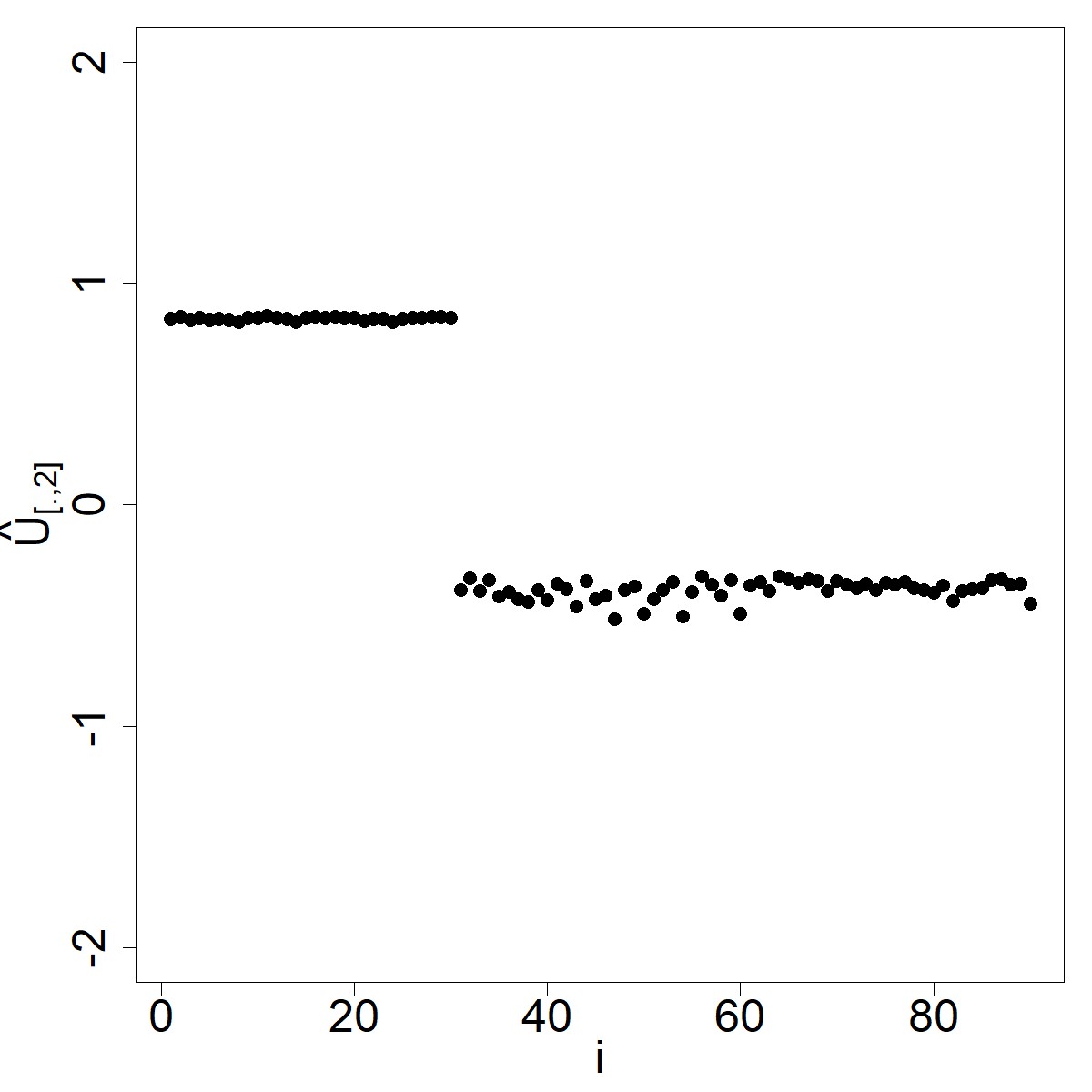}%
        \caption{\it {\small  eigenvector $\hat{\mathcal{U}}_{\cdot,2}$}}
        \label{fig:mean and std of net34}
    \end{subfigure}
    \begin{subfigure}[t]{.25\columnwidth}   
        \centering 
        \includegraphics[width=\columnwidth]{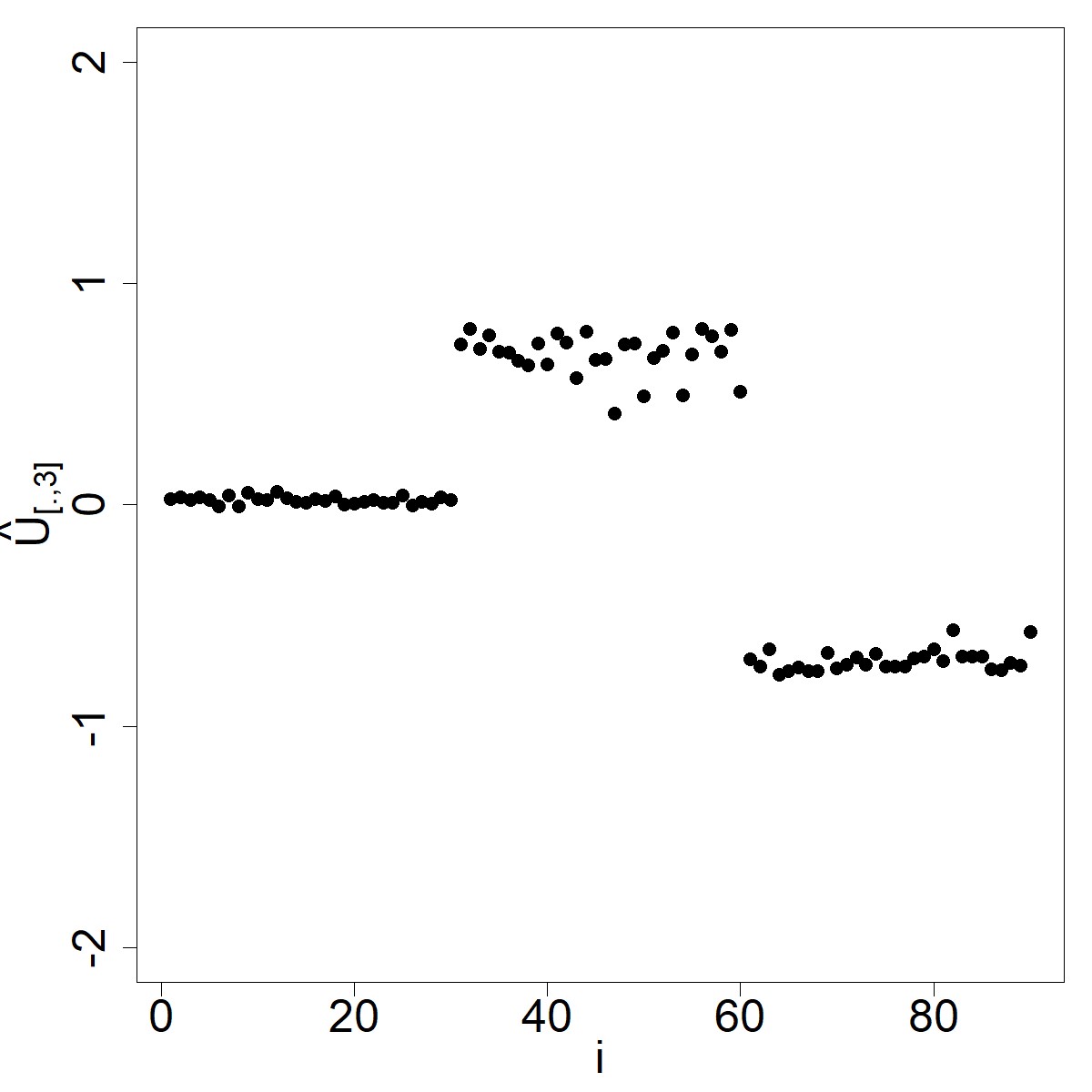}%
        \caption{\it {\small  eigenvector $\hat{\mathcal{U}}_{\cdot,3}$}}
        \label{fig:mean and std of net44}
    \end{subfigure}
    \end{tabu}}
    \caption[\it  The average and standard deviation of critical parameters ]
    {\it \small The spectrum of $\hat{L}$ and the three row-normalized eigenvectors belonging to the eigenvalues that belong to the three smallest eigenvalues $\hat\lambda_1 \le \hat\lambda_2 \le \hat\lambda_3$.} 
    \label{fig:spectrum}
\end{figure}
\end{exmp}

\subsection{Clustering the embedded points using $k$-means}

In this section, we analyze the final step where  $k$-means 
is applied to the embedded points $\hat{\mathcal{U}}_{1,\cdot},\ldots$, $\hat{\mathcal{U}}_{d,\cdot} \in\ \mathbb{R}^{k}$. We show that the $k$-means algorithm, clusters, with high probability, the data correctly. More specifically, we derive a non-asymptotic bound on the number of points that are misclustered and show, under regularity conditions, that this converges to zero as $T \to \infty$. 

The $k$-means objective aims to partition the $d$ embedded points $\hat{\mathcal{U}}_{1,\cdot},\ldots, \hat{\mathcal{U}}_{d,\cdot} \in \mathbb{R}^{k}$ into $k$ clusters $\{C_1,\ldots, C_k\}$ in such a way that the pairwise squared deviations of points within the same cluster is minimized. The algorithm thus returns the centroids $\{c^\star_1,\ldots, c^\star_k\}$ from the objective function,
\[
\min_{\{c_1, \ldots, c_k\} \subset \mathbb{R}^k} \sum_{i} \min_j \|\hat{\mathcal{U}}_{i,\cdot}-c_{j} \|^2_2.
\]
The data points $\hat{\mathcal{U}}_{i,\cdot}$ and $\hat{\mathcal{U}}_{j,\cdot}$ are then put in the same cluster if $c^\star_i = c^\star_j$. More formally but equivalently, for $\Vn$ defined in \eqref{eq:empEmb}, the $k$ means algorihm should return a matrix $C^{\star} \in \mathbb{R}^{d \times k}$ with at most $k$ unique rows such that
\begin{align} \label{eq:kmeansobj}
C^{\star}=\argmin_{C \in \mathcal{M}(d,k)} \|\Vn-C \|^2_2.
\end{align}
where $\mathcal{M}(d,k) =\{M \in \mathbb{R}^{d \times k}: M \text{ has at most $k$ distinct rows}\}$.

To analyze the algorithm, we need to define what we mean with a point being correctly clustered. Let $C^\star$ as in \eqref{eq:kmeansobj}, i.e.,  the matrix returned from applying the $k$-means algorithm to $\Vn$. Intuitively, a point $\hat{\mathcal{U}}_{i,\cdot}$ is correctly clustered if there is no other row of the population matrix $\mathcal{U}$ defined in \eqref{eq:popEmb} that is closer to $C^\star_{i,\cdot}$ than the $i$-th row. Using this intuition, we can provide a meaningful definition of the set of correctly clustered point (\autoref{lem:sigma}) and hence of its complement set. The next theorem gives a nonasymptotic bound on the number of misclustered points.

\begin{thm}\label{thm:miscl}
Assume the graph has $k$ components. Then the cardinality of the set of misclustered points, denoted by $|\Sigma|$, satisfies
\begin{align}\label{miscl}
&|\Sigma|\le  \iota k \frac{1}{\min_{i}\|{U}_{i,\cdot}\|^2_2}\frac{\snorm{\hat{L}-L}^2_\infty}{\lambda^2_{k+1}}  \le \iota k {\frac{ \mathcal{C_{\text{max}}}}{\min_i D_{i}}} \frac{\snorm{\hat{L}-L}^2_\infty}{\lambda^2_{k+1}}.
\end{align}
for some constant $\iota > 0$ and where $\mathcal{C_{\text{max}}}=\max_{i}\sum_{i_1 \in G_i} \sum_{i_2 \in G_i}W_{i_1,i_2}$. If the conditions underlying \autoref{thm:con} hold, then $|\Sigma|$ converges to zero 
in probability as $T \to \infty$. 
\end{thm}

This upper bound implies that the probability of misclustering is affected by various properties of the (data-induced) graph. Firstly, one can see it is an increasing function of the number of clusters $k$.  Secondly, it is a decreasing function of the minimum degree. In particular, the second inequality implies that for an unbalanced graph as measured by the maximal sum of all entries of any of the groups relative to the minimal degree $\min_i D_{i}$ the probability to miscluster has a less tight upper bound. In other words, if the graph contains isolated vertices and or many points that are highly concentrated, we can expect a higher probability to miscluster. Thirdly, it is a decreasing function of the distance between the zero eigenvalues and the first nonzero eigenvalue $\lambda_{k+1}$. Finally, it is affected by how well $\hat{L}$ estimates $L$. 


\begin{exmp}[continued]
For our example, the embedded points are shown in the left graph of \autoref{fig:embedres}. Applying k-means identifies all clusters exactly where the colors indicate the cluster the respective data points belong to.
\begin{figure*}
    \centering
    {\renewcommand{\arraystretch}{0}
    \begin{tabular}{c@{}c}
    \begin{subfigure}[t]{.35\columnwidth}
        \centering
        \includegraphics[width=1.0\columnwidth]{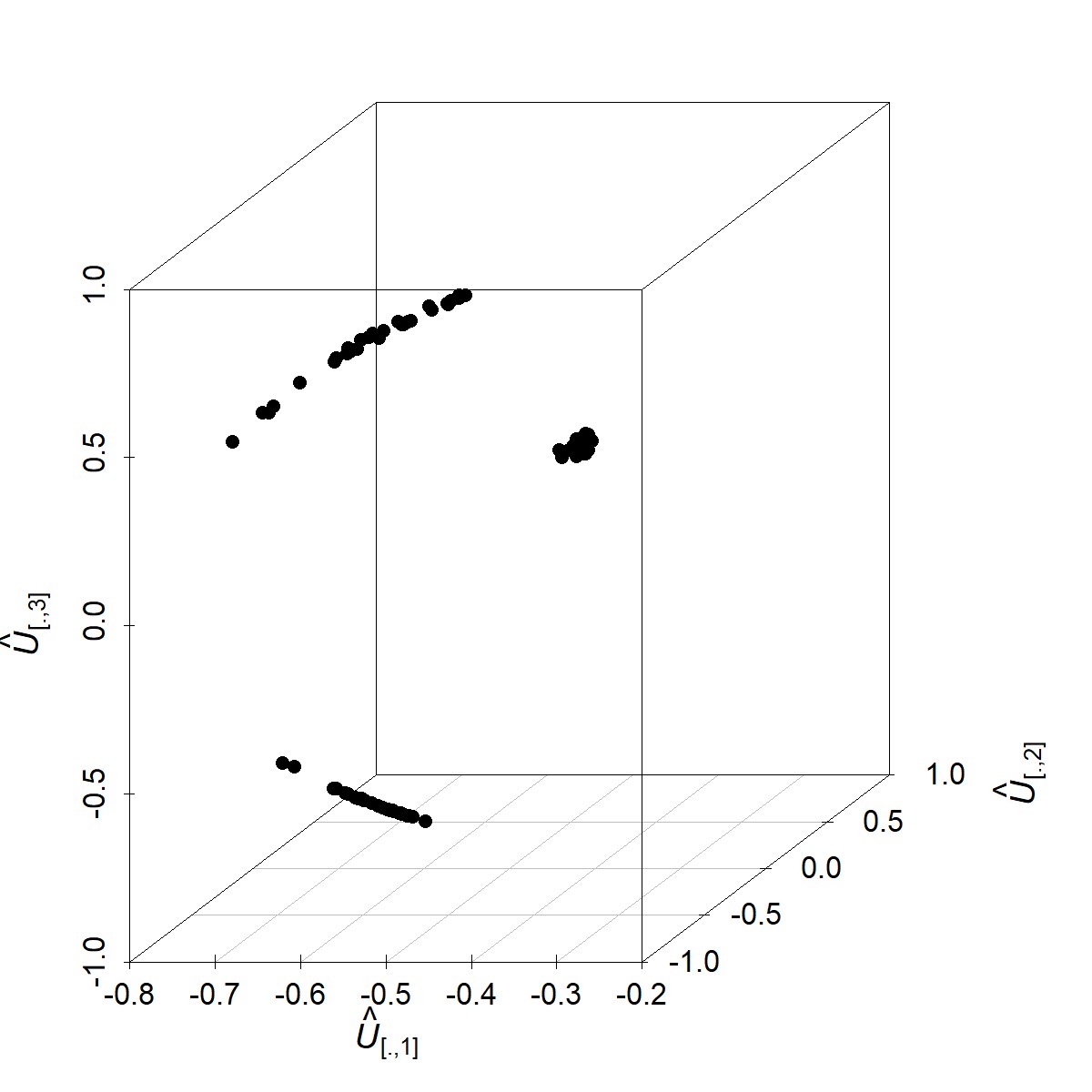}%
    \end{subfigure}&\hfill
    \begin{subfigure}[t]{.35\columnwidth}  
        \centering
        \includegraphics[width=1.0\columnwidth]{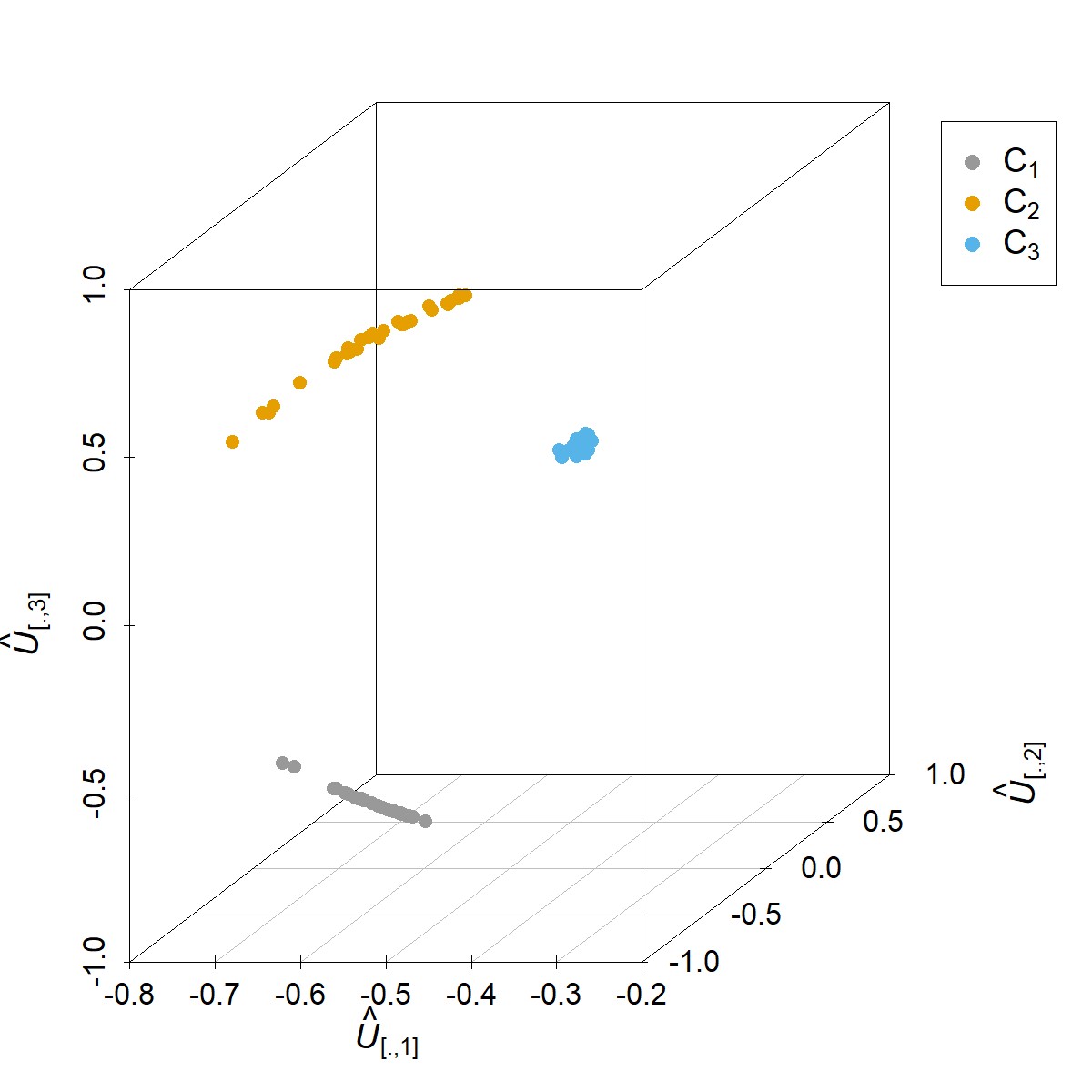}%
    \end{subfigure}
    \end{tabular}}
    \caption[ ]
    {\small \it The embedded points $\hat{\mathcal{U}}_{[1,\cdot]},\ldots, \hat{\mathcal{U}}_{[90,\cdot]}$ (left) and the result of applying $k$-means to these points  where the color indicate the cluster the point belongs to (right). } 
    \label{fig:embedres}
\end{figure*}
\end{exmp}

\subsection{The choice of  $k$} \label{sec34}

So far we have considered the case where the number of clusters $k$ is known. In the remaining part 
of this section we briefly discuss a data driven choice of $k$.
This problem has received much attention within the clustering literature and a wide range of methods has been developed 
 \citep[see][for overviews]{Gordon1999,TK2008} and compared \citep[see][and references therein]{MC1985,TWH2001}. Most  methods to pick $k$ are formulated in terms of optimizing a relative criteria. For different choices of $k$, the quality of the clustering structure is evaluated according to some measure, such as the intra-cluster dispersion, and the number of clusters chosen that optimizes this criteria. 

There is however no universally optimal method because the definition of `optimal' number of clusters can highly depend on the application and on the method used to cluster to data.  If the data is well-separated and there are clear pronounced clusters then there are several successful methods that will correctly detect the underlying clusters. However, in noisy datasets with overlapping of clusters, different methods will detect different number of clusters. In the case of spectral clustering, an intuitive alternative would be to pick the number of clusters $k$ such that the first $\lambda_1,\ldots, \lambda_k$ eigenvalues of $\hat{L}$ are small but $\lambda_{k+1}$ is `relatively' large. This `eigengap' heuristic can be justified through spectral properties of the population graph Laplacian and \autoref{lem:Lcons}. In practice it is however also known to be very sensitive to the construction of the similarity graph and can quickly fail unless the data is well separated. Deemed therefore an unsolvable problem, it is not uncommon to apply multiple criteria and pick the criteria that works best for the particular problem at hand \citep[see][for an implementation of various criteria]{Nbclust14}.  

A thorough development of a new method would be beyond the scope of this paper. In our empirical study in \ref{sec4} we investigate the performance of available methods to choose $k$  for our particular algorithm. Additionally, we consider two variations of the eigengap heuristic each with a different interpretation of a  ``relatively large'' eigengap.

\section{Testing for equality}  \label{sec5}
\def\theequation{5.\arabic{equation}}
\setcounter{equation}{0}

Besides from the application of the similarity  measure as a basis for spectral clustering of functional time series, a well-defined limiting distribution allows it to be meaningful in a variety of statistical applications and in particular for the construction of hypothesis tests. The problem to detect similarities or to compare time series is of interest in a wide range of scientific fields and for classical time series a variety of methods have been proposed (see references in the introduction and examples therein). In case of functional data, the literature is less well-developed. \citet{HKR13} proposed a procedure to test the hypothesis that two sets of functional data are identical and independently distributed using the sum of $L^2$-distances of the sequence of correlation functions. \citet{PanTav2016} instead proposed a test between two stationary functional time series based upon the Hilbert-Schmidt norm of the difference of the sample spectral density operators, restricted to a Hilbert-Schmid space of finite dimension. Bootstrap-based methods to test for equality of mean functions or covariance operators are proposed in \citet{PapSap2016} and, more recently, \citet{LPapSap2018} disucssed a test for the equality of spectral density operators.

To our knowledge, no procedure is available that allows to test for similarity of functional time series of which the second order structure is allowed to be time-dependent. In this section, we develop such a test using the previously defined similarity measure ${\A}_{i_1,i_2}$ in 
\eqref{eq:dist}. For the sake of brevity we restrict ourselves to  the
case of two functional series and consider the hypothesis
 \begin{equation}
H_0: \mathcal{F}^{(i_1)}_{u,\omega} \equiv \mathcal{F}^{(i_2)}_{\omega} \hspace{0.1 in} \text{a.e. on } [-\pi,\pi]\times[0,1] 
\label{h0} 
\end{equation}
\centerline{versus}
 \begin{equation}\label{h1} 
 H_a: \mathcal{F}^{(i_1)}_{u,\omega} \neq \mathcal{F}^{(i_2)}_{\omega} \text{on a subset of } [-\pi,\pi]\times[0,1] \text{ of positive Lebesgue measure}
\end{equation}
The similarity measure in \eqref{eq:dist} lends itself quite naturally to test this hypothesis, i.e., we can equivalently formulate the hypothesis as
\begin{equation} \label{hequiv0}
H_0: {\A}_{i_1,i_2}=0 \hspace{0.3 in} \text{versus} \hspace{0.3 in} H_a: {\A}_{i_1,i_2}> 0.
\end{equation}
By \autoref{thm:con}, the statistic $\hat{\A}_{i_1,i_2}$ 
defined in \eqref{eq:est_dist} 
is a consistent estimator of the normalized 
distance $ {\A}_{i_1,i_2}$. Therefore
it is reasonable to reject the null hypothesis 
for large values of the estimator $\hat{\A}_{i_1,i_2}$. Using the set of reqularity conditions provided in Appendix \autoref{sec:dist} the following result gives the asymptotic distribution of  $\hat{\A}_{i_1,i_2}$.
\begin{thm} \label{thm:AN}
Suppose that \autoref{cumglsp} with $m \ge 1$ and \autoref{ratesNM} hold. Then,
\[
\left\{\sqrt{T}\left(\hat{\A}_{i_1,i_2} - \A_{i_1,i_2}\right)\right\}_{\{i_1,i_2 \in [d]\}} \to \mathcal{N}(\boldsymbol{0},\boldsymbol{\Gamma}) \quad \text{ as } T \to \infty,
\]
where $\boldsymbol{0} \in \mathbb{R}^{d}$ and $\boldsymbol{\Gamma}$ is a positive definite element of $\,\mathbb{R}^{d \times d}$.
\end{thm}

Under the null hypothesis, the asymptotic variance reduces to a very succinct form in case the processes are moreover independent.
\begin{thm}\label{dist_H0}
Suppose that \autoref{cumglsp} with $m \ge 1$ and \autoref{ratesNM} hold and suppose $\{X_{i_1,t,T}\}$ and $\{X_{i_2,t,T}\}$ are independent. Then, under the null hypothesis $H_0: \F^{i_1}_{u,\omega}\equiv \F^{i_2}_{u,\omega}$, we have
\begin{align*}
\sqrt{T}\hat{\A}_{i_1,i_2}
 \sim  \mathcal{N} \big(0, \sigma^2_{H_0}\big),
 \end{align*}
 where the asymptotic variance is given by
\begin{align} \label{eq:sigma0}
\sigma^2_{H_0} = 4\pi \frac{\int_{-\pi}^{\pi} \int_0^{1}\snorm{\F^{i_1}_{u,\omega}}^4_2 du d\omega}{\big(\int_{-\pi}^{\pi} \int_0^{1}\snorm{\F^{i_1}_{u,\omega}}^2_2 du d\omega\big)^2}. 
\end{align}
\end{thm}
Let $I_{p}^{u_j,\omega_k}=( I_{i_1}^{u_j,\omega_k}+ I_{i_2}^{u_j,\omega_{k}})/2$ be the pooled periodogram operator evaluated at $u_j$ and $\omega_k$. The asymptotic variance under the null can be estimated by 
 \begin{equation} 
 \label{varesth0} 
 \hat \sigma_{H_0}^2 = \frac{2}{3 T} \sum_{j=1}^M \sum_{k=1}^{\lfloor N/2 \rfloor}  \big(\innprod{I_{p}^{u_j,\omega_k}}{I_{p}^{u_j,\omega_{k-1}}}_{HS}\big)^2 \Big/
\Big( \frac{2}{T}\sum_{j=1}^M \sum_{k=1}^{\lfloor N/2 \rfloor} \langle I_{p}^{u_j,\omega_k}, I_{p}^{u_j,\omega_{k-1}}\rangle_{HS}\Big)^2.
 \end{equation}
\begin{lemma}\label{lem:est_sig}
Under the conditions of \autoref{dist_H0}, the estimator defined in \eqref{varesth0} satisfies
\[ \hat \sigma_{H_0}^2 \overset{p}{\to} \sigma_{H_0}^2  \]
as $T \to \infty$. 
\end{lemma}
Consequently, a test for the hypotheses \eqref{h0} and \eqref{h1} can be based upon rejecting the null if
\begin{eqnarray} \label{testwn}
\hat{\A}_{i_1,i_2}~> ~  \frac{\hat{\sigma}^2_{H_0}}{\sqrt{T}} z_{1-\alpha}~,
 \end{eqnarray}
 where $z_{1-\alpha}$ denotes the $({1-\alpha})$-quantile of the standard normal distribution. The finite sample performance of this test is studied in a simulation study which is provided at the end of \autoref{sec4}.
  
 \begin{Remark}
 We emphasize that \autoref{dist_H0} still holds in case the series are dependent but the expression of the asymptotic variance is slightly more involved. It can however still be estimated similar in spirit to \eqref{varesth0}. See Appendix \autoref{sec:dist} for more details.
 \end{Remark}
 
\section{Simulation study}  \label{sec4}

\def\theequation{4.\arabic{equation}}
\setcounter{equation}{0}


To study the performance in finite samples, we apply the algorithm to a mixture of stationary and non-stationary models in a variety of settings. We vary the number of clusters $k$ and the number of observations per cluster $n$ as well as the included models. The algorithm is moreover considered both in the case were the number of true clusters are known and when this is unknown. In the latter case, we obtain an additional source of variability as the number of clusters needs to be estimated from the data. Finally, we consider the effect on the simulations of applying a scaling factor $\eta$ to our similarity matrix.   Before we discuss how to determine $k$, we start by introducing the simulation setting and data-generating processes. 

\subsection{Simulation setting}
 We simulate functional autoregressive and functional moving average via their bais representation as follows. A $p$-th order (time varying) functional autoregressive process (tvFAR(p)), $(X_t, t \in \mathbb Z)$ can be defined as
\begin{equation}
\label{tvFAR}
X_{t} = \sum_{t'=1}^p A_{t,t'}(X_{t-t'}) + \epsilon_t, 
\end{equation}
where $A_{t,1}, \dots, A_{t,p}$ are time-varying autoregressive operators  on $L^2([0,1])$ and $\{\epsilon_t(\tau)\}_{t\in \mathbb{Z}}$ is a sequence of mean zero innovations taking values in $L^2([0,1])$. To generate such processes, let $\{\psi_l\}_{l=1}^{\infty}$ be a Fourier basis of $L^2([0,1])$. By means of a basis expansion, one can show that  \citep[see e.g.,][]{avd16_main}  the first $L_{max}$ coefficients of \eqref{tvFAR} are generated using the $p$-th order vector autoregressive, VAR(p), process
$$\widetilde{X}_t = \sum_{t'=1}^p \widetilde{A}_{t,t'}\widetilde{X}_{t-t'} + \widetilde{\epsilon}_t,$$
where $\widetilde{X}_t := \left(\langle X_t, \psi_1 \rangle, \dots, \langle X_t, \psi_{L_{max}} \rangle \right)^\top$ is the vector of basis coefficients and the $(l,l')$-th entry of $\widetilde{A}_{t,j}$ is given by $\langle A_{t,j}(\psi_l),\psi_{l'}\rangle$ and $\widetilde{\epsilon}_t := \left(\langle \epsilon_t, \psi_1 \rangle, \dots, \langle \epsilon_t, \psi_{L_{max}} \rangle \right)^T$. The entries of the matrix $\widetilde{A}_{t,j}$ are generated as $N\big (0, \nu_{l,l'}^{(t,j)}\big )$ with $\nu_{l,l'}^{(t,j)}$ specified below. To ensure stationarity or existence of a causal solution the norms $\kappa_{t,j}$ of $A_{t,j}$ are required to satisfy certain conditions (see \cite{bosq2000} for stationary and \cite{vde16} for  local stationary time series, respectively). We also consider the following time varying functional moving-average process or order $1$:
\begin{equation}
\label{tvFMA}
X_{t,T}= B_1(\epsilon_t)- \frac{1}{2}\left(1 + b \cos\left(2\pi \frac{t}{T}\right)\right)B_2(\epsilon_{t-1})
\end{equation}
where $B_1$ and $B_2$ are bounded linear operators on $L^2([0,1])$ and $b \in \mathbb{R}$. Similarly as above, we use a basis expansion and generate data from the model
$$\widetilde{X}_{t,T} = \widetilde{B}_1\widetilde{\epsilon}_t - \frac{1}{2}\left(1 + b \cos\left(2\pi \frac{t}{T}\right)\right)\widetilde{B}_2\widetilde{\epsilon}_{t-1}$$
where $\widetilde{X}_{t,T} = \left(\langle X_{t,T}, \psi_1 \rangle, \dots, \langle X_{t,T}, \psi_{L_{max}} \rangle \right)^T$ is the vector of basis coefficients, the $(l,l')$-th entry of $\widetilde{B}_1$ and $\widetilde{B}_2$ are given by $\langle B_1(\psi_l),\psi_{l'}\rangle$ and $\langle B_2(\psi_l),\psi_{l'}\rangle$ respectively and $\widetilde{\epsilon}_t$ is as above.
 
we consider the following data generating processes:
\begin{itemize}
\item [(I)] The functional white noise variables $\epsilon_1, \dots, \epsilon_T$ i.i.d. with coefficient variances $\text{Var}(\langle \epsilon_t, \psi_l \rangle) = \exp(-(l-1)/10)$.
\item[(II)] A FAR(2) variables $X_1, \dots, X_T$ with operators specified by variances $\nu_{l,l'}^{(1)} = \exp(-l-l')$ and $\nu_{l,l'}^{(2)} = 1/(l+l'^{3/2})$ with norms $\kappa_1 = 0.75$ and $\kappa_2 = -0.4$ and the innovations  $\epsilon_1, \dots, \epsilon_T$ are as in (I).
\item[(III)] A MA(1) with $b=0$ and operators specified by variances $\nu_{l,l'}^{(1)} = \exp(-l-l')$.
\item [(IV)] A tvFAR(1) with operator specified by variances $\nu_{l,l'}^{(t,1)} = \nu_{l,l'}^{(1)} = \exp(-l-l')$ and norm $\kappa_1 = 0.8$, and innovations are as in (I) but with a multiplicative time-varying variance
$$\sigma^2(t) = \cos\Big(\frac{1}{2} + \cos\big(\frac{2\pi t}{T}\big) + 0.3 \sin\big(\frac{2\pi t}{T}\big)\Big).$$
\item[(V)] A tvFAR(2) with operators as in (IV), but with time-varying norm
$$\kappa_{1,t}=1.8 \cos\left( 1.5 - \cos\left(\frac{4\pi t}{T}\right)\right)$$
and constant norm $\kappa_2 = -0.81$ and innovations are as in (I).
\item[(VI)] A FAR(2) with structural break;
\begin{itemize}
\item for $t \leq 3T/8$, the operators are as in (II) with norms $\kappa_1 = 0.7$ and $\kappa_2 = 0.2$, with innovations as in (I).
\item for $t > 3T/8$, the operators are as in (II) with norms $\kappa_1 = 0$ and $\kappa_2 = -0.2$, with innovations as in (I) but with  coefficient variances $\text{Var}(\langle \epsilon_t, \psi_l \rangle) = 2\exp((l-1)/10)$.
\end{itemize}
\end{itemize}


Our simulation consists of the following settings:
\begin{enumerate}
\item[]{\it Setting 1:} $k=3$ with models I, II and III, where the replications per cluster are taken $n=10$ and $n=30$
\item[]{\it Setting 2:} $k=3$ with models IV, V and VI, where the replications per cluster are taken $n=10$ and $n=30$
\item[]{\it Setting 3:} $k=6$ with models I-VI, where the replications per cluster are taken $n=10$ and $n=30$ and $n=50$.
\end{enumerate}
Per set-up, we run 500 simulations for both $T=256$ wit $M=8$ and $T=512$ with  $M=16$. 

For the choice of $k$  we investigated  a subset of well-known classical criteria that have been demonstrated to work well in aforementioned comparison studies on classical clustering: the Silhouette Index \citep{KR1990}, the CH Index\citep{CH1974},  the Hartigan Index \citep{Hart1975} and the KL Index \citep{KL1988}. We respectively refer to these in the tables as `sil' ,`ch', `hartigan' and `kl'. Because these cannot be applied to the spectral embedding directly, the respective criteria were constructed using the similarity graph $\hat{W}$. 

Additionally, we considered two variations of the eigengap heuristic each with a different interpretation of `relatively large' eigengap.
Let $0=\hat \lambda_1 \le \ldots \le \hat{\lambda}_d$ be the estimated eigenvalues of $\hat L$ in ascending order
\begin{enumerate}
\item `Relgap':  define the relative contribution of the $k$-th eigenvalue as
\[
\rho_k=\frac{\big(\hat{\lambda}_{k} - \hat{\lambda}_{k-1}\big)}{\hat{\lambda}_k}.
\]
Then the rule is to pick $k^\star$ as the largest $k$ for which the $k$-th contribution is still larger than a threshold value that is allowed to depend on the scaling parameter $\eta$ of the graph (see \eqref{eq:Weta} below).
\[
k^\star =\max_k \big\{k \in \{1,\ldots, k_{\text{max}}\}: \rho_k  \le 0.01\eta \big\}
\]
\item `sd1gap': let
\[
\sigma(\hat{\lambda}_{[-1:k]}) =\frac{1}{d-k}\sum_{j =k+1}^{d} \big(\hat{\lambda}_j -\frac{1}{d-k}\sum_{j=k+1}^d \hat{\lambda}_j\big)^2
\]
be the squared deviation from the mean excluding the first $k$ eigenvalues. Then the rule is to pick $k^\star$ as the largest $k$ for which the $k$-th gap is still larger than $\sigma(\hat{\lambda}_{[-1:k]}) $, i.e.,
\[
k^\star =\max_k \big\{k \in \{1,\ldots, k_{\text{max}}\}: \hat \lambda_{k+1} -\hat{\lambda}_k \ge \sigma(\hat{\lambda}_{[-1:k]})  \big\}
\]
\end{enumerate}
For all criteria, the maximum numbers of clusters to consider was set to $k_{\max}=15$.
\subsection{Simulation results}

 \autoref{Simk} provides the average $k$ chosen according to the different criteria in each setting, while the corresponding percentage of misclustered points averaged over simulations are given in \autoref{Sim}. 

From the first row for $T=256$ and $T=512$ of \autoref{Sim}, we find that our algorithm does very well if the true $k$ is known as it has a very low percentage of misclustered points across the different settings. Based upon the percentage of misclustered points, the most difficult model is clearly setting 2 with $T=256$. This finding is in accordance with the upper bound in \eqref{miscl}, which is less tight for larger $k$ and lower estimation precision on $\hat{L}$. 

 If the true $k$ is not known we obtain higher percentages of misclustered points, where the percentages appear to be directly caused by the way $k$ was selected. As can be seen from \autoref{Simk}, the CH Index does best in determining the true number of clusters, while the KL index does worst when the number of true clusters increases. Both the Silhouette Index and the Hartigan Index tend to pick $k$ more conservatively. The two rules based upon the estimated eigenvalues of the graph Laplacian --`Relgap' and `sd1gap'-- appear competitive with the CH index, except in setting 2 for $n=10$. What is in particular clear is the overall improvement as both $n$ and $T$ increase.  It appears that in particular the eigenvalue-based methods suffer from more variation when $n$ and $T$ are small. This is intuitive, since the choice of $k$ directly depends upon the estimation precision of the spectral properties, which can be expected to be more sensitive to estimation error for small $T$ and $n$  (see also \autoref{thm:miscl}) . 
From the results in table \autoref{Sim}, we find the CH index to work best in combination with our algorithm. It is most stable across the different settings and has the lowest percentage of misclustered points, which appears to be a direct consequence of the fact that this index estimated the true number of clusters best and that our algorithm suffers from low error conditionally upon knowing the correct number of clusters. 

\begin{table*}[h!]
\begin{center}
\caption{\it chosen k per method averaged over simulations (standard deviation in brackets)} 
\label{Simk} 
\begin{tabu}{$l|^l^l ^l^l ^l^l ^l}
\hline
&  \multicolumn{2}{c} {Setting 1} & \multicolumn{2}{c} {Setting 2} & \multicolumn{3}{c} {Setting 3}\\
\cmidrule(lr){2-3} \cmidrule(lr){4-5} \cmidrule(lr){6-8}
method & $n=10$ &$n=30$  &$n=10$ &$n=30$  &$n=10$ &$n=30$ & $n=50$\\
\hline
\multicolumn{8}{c} {$T=256$}\\
\hline
true $k$ & 3  & 3 & 3   & 3  & 6  & 6 &6 \\ 
\hline
sil & 2  (0) & 2  (0) & 3  (0.1) & 3  (0) & 5.34  (0.8) & 5.24  (0.5)& 5.17 (0.4)  \\ 
ch & 3  (0.5) & 2.96  (0.2) & 3.05  (0.2) & 3  (0) & 6.37  (0.6) & 6.08  (0.3)  & 6.03 (0.2)\\ 
  kl & 2.19  (1.4) & 2.04  (0.7) & 3.09  (1) & 3.22  (1.4) & 10.4  (3.1) & 11.13  (3.5) & 10.79 (3.6)\\ 
  hart & 3.04  (0.2) & 3  (0) & 3  (0) & 3  (0) & 4.86  (0.9) & 4.87  (0.7)& 4.87 (0.6)  \\ 
  Relgap & 3.04 (0.2) & 3 (0) & 3.92 (1.1) & 3.09 (0.3) & 5.14 (0.4) & 5.15 (0.4) &5.19 (0.4)\\ 
 sd1gap & 3.09 (0.3) & 3.03 (0.2) & 3.83 (1.2) & 3.55 (1) & 5.65 (0.9) & 6.23 (0.7) & 6.41 (0.9)\\ 

\hline
\multicolumn{8}{c} {$T=512$}\\
\hline
 true $k$ & 3 & 3 & 3 & 3 & 6 & 6  &6\\ 
\hline
sil & 2  (0) & 2  (0) & 3  (0) & 3  (0) & 5.92  (0.4) & 5.95  (0.2) & 5.98 (0.1) \\ 
ch & 3.13  (0.4) & 3  (0) & 3  (0) & 3  (0) & 6.52  (0.7) & 6.16  (0.4)& 6.08 (0.3) \\ 
  kl & 2.44  (1.7) & 2.18  (1) & 3.05  (0.7) & 3.43  (2.1) & 9.87  (3.3) & 10.45  (3.5) & 9.47 (3.5)\\ 
  hart & 3  (0) & 3  (0) & 3  (0) & 3  (0) & 5.28  (0.7) & 5.4  (0.6) & 5.43 (0.5)\\ 
Relgap & 3 (0) & 3 (0) & 3.57 (0.8) & 3.03 (0.2) & 5.53 (0.6) & 5.88 (0.3)& 5.99 (0.1) \\ 
  sd1gap & 3.09 (0.4) & 3.01 (0.1) & 3.92 (1.2) & 3.4 (0.8) & 6.38 (0.9) & 6.3 (0.6) & 6.2 (0.5)\\ 

\hline
\end{tabu}
\end{center}
\end{table*}

\begin{table*}[h!]
\begin{center}
\caption{\it percentage of misclustered points of the spectral clustering algorithm averaged over simulations (standard deviation in brackets} 
\label{Sim} 
\hspace*{-20pt}
\begin{tabu}{$l|^l^l ^l^l ^l^l ^l}
\hline
&  \multicolumn{2}{c} {Setting 1} & \multicolumn{2}{c} {Setting 2} & \multicolumn{3}{c} {Setting 3}\\
\cmidrule(lr){2-3} \cmidrule(lr){4-5} \cmidrule(lr){6-8}
method & $n=10$ &$n=30$  &$n=10$ &$n=30$  &$n=10$ &$n=30$ & $n=50$ \\
\hline
\multicolumn{8}{c} {$T=256$}\\
\hline
  true & 0.11 (0.6) & 0.05 (0.2) & 0.01 (0.1) & 0 (0) & 3.22 (6.6) & 0.3 (1.1) & 0.13 (0.2)  \\ 
\hline
sil & 33.33 (0) & 33.33 (0) & 0.03 (0.4) & 0 (0) & 13.16 (10.2) & 12.97 (8.2)& 13.87 (6.9)  \\ 
 ch & 4.65 (10.4) & 1.41 (6.6) & 0.49 (2.6) & 0 (0) & 5.11 (6.8) & 0.82 (2.2)& 0.31 (1.1)  \\ 
  kl & 33.3 (6.4) & 33.48 (2.4) & 0.55 (6) & 1.45 (9.4) & 25.6 (14.7) & 25.93 (15.6)& 22.59 (14.4) \\ 
  hart & 0.61 (2.6) & 0.05 (0.2) & 0.01 (0.1) & 0 (0) & 19.76 (14.6) & 18.93 (12.5)& 19.02 (10.7) \\ 
  Relgap & 0.57 (2.4) & 0.05 (0.2) & 10.78 (11.9) & 1.27 (4.2) & 15.75 (5.6) & 14.49 (6.3) & 13.68 (6.7)\\ 
  sd1gap & 1.33 (4) & 0.46 (2.3) & 9.11 (12.5) & 6.55 (10.9) & 12.33 (7.5) & 3.02 (5.5) & 2.79 (5.2)\\ 
\hline
\multicolumn{8}{c} {$T=512$}\\
\hline
   true & 0 (0) & 0 (0) & 0 (0) & 0 (0) & 0.23 (1.7) & 0.01 (0.1) & 0 (0)\\  
\hline
sil & 33.33 (0) & 33.33 (0) & 0.03 (0.4) & 0 (0) & 2.68 (5.8) & 0.77 (3.5) & 0.37 (2.4)  \\ 
  ch & 1.67 (4.6) & 0 (0) & 0.44 (2.1) & 0 (0) & 3.42 (4.4) & 1.14 (2.6)  & 0.61 (2)  \\ 
  kl & 27.97 (14.5) & 30.63 (10.2) & 1.68 (10.3) & 2.59 (12.5) & 20.93 (16.8) & 21.37 (16.3) & 15.56 (14.4) \\ 
  hart & 0 (0) & 0 (0) & 0 (0) & 0 (0) & 12.25 (12.1) & 9.98 (9.4) & 9.57 (8.7)\\ 
 Relgap & 0.02 (0.4) & 0 (0) & 6.92 (9.7) & 0.39 (2.2) & 8.73 (8.4) & 2 (5.4) & 0.17 (1.7) \\ 
  sd1gap & 1.13 (4.4) & 0.14 (1.4) & 10.16 (12.6) & 4.99 (9.7) & 4.15 (6) & 1.92 (3.9) & 1.42 (3.4)\\ 
\hline
\end{tabu}
\end{center}
\end{table*}

As explained in \autoref{sec2}, we apply a local scaling to avoid sensitivity on a global scaling parameter. To investigate the robustness with respect to scaling, we consider our clustering algorithm as explained in \autoref{sec3} but with an additional scaling parameter $\eta$ in the construction of the adjacency matrix. That is, we consider the simulations but with
\begin{align} \label{eq:Weta}
\hat{W}=e^{-\eta \hat{\A}} \in  \mathbb{R}^{d \times d}.
\end{align}
where $\eta=\{0.5,2.5,5,10\}$. The results in \autoref{Sim} correspond to the choice $\eta=1$ and in \autoref{Simset1} we present the four alternative choices. Because of space constraints, we only report them for the specification $T=512$ and $n=30$. The results for the other cases show a very similar picture and are available upon request. It can be observed from the first row for each of the settings that the outcomes are fairly similar when $k$ is known. Variation again therefore seems mostly caused by the way $k$ is chosen, where we find the Silhouette index is sensitive as well as the KL index, but also the eigengap heuristics show some sensitivity for $\eta=10$.  Overall, we may conclude that our method seems capable to detect the correct number of clusters, despite the highly complicated nature of the data. The numerical study moreover suggests that the CH index should be used to find the numbers of clusters if these are unknown.
\begin{table*}[h!]
\begin{center}
\caption{\it Percentage of misclustered points and specified $k$ for various values of $\eta$ averaged over simulations (standard deviation in brackets) for $T=512$, $n=30$} 
\label{Simset1} 
\hspace*{-20pt}\small{
\begin{tabu}{l|^l^l^l^l @{\hspace{1cm}}| ^l^l^l^l^l}
\hline
meth.& \multicolumn{4}{c}{\% of misclustered points}  & \multicolumn{4}{c}{average chosen $k$}\\
& $\eta=0.5$ & $\eta=2.5$ & $\eta=5$  &  $\eta=10$  &  $\eta=0.5$ & $\eta=2.5$ & $\eta=5$  &  $\eta=10$   \\ 
\cmidrule(lr){2-5} \cmidrule(lr){6-9}
   \multicolumn{9}{c} {Setting 1: $T=512, n=30$}\\
     \hline
 true & 0(0) & 0(0) & 0(0) & 0(0) & 3(0) & 3(0) & 3(0) & 3(0) \\ 
\hline
  sil & 33.33(0) & 33.33(0) & 2.13(8.2) & 0(0) & 2(0) & 2(0) & 2(0) & 2.94(0.2) \\ 
  ch & 0.03(0.6) & 0(0) & 0(0) & 0(0) & 3(0) & 3(0) & 3(0) & 3(0) \\ 
  kl & 33.02(5.1) & 2.8(9.3) & 0(0) & 0.27(4.3) & 2.08(0.9) & 2.18(1) & 2.92(0.3) & 3(0) \\ 
  hart & 0(0) & 0(0) & 0(0) & 0(0) & 3(0) & 3(0) & 3(0) & 3(0) \\ 
  Relgap & 0(0) & 0.02(0.5) & 0.02(0.5) & 0.03(0.5) & 3(0) & 3(0) & 3(0) & 3(0) \\ 
  sd1gap & 0.05(0.8) & 0.88(3.5) & 1.08(3.9) & 0.6(2.7) & 3(0.1) & 3.01(0.1) & 3.07(0.3) & 3.09(0.3) \\ 
  \hline
    \multicolumn{9}{c} {Setting 2: $T=512, n=30$}\\ \hline
   true & 0(0) & 0(0) & 0(0) & 0(0) & 3(0) & 3(0) & 3(0) & 3(0) \\ 
  \hline
  sil & 0(0) & 0(0) & 0(0) & 0(0) & 3(0) & 3(0) & 3(0) & 3(0) \\ 
  ch & 0(0) & 0(0) & 0(0) & 0(0) & 3(0) & 3(0) & 3(0) & 3(0) \\ 
  kl & 0.35(4.6) & 2.44(11.2) & 1.58(9.4) & 0.36(4.6) & 3.05(0.7) & 3.43(2.1) & 3.43(2) & 3.29(1.7) \\ 
  hart & 0(0) & 0(0) & 0(0) & 0(0) & 3(0) & 3(0) & 3(0) & 3(0) \\ 
  Relgap & 0.1(1.2) & 2.28(5.8) & 6.66(9.3) & 7.14(8.5) & 3.01(0.1) & 3.03(0.2) & 3.18(0.5) & 3.59(0.9) \\ 
  sd1gap & 4.75(9.5) & 5.48(9.7) & 4.58(8.7) & 2.89(5.9) & 3.38(0.8) & 3.4(0.8) & 3.45(0.9) & 3.42(0.8) \\ 
  \hline
    \multicolumn{9}{c} {Setting 3: $T=512, n=30$}\\ \hline
  true & 0.01(0.1) & 0(0) & 0.04(0.9) & 0.06(0.9) & 6(0) & 6(0) & 6(0) & 6(0) \\ 
  \hline
  sil & 1.03(3.9) & 0.03(0.7) & 0.04(0.9) & 0.06(0.9) & 5.95(0.2) & 5.95(0.2) & 6(0) & 6(0) \\ 
  ch & 1.65(3.2) & 0.29(1.4) & 0.04(0.9) & 0.06(0.9) & 6.23(0.5) & 6.16(0.4) & 6.04(0.2) & 6(0) \\ 
  kl & 22.72(16.4) & 13.94(15.5) & 7.4(12.2) & 4.56(10.5) & 10.54(3.4) & 10.45(3.5) & 9.2(3.6) & 7.84(3.1) \\ 
  hart & 11.48(8) & 7.17(13.5) & 0.04(0.9) & 0.06(0.9) & 5.31(0.5) & 5.4(0.6) & 5.57(0.8) & 6(0) \\ 
  Relgap & 7.17(8.3) & 0.57(2.2) & 3.18(3.9) & 3.93(4.4) & 5.57(0.5) & 5.88(0.3) & 6.06(0.3) & 6.57(0.8) \\ 
  sd1gap & 1.39(3.4) & 4.5(5.1) & 6.02(5.6) & 4.18(5.1) & 6.21(0.5) & 6.3(0.6) & 6.77(1) & 7.2(1.3) \\ 
  \hline
\end{tabu}}
\end{center}
\end{table*}

\subsection{Testing for equality}
We conclude this section with a small investigation of the proposed asymptotic $\alpha$-level test in \eqref{testwn} for the hypothesis of equality of  (possibly) time-varying spectral density operators. To investigate the finite samples properties of the test, we performed a simulation study which includes the previously defined stationary models I and II and nonstationary models V and VI with parameter specification $T=512$ and $M=16$. The pairwise rejection probabilities at the $5\%$ and $10\%$ over $1000$ replications are provided in \autoref{tab:testeqrp}, where the diagonal elements correspond to the null hypothesis. We observe a good 
approximation of the nominal level with model II a bit undersized. The off-diagonal shows good power overall, with model II and model I appearing to be more difficult to distinguish than the other choices. Regardless of the second-order properties being time-varying or not, it appears therefore that the quantiles of the normal distribution are well-captured for the various models if $H_0$ is true, while good power is observed under $H_A$. 
\begin{table}[ht]
\centering
\begin{tabular}{rrrrr}
  \hline
 & I & II & V & VI \\ 
  \hline
I & 10.8 & 96.6 & 100.0 & 100.0 \\ 
II  & 96.6 & 8.0 & 99.9 & 100.0 \\ 
 V  & 100.0 & 99.9 & 9.8 & 100.0 \\ 
 VI & 100.0 & 100.0 & 100.0 & 10.0 \\ 
   \hline
\end{tabular}
\hspace{30pt}
\begin{tabular}{rrrrr}
  \hline
 & I & II & V & VI \\ 
  \hline
I & 5.4 & 93.1 & 100.0 & 100.0 \\ 
II & 93.1 & 3.2 & 99.7 & 100.0 \\ 
V & 100.0 & 99.7 & 3.6 & 100.0 \\ 
VI & 100.0 & 100.0 & 100.0 & 4.7 \\ 
   \hline
\end{tabular}
\caption{\it Rejection probabilities of the pairwise equality test \eqref{testwn} at the 10\% (left); and 5\% level.}
\label{tab:testeqrp}
\end{table}

\FloatBarrier

\section{Data Application}\label{sec6}

\def\theequation{6.\arabic{equation}}
\setcounter{equation}{0}
\setcounter{figure}{0}

We illustrate our clustering algorithm by means of an application to high resolution infrared surface temperature measurements recorded during the years 2015-2018 in the conterminous U.S., Alaska and Hawaii. The exact locations of the 126 included stations in our study are depicted in \autoref{fig:surftemp}. 
\begin{figure}[h!]
    \vspace{-35pt}
    \centering
         \includegraphics[width=0.65\textwidth]{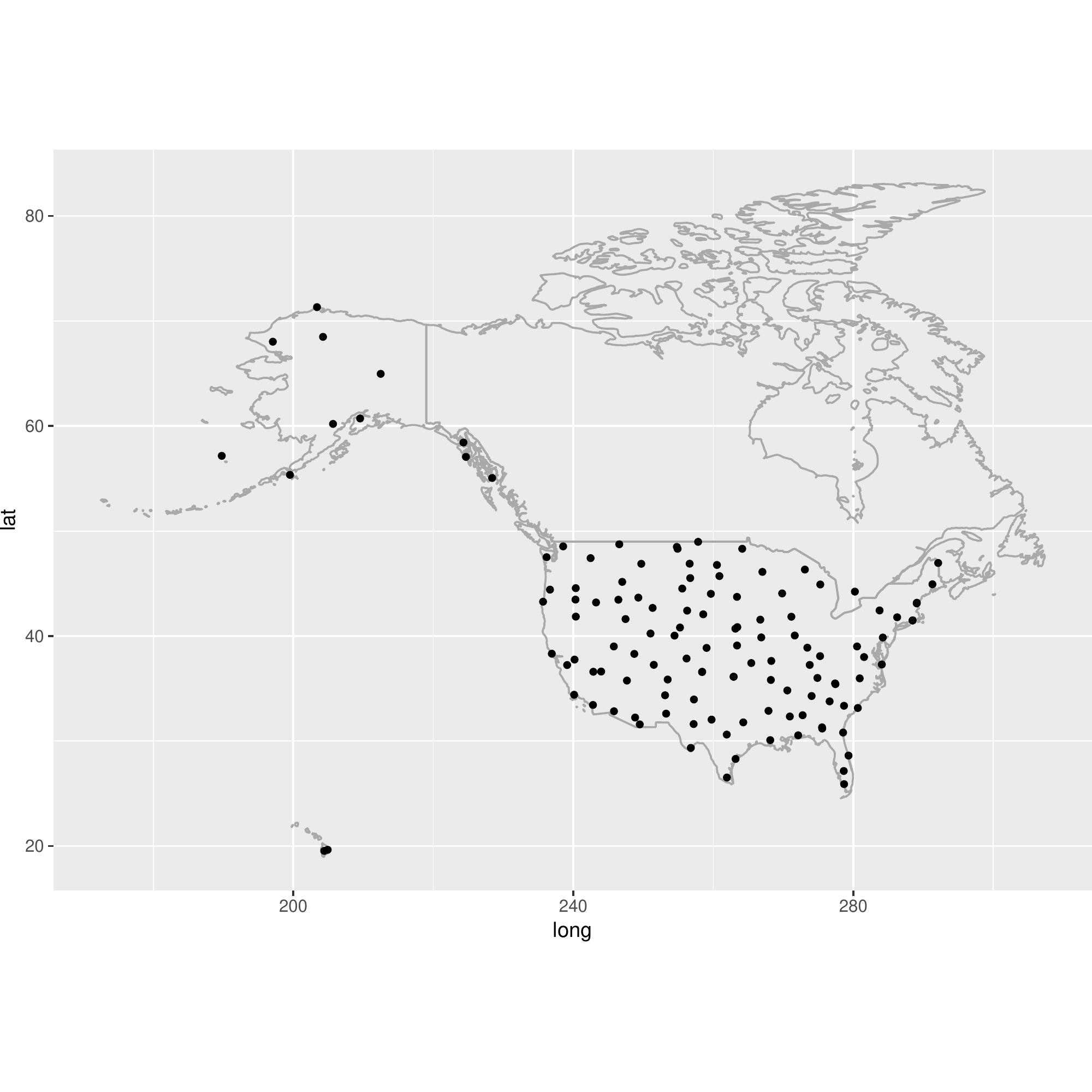}%
    \vspace{-35pt}
   \caption[]{\it  \small Measurement stations projected upon the map of North America acording to their longitude (long) and latitude (lat) coordinates.}. 
    \label{fig:surftemp}
\end{figure}

The measurements are part of a quality controlled dataset \citep{Diam13} monitored within the United States Climate Reference Network (USCRN). These are publicly available on the website of the NOAA U.S. government agency. \\
Understanding the characteristics of land surface temperature is of interest in various applications as it influences for example surface energy balance and various soil hydrologic processes. Our goal is to investigate whether we can identify spatial or geographical patterns in terms of variation in the surface temperature. Variation in surface temperature can be expected to be influenced by numerous factors such as wind, elevation, proximity of major water bodies or yet land cover types. 
For each location $i=1,\ldots, 126$, we have of a sample of observations $\{X_{i,t}(\tau_{1}),\ldots, X_{i,t}(\tau_{\text{max}}) \}_{t=1,\ldots, 1288}$, where $X_{i,t}$ represents the sampled temperature curve on day $t$ for location $i$. The first measurement day $t=1$ corresponds to February 28 2015 and the last measurement day $t=T=1288$ corresponds to September 7 2018. The recordings were made at a 5 minute time-interval leading to a maximum of $288$ observations per day. We remark that the measurement stations included have at most 10\% of missing observations per day over the measurement period. The gridded observations were transformed into functional objects using 21 Fourier basis functions, rescaled to the unit interval, i.e., $\tau_1$ corresponds to local time 0.00am and $\tau_{288}=\tau_{max}$ to local time 23.55pm. A few sets of realized curves are depicted in \autoref{fig:FTStemp}. 
   \setcounter{figure}{2}
 \begin{figure}[h!]
 \hspace*{-25pt}
    \centering
    {\renewcommand{\arraystretch}{0}
    \begin{tabu}{c@{}c}
   \begin{subfigure}[t]{0.35\columnwidth}
        \centering
        \includegraphics[width=1\textwidth]{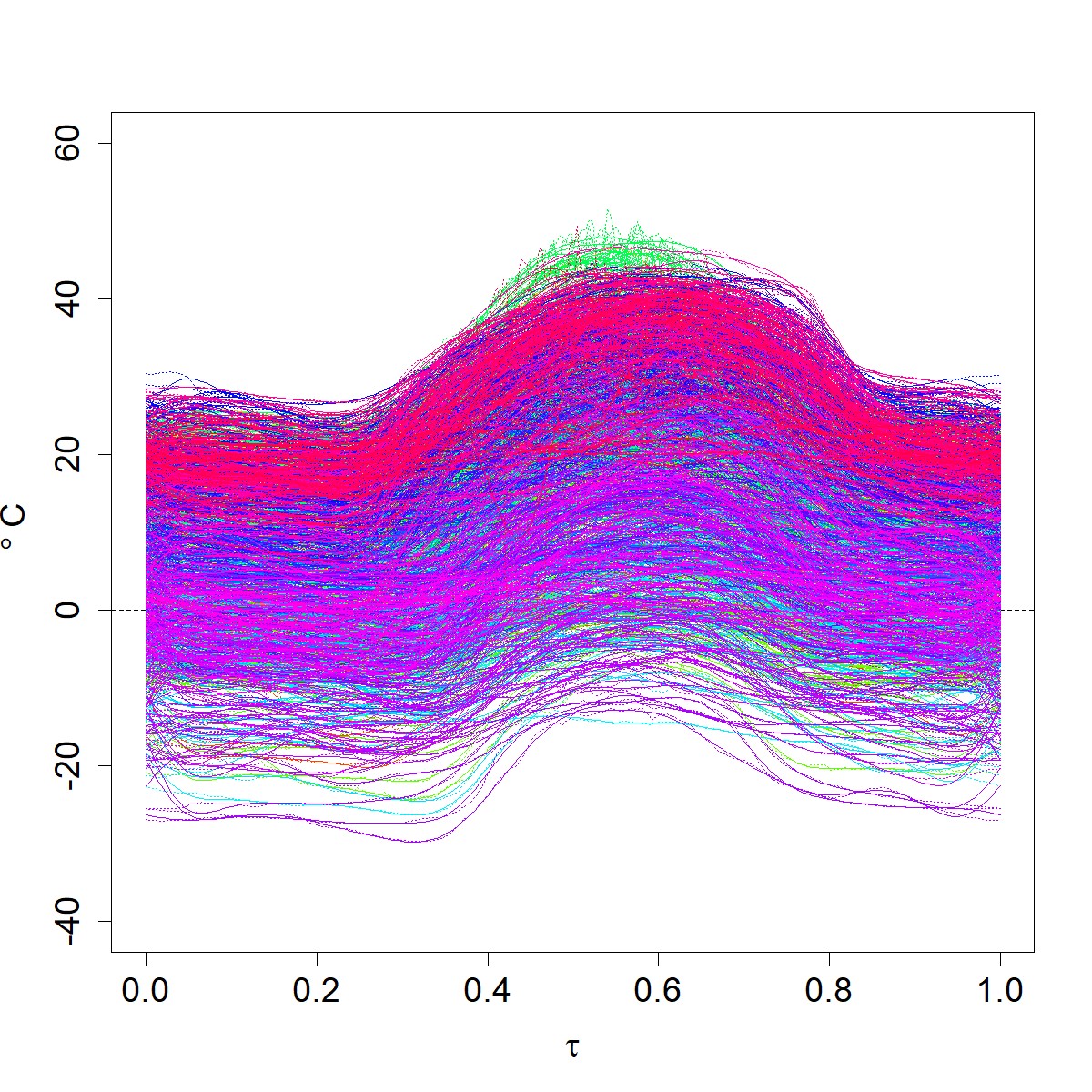}%
            \vspace{-5pt}
         \caption{\it {Lincoln, Nebraska}}       
       \end{subfigure}\hfill
     \begin{subfigure}[t]{0.35\columnwidth}
        \centering
          \includegraphics[width=1\textwidth]{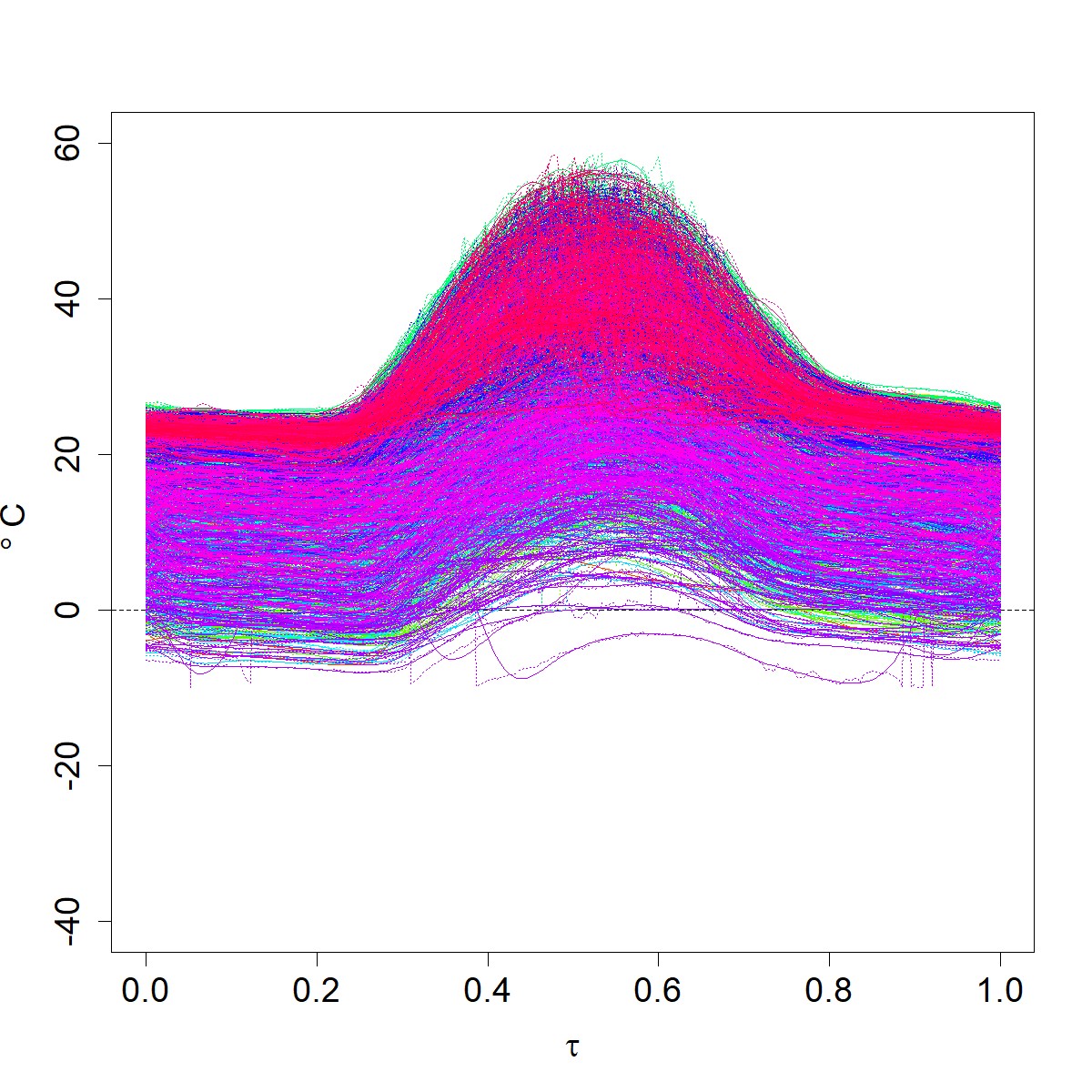}%
       \vspace{-5pt}
        \caption{\it {Selma, Alabama}}
    \end{subfigure}
      \begin{subfigure}[t]{0.35\columnwidth}
        \centering
        \includegraphics[width=1\textwidth]{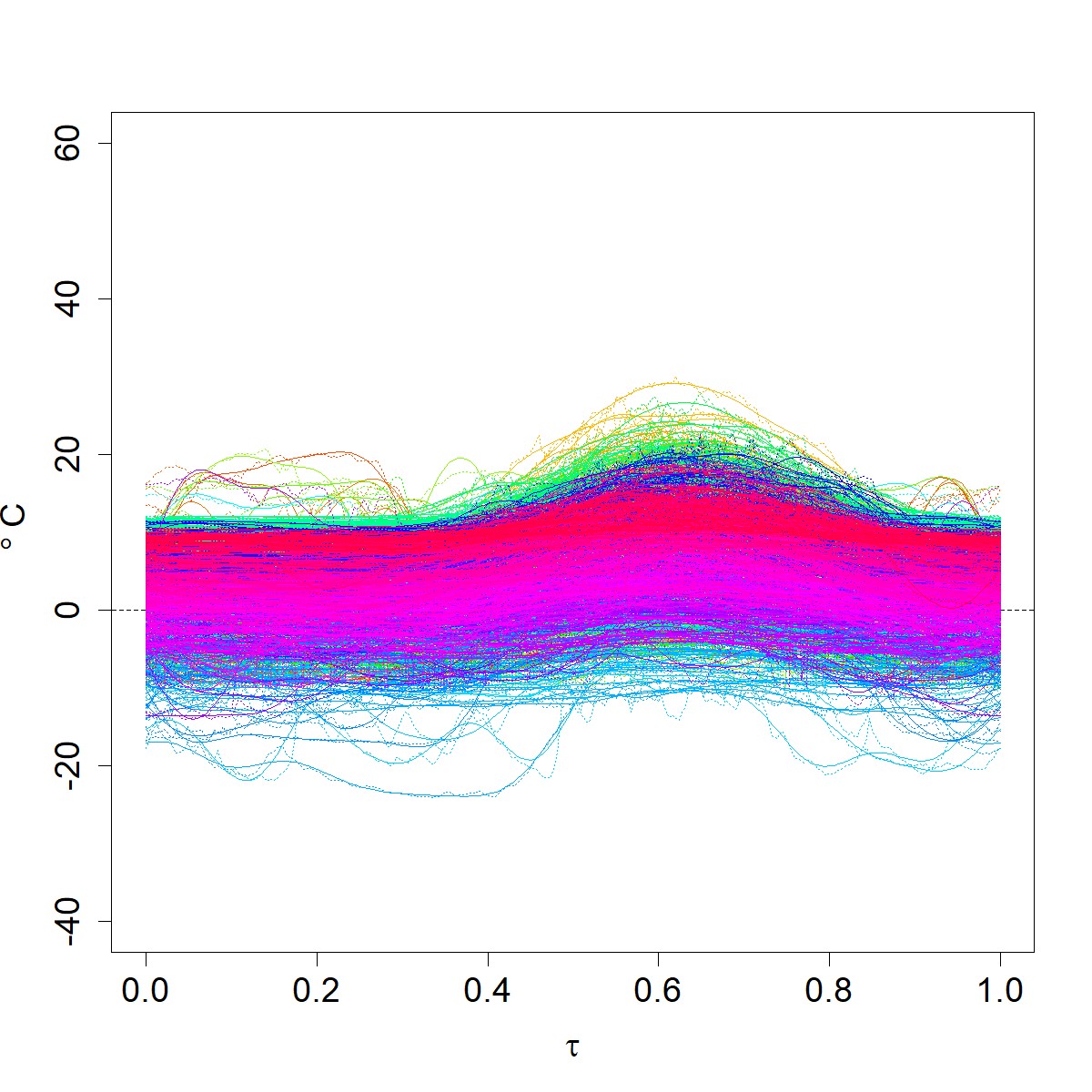}%
           \vspace{-5pt}
         \caption{\it {St. Paul, Alaska}}
       \end{subfigure}
    \end{tabu}}
    \caption[\it  Daily surface temperature curves (solid) fitted to the observations (dotted) for the period 28-Feb 2015-07 Sept 2018. ]
    {\it \small Daily surface temperature curves (solid) fitted to the observations (dotted) for the period 28-Feb. 2015-07 Sept. 2018.} 
    \label{fig:FTStemp}
\end{figure}

It can be observed there is more between-curves variation for Lincoln, Nebraska while there is more within-curve variation for Selma, Alabama. For St. Paul, the more visible variety in colors indicate a stronger variation over longer periods of time. 
To cluster the locations, we factorize the 1288 daily curves into $M=14$ complete seasons, each of length $N=92$. We applied our algorithm described in \autoref{sec3}. For the choice of $k$, we looked at the various criteria discussed in the previous section. These indicated 2 to 4 clusters, with the majority picking 3 clusters. To choose the number of clusters, we used the index that performed best in the simulation study, which was the CH index and which indicates 3 clusters. The results of clustering the locations in 3 clusters by means of our algorithm are given in \autoref{fig:surftemp2}.  The figure seems to indicate the seasonal variation for two clusters is related to two prevailing wind patterns; those in red to the westerlies wind, while those in green to northeasterly trade wind. We moreover observe a separate cluster consisting of mainly locations on the coast in a bay area, island or on a peninsula. In particular, these locations correspond to areas known to be more affected by different atmospheric pressures which are responsible for forming cyclones and that can vary with the El Ni{\~n}o-Southern Oscillations cycle. It is worth remarking that a similar pattern was observed when different time periods were considered.

\setcounter{figure}{2}
\begin{figure}[h!]
    \vspace{-35pt}
    \centering
         \includegraphics[width=0.65\textwidth]{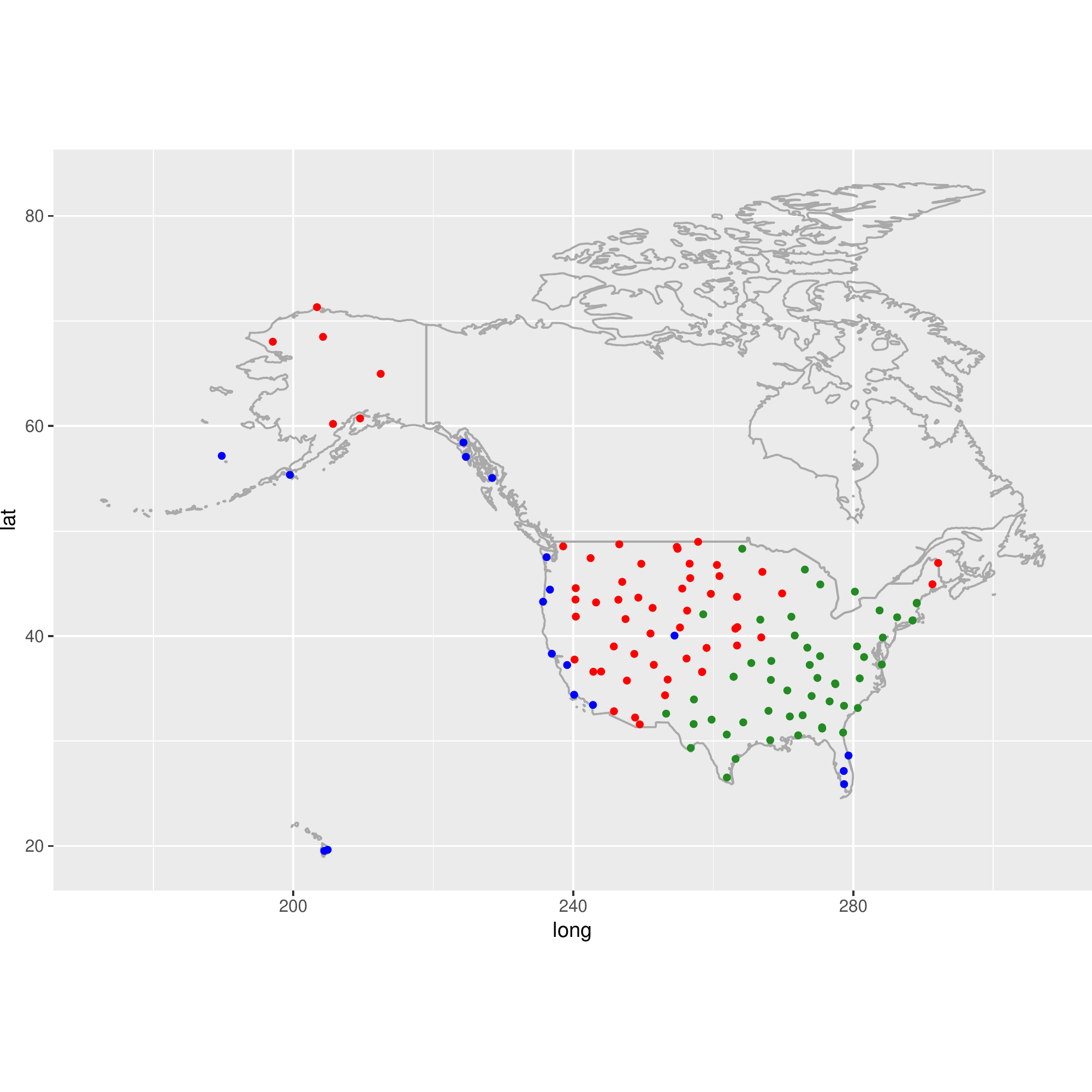}%
    \vspace{-35pt}
   \caption[]{\it  \small Measurement stations clustered according to variation in surface temperature}. 
    \label{fig:surftemp2}
\end{figure}

\FloatBarrier

{\bf Acknowledgements} This work has been supported in part by the
Collaborative Research Center ``Statistical modeling of nonlinear
dynamic processes'' (SFB 823, Teilprojekt A1, A7, C1) of the German Research Foundation
(DFG).  
 
\begin{appendices}   
\def\theequation{6.\arabic{equation}}
\setcounter{equation}{1}

\section{Background material and auxiliary results} \label{ap:background}

Let $\mathcal{H}_i$ for each $i=1,\ldots,n$ be a Hilbert space with inner product $\langle \cdot, \cdot \rangle$. The tensor of these is denoted by
\[
\mathcal{H}:= \mathcal{H}_1 \otimes \ldots \otimes \mathcal{H}_n = \bigotimes_{i=1}^n \mathcal{H}_i
\]

If $\mathcal{H}_i =\mathcal{H}\, \forall i$, then this is the n-th fold tensor product of $\mathcal{H}$. For $A_i \in H_i, 1\le i\le n$ the object $\bigotimes_{i=1}^n A_i$ is a multi-antilinear functional that generates a linear manifold, the usual algebraic tensor product of vector spaces $\mathcal{H}_i$, to which the scalar product  
\[
\innprod{\bigotimes_{i=1}^n A_i}{\bigotimes_{i=1}^n  B_i} = \prod_{i=1}^n\innprod{A_i}{B_i}
\]
can be extended to a pre-Hilbert space. The completion of the above algebraic tensor product is $\bigotimes_{i=1}^n \mathcal{H}_i$. \\
For $A, B, C \in \mathcal L(\mathcal{H})$ we define the following bounded linear mappings. The kronecker product is defined as $(A \widetilde{\bigotimes} B)C = ACB^{\dagger}$, while the transpose Kronecker product is given by $(A \widetilde{\bigotimes}_{\top} B)C = (A \widetilde{\bigotimes} \overline{B})\overline{C}^{\dagger}$. For $A, B, C \in S_2(\mathcal{H})$, we shall denote, in analogy to elements $a,b \in \mathcal{H}$, the Hilbert tensor product as $A \bigotimes B$. We list the following useful properties:
\begin{properties} \label{tensorprop}
Let $\mathcal{H}_i = L^2_\mathbb{C}([0,1]^k)$ for $i = 1,\ldots, n$. Then for $a_i, b_i \in \mathcal{H}_i$ and $A_i, B_i \in S_2(\mathcal{H}_i)$, we have 
\begin{enumerate}
\item $\innprod{A}{B}_{HS} =\Tr({A}{B}^{\dagger})$
\item $\innprod{\bigotimes_{i=1}^n A_i}{\bigotimes_{i=1}^n  B_i}_{{HS}} = \prod_{i=1}^n\innprod{A_i}{B_i}_{HS}$ 
\item $\innprod{a_1\otimes a_2}{b_1\otimes b_2}_{HS} =\innprod{a_1\otimes \overline{a}_2}{b_1\otimes \overline{b_2}}_{\mathcal{H}_i\otimes \mathcal{H}_i} =\innprod{a_1}{b_1}\overline{\innprod{a_2}{b_2}}$
\item If $A_i \in S_1(H)$, then $\prod_{i=1}^{n}\Tr(A_i) = \Tr(\widetilde{\bigotimes}_{i=1}^n  A_i)$
\item $\big((a_1 \otimes \overline{a}_2) \bigotimes (a_3 \otimes \overline{a}_4)\big) =\big((a_1 \otimes {a}_3) \widetilde{\bigotimes} (\overline{a}_2 \otimes \overline{a}_4)\big) =\big((a_1 \otimes {a}_4) \widetilde{\bigotimes}_{\top} {a}_2 \otimes {a}_3)\big)$
\end{enumerate}
\end{properties}

Let $X$ be a random element on a probability space $(\Omega,\mathcal{A},\mathbb{P})$ that takes values in a separable Hilbert space $H$. More precisely, we endow $H$ with the topology induced by the norm on $H$ and assume that $X:\Omega\to H$ is Borel-measurable. The {\em $k$-th order cumulant tensor} is defined by \citep{vde16}
\begin{align}
\label{cumtens}
\cm \big( X_1, \ldots, X_k\big) 
= \sum_{l_1, \ldots l_k\in \nnum} \cm \big( \innprod{X_1}{\psi_{l_1}},\ldots,\innprod{X_k}{\psi_{l_k}}\big)
 (\psi_{l_1} \otimes \cdots \otimes \psi_{l_k}), 
\end{align}   
where $\{\psi_l\}_{l \in \mathbb{N}}$ is an orthonormal basis of $H$ and the cumulants on the right hand side are as usual given by
\[
\cm\big(\innprod{X_1}{\psi_{l_1}},\ldots,\innprod{X_k}{\psi_{l_k}}\big)
=\sum_{\nu=(\nu_1,\ldots,\nu_p)}(-1)^{p-1}\,(p-1)!\,\prod_{r=1}^p \E\Big[\prod_{t\in\nu_r}\innprod{X_t}{\psi_{l_t}}\Big],
\]
where the summation extends over all unordered partitions $\nu$ of $\{1,\ldots,k\}$. The product theorem for cumulants \citep[Theorem 2.3.2]{brillinger} can then be generalised \citep[see e.g.][]{vde16} to simple tensors of random elements of $H$, i.e., $X_t=\otimes_{j=1}^{J_t}X_{tj}$ with $j=1,\ldots,J_t$ and $t=1,\ldots,k$. The joint cumulant tensor is then be given by 
\begin{align} \label{prodcumthm} 
\cm(X_1,\ldots,X_k)
=\sum_{\nu=(\nu_1,\ldots,\nu_p)}S_\nu\Big(\otimes_{n=1}^{p}
\cm\big(X_{tj}|(t,j)\in\nu_n\big)\Big),
\end{align}
where $S_\nu$ is the permutation that maps the components of the tensor back into the original order, that is, $S_\nu(\otimes_{r=1}^p\otimes_{(t,j)\in\nu_r} X_{tj})=X_{11}\otimes\cdots\otimes X_{kJ_t}$.

\section{Distributional properties of the similarity measure} \label{sec:dist}

In order to derive the distributional properties of $\hat{\A}_{i_1,i_2}$ defined in \eqref{eq:est_dist}, we require the following assumptions on the functional processes $\{\,X_{i,t,T} \colon t \in \znum\}_{T\in\nnum}, i =1,\ldots, d$
\begin{assumption}\label{cumglsp}
{\em Assume $\{\,X_{i,t,T} \colon t \in \znum\}_{T\in\nnum}$ are $d$ locally stationary 
zero-mean stochastic processes taking values in $\mathcal{H}_{\rnum}$ and let $\cumbp{m}{m} : L^2([0,1]^{\lfloor m/2\rfloor}) \to  L^2([0,1]^{\lfloor \frac{m+1}{2}\rfloor})$ be a positive operator independent of $T$ such that, for all $j=1,\ldots,m-1$ and some $\ell\in\nnum$,
\begin{align} 
\label{eq:kapmix}
\sum_{t_1,\ldots,t_{k-1} \in \mathbb{Z}} (1+|t_j|^\ell)\snorm{\cumbp{m}{m}}_1 <\infty.
\end{align}
Let us denote
\begin{equation}\label{eq:repstatap}
Y^{(T)}_{i,t}=\XT{i,t}-X_{i,t}^{(t/T)}
\qquad\text{and}\qquad
Y_{i,t}^{(u,v)}=\frac{\Xu{i,t}- X^{(v)}_{i,t}}{(u-v)}
\end{equation}
for $T\ge1$, $1\le t\le T$ and $u,v\in[0,1]$ such that $u\ne v$. Suppose furthermore that $m$-th order joint cumulants satisfy
\begin{enumerate}[(i)]\itemsep-.3ex
\item$\|\cm(\XT{i_1,t_1},\ldots,\XT{i_{m-1},t_{m-1}},Y^{(T)}_{i_m,t_m}) \|_2 \le \frac{1}{T}\snorm{\kappa_{m;t_1-t_m,\ldots,t_{m-1}-t_m}}_1 $,
\item$\|\cm(X_{i_1,t_1}^{(u_1)},\ldots,X_{i_{m-1},t_{m-1}}^{(u_{m-1})},Y_{i_m,t_{m}}^{(u_m,v)}) \|_2 \le \snorm{\kappa_{m;t_1-t_m,\ldots,t_{m-1}-t_m}}_1 $,
\item$\sup_u \|\cm(X_{i_1,t_1}^{(u)},\ldots,X_{i_{m-1},t_{m-1}}^{(u)},X_{i_m,t_m}^{(u)}) \|_2 \le \snorm{\kappa_{m;t_1-t_m,\ldots,t_{m-1}-t_m}}_1$,
\item$\sup_u \|\frac{\partial^\ell}{\partial u^\ell} \cm(X_{i_1,t_1}^{(u)},\ldots,X_{i_{m-1},t_{m-1}}^{(u)},X_{i_m,t_m}^{(u)}) \|_2 \le \snorm{\kappa_{m;t_1-t_m,\ldots,t_{m-1}-t_m}}_1$,
\end{enumerate}}
\end{assumption}
We remark that if the process is in fact stationary, the dependence on localized time $u$ drops. As explained in \autoref{sec2}, the estimator requires splitting the sample $T\in\mathbb N$ as $T = N(T)M(T)$, where $N(T)$ defines the resolution in frequency of the local fDFT and $M(T)$ controls the number of nonoverlapping local fDFT's in \eqref{Fij}. Since they correspond to the number of terms used in a Riemann sum approximating the integrals with respect to $du$ and $d\omega$ they have to be sufficiently large. We assume 
\begin{assumption}\label{ratesNM}
 $M\to\infty$, $N\to\infty$ as $T\to\infty$, such that 	
\[N/M\to\infty
    \quad\text{and}\quad
    N/M^3\to0 \]
\end{assumption}
The number of elements in the blocks grows therefore must grow faster than the number of blocks, but slower than the cube number of blocks. The choice of the number of blocks is carefully discussed in \cite{vDCD18}, who used the integrated periodogram operators as a basis for a stationarity test. 

In order to derive the results in \autoref{sec5}, we require four auxiliary results, which can be proved using similar
arguments as given in \citet{vDCD18} and the proofs are therefore omitted. The first one expresses the cumulants of Hilbert-Schmidt inner products of local periodogram tensors, into the trace of cumulants of simple tensors 
of the local functional DFT's. 

\begin{thm} \label{thm:cumHSinprod}
Let $\E\prod_{i_l \in \{1,\ldots, d\} }^{2n}\snorm{I_{i_l}^{u,\omega}}_2 <\infty$ for some $n \in \mathbb{N}$ and uniformly in $u$ and $\omega$. Then for any $i_1,\ldots, i_{2n} \in \{1,\ldots, d\}$.
\begin{align*}
\cm &\Big(\innprod{ I_{i_1}^{u_{j_1},\omega_{k_1}}}{I_{i_2}^{u_{j_2},\omega_{{k_2}}}}_{HS}, \ldots, \innprod{ I_{i_{2n-1}}^{u_{j_{2n-1}},\omega_{k_{2n-1}}}}{I_{i_{2n}}^{u_{j_{2n}},\omega_{{k_{2n}}}}}_{HS}\Big) 
\\& =\Tr\Big(\sum_{\boldsymbol{P} = P_1 \cup \ldots \cup P_G}S_{\boldsymbol{P}}\Big(\otimes_{g=1}^{G}
 \operatorname{cum}\big(D_{i_{p}}^{u_{j_p},\omega_{k_p}}\ | p \in P_g\Big)  \Big).  
\end{align*}
where the summation is over all indecomposable partitions $\boldsymbol{P} = P_1 \cup \ldots \cup P_G$ of the array 
\[\begin{matrix} 
(1,1) & (1,2) & (1,3) & (1,4) \\
(2,1) & (2,2) & (2,3) & (2,4) \\
\vdots & & &\vdots  \\
\vdots & & &\vdots  \\
(n,1) & (n,2) & (n,3) & (n,4) \\
\end{matrix} \tageq \label{tab:highcumF1}\]
where $p=(l,m)$ and $i_{p} = i_{2l - \delta\{m \in \{1,2\}\}}$,  $k_{p} = (-1)^{m} k_{2l-\delta{\{m \in \{1,2\}\}}}$ and $j_p= j_{2l-\delta{\{m \in \{1,2\}\}}}$ for $l \in\{1,\ldots, n\}$ and $m \in \{1,2,3,4\}$ and where $\delta_{\{A\}}$ equals 1 if event $A$ occurs and $0$ otherwise.
\end{thm}

We additionally require  auxiliary results on the properties of the local fDFTs.

\begin{lemma} \label{lem:cumfixu}
Suppose \autoref{cumglsp} is satisfied $k$ moments and $\sum_{j=1}^{k} \omega_{j} \equiv 0 \mod 2\pi$ then
\begin{align*} 
\bigsnorm{\cm\left(D_i^{u_j,\omega_1},\dots,D_i^{u_j,\omega_k}\right)-\frac{(2\pi)^{1-k/2}}{N^{k/2-1}}\mathcal{F}_{u_i,\omega_1,\dots,\omega_{k-1}}}_1 = O\left(N^{-k/2} \times \frac{N}{M^2}\right).
\end{align*}
additionally, for $|i_1|=|i_2|=k/2$, we have 
\begin{align*} 
\bigsnorm{\cm\left(D_i^{u_j,\omega_1},\dots,D_i^{u_j,\omega_k}\right)-\frac{(2\pi)^{1-k/2}}{N^{k/2-1}}\mathcal{F}^{i_1,i_2}_{u_i,\omega_1,\dots,\omega_{k-1}}}_1 = O\left(N^{-k/2} \times \frac{N}{M^2}\right).
\end{align*}
where $ \mathcal{F}^{i_1,i_2}_{u_i,\omega_1,\dots,\omega_{k-1}} $ denotes the cross spectral density operator of series $i_1$ and $i_2$ of order $k$ at time $u_i$.
\end{lemma}

When evaluated off the manifold, i.e., $\sum_{j=1}^{k} \omega_{j} \ne 0 \mod 2\pi$, the above cumulant is of lower order (see \autoref{cor:cumbound}). 
A direct consequence of the proof of \autoref{lem:cumfixu} is the following corollary
\begin{Corollary} \label{cor:cumbound}
For $i_1, \ldots, i_k \in \{1,\ldots, d\}$ We have for any $p \ge 1$
\begin{align}\bigsnorm{\cm\left(D_{i_1}^{u_1,\omega_1},\dots,D_{i_k}^{u_k,\omega_k}\right)}_p = O\left(N^{1-k/2}\right) \label{eq:boundonmani}\end{align}
uniformly in $\omega_1,\ldots,\omega_k$ and $u_1,\ldots, u_k$. Moreover, if $~\sum_{j=1}^{k}\omega_j~ \ne 0\mod 2\pi$ then
 \begin{align}\bigsnorm{\cm\left(D_{i_1}^{u_1,\omega_1},\dots,D_{i_k}^{u_k,\omega_k}\right)}_p =O\left(N^{-k/2}\right).\label{eq:boundoffmani}\end{align}
\end{Corollary} 
Additionally, when the local fDFT's are evaluated on different midpoints then we have the following lemma. 
\begin{lemma}\label{lem:cumdifu}
If \autoref{cumglsp} is satisfied and $|{j_1}-{j_2}|>1$ for some midpoints $u_{j_1}$ and $u_{j_2}$ then 
\begin{align*} & \bigsnorm{\cm\left(D_{i_1}^{u_1,\omega_1},\dots,D_{i_k}^{u_k,\omega_k}\right)}_1 = O\left(N^{-k/2} M^{-1} \right)
\end{align*}
uniformly in $\omega_1,\ldots,\omega_k$.
\end{lemma}
In the following, for random variables $Y_0, Y_1, Y_2, Y_3$ let $\operatorname{cum}_{m_0, m_1, m_2, m_3}(Y_0, Y_1,Y_2, Y_3)$ denote the joint cumulant
\[
\operatorname{cum}(\underbrace{Y_0,\ldots,Y_0}_{m_0\ \text{times}},\underbrace{Y_1,\ldots,Y_1}_{m_1\ \text{times}},\underbrace{Y_2,\ldots,Y_2}_{m_2\ \text{times}},\underbrace{Y_3,\ldots,Y_3}_{m_3\ \text{times}}),
\]
where $0 \le m_i \le m, i=0,1,2,3$ s.t. $\sum_{i=0}^3 m_i =m$. 
Using the previous statements, we can derive the order of the higher order joint cumulants of elements defined in $\eqref{Fij}$.
\begin{thm} \label{thm:highcumF1234}
If \autoref{cumglsp} is satisfied then for finite $m$
\begin{align*}
& T^{m/2}\operatorname{cum}_{m_0,m_1,m_2,m_3}(\hat F_{i_1,i_2},\hat F_{i_3,i_4},\hat F_{i_5,i_6},\hat F_{i_7,i_8})
\\& \frac{1}{T^{m/2}}\sum_{k_1,\ldots,k_m=1}^{\lfloor N/2 \rfloor} \sum_{\substack{j_1,\ldots,j_{m}=1}}^M\Tr\Big(\sum_{\boldsymbol{P} = P_1 \cup \ldots \cup P_G}S_{\boldsymbol{P}}\Big(\otimes_{g=1}^{G}
\operatorname{cum}\big(D_{i_{p}}^{u_{j_p},\omega_{k_p}}\ | p \in P_g\Big)\Big) =O(T^{1-m/2}).
\end{align*}
uniformly in  $0 \le m_i \le m$ s.t. $\sum_{i=0}^3 m_i=m$.
\end{thm}
\begin{proof}[Proof of \autoref{thm:highcumF1234}]
For a fixed partition $P=\{P_1,\ldots,P_G\}$,  let the cardinality of set $P_g$ be denoted by $|P_{g}|=\mathscr{C}_g$. 
By \eqref{eq:boundonmani} of Corollary \ref{cor:cumbound} and Lemma \ref{lem:cumdifu} an upperbound of \eqref{tab:highcumF1} is given by
\begin{align*}
O \Big ( T^{-m/2} \sum_{k_1,\ldots,k_m=1}^{\lfloor N/2 \rfloor}\sum_{\substack{j_1,\ldots,j_{m}=1}}^M \prod_{g=1}^{G} \frac{1}{N^{\mathscr{C}_g/2-1}} M^{-\delta_{\{\exists p_1, p_2 \in P_g: |j_{p_1} - j_{p_2}|>1\}}}  \Big )
\tageq \label{cumhighF1UB}
\end{align*}
Similar to Lemma 4.3 of \citet{vDCD18}, we can show inductively that the indecomposability of the array \eqref{tab:highcumF1} and the behavior of the joint cumulants of the local fDFT's at different midpoints imply this is at most of order 
\[
O(N^{m/2} M^{-m/2} E^m M N^{-2m+G})   =O(T^{1-m/2} N^{G-m-1}). 
\]
Thus, partitions of size $G \le m+1$ will vanish as $T \to \infty$. For $G \ge m+2$, indecomposability of the array requires to stay on the frequency manifold  (see equation \eqref{eq:boundoffmani} of Corollary \ref{cor:cumbound}) and therefore imposes additional restrictions in frequency direction. It can be shown \citep[][Proposition 4.1]{ vDCD18} that for a partition of size $G=m+r_1+1$ with $r_1 \ge 1$ of the array \eqref{tab:highcumF1} only partitions with at least $r_1$ restrictions in frequency direction are indecomposable if $m > 2$, while if $m=2$ there must be at least 1 restriction in frequency direction. Consequently, the joint cumulant is at most of order $O(T^{1-n/2} N^{m+r_1+1-m-1}N^{-{r_1}})=O(T^{1-n/2})$.
\end{proof}

\begin{proof}[\bf  Proof of \autoref{thm:con}]
Using \autoref{thm:cumHSinprod} with $n=1$ implies

\begin{align*}
\E {F}_{i_1,i_2} &	=\frac{1}{T}\sum_{k=1}^{\lfloor N/2 \rfloor} \sum_{j=1}^M \Tr \Big( \E \big[\fdft{1}{1}{i_1} \otimes \fdftc{1}{1}{i_1} \otimes \fdftcl{1}{1}{i_2}\otimes\fdftl{1}{1}{i_2} \big]\Big).
\end{align*}
Rewriting this expectation in cumulant tensors, we get
\begin{align*}
\E {F}_{i_1,i_2} &	=\frac{1}{T}\sum_{k=1}^{\lfloor N/2 \rfloor} \sum_{j=1}^M \Tr \Big(S_{1234 } \cm \big((\fdft{1}{1}{i_1}, \fdftc{1}{1}{i_1},\fdftcl{1}{1}{i_2},\fdftl{1}{1}{i_2}) \big)
\\&+\frac{1}{T}\sum_{k=1}^{\lfloor N/2 \rfloor} \sum_{j=1}^M \Tr \Bigg(S_{1234 }\Big(\cm(\fdft{1}{1}{i_1}, \fdftc{1}{1}{i_1}) \otimes \cm(\fdftcl{1}{1}{i_2},\fdftl{1}{1}{i_2}) \Big)\Bigg)
\\& +\frac{1}{T}\sum_{k=1}^{\lfloor N/2 \rfloor} \sum_{j=1}^M \Tr\Bigg( S_{1324 }\Big(\cm(\fdft{1}{1}{i_1}, \fdftcl{1}{1}{i_2}) \otimes \cm(\fdftc{1}{1}{i_1},\fdftl{1}{1}{i_2}) \Big)\Bigg)
\\&+\frac{1}{T}\sum_{k=1}^{\lfloor N/2 \rfloor} \sum_{j=1}^M \Tr \Bigg(S_{1423 }\Big(\cm(\fdft{1}{1}{i_1}, \fdftl{1}{1}{i_2}) \otimes \cm(\fdftc{1}{1}{i_1},\fdftcl{1}{1}{i_2}) \Big)\Bigg)
\end{align*}
By \autoref{cor:cumbound} and\autoref{lem:cumfixu} we thus find
\begin{align*}
\E {F}_{i_1,i_2} &	= \frac{1}{T}\sum_{k=1}^{\lfloor N/2 \rfloor} \sum_{j=1}^M \innprod{\mathcal{F}_{u_j,\omega_{k}}^{i_1}}{\mathcal{F}_{u_j,\omega_{k-1}}^{i_2}}_{HS} + O(\frac{1}{M^2})+O(\frac{1}{N}).
\end{align*}
Hence,
\[
\lim_{N,M\to\infty}\operatorname \E F_{i_1,i_2}
	=\frac1{2\pi}\int_0^\pi\int_0^1 \innprod{\mathcal{F}_{u,\omega}^{i_1}}{\mathcal{F}_{u,\omega}^{i_2}}_{HS} \mathrm du\mathrm d\omega = \frac1{4\pi}\int_{-\pi}^\pi\int_0^1 \innprod{\mathcal{F}_{u,\omega}^{i_1}}{\mathcal{F}_{u,\omega}^{i_2}}_{HS}  \mathrm du\mathrm d\omega.
\]
Secondly, we have for any $i_1, i_2, i_3, i_4 \in \{1, \ldots,d\}$ 
\[T \cv({F}_{i_1,i_2}, {F}_{i_3,i_4})  = T\cm\big({{F}_{i_1,i_2}, \overline{{F}_{i_3,i_4}}}\big)\]
Hence, if \autoref{cumglsp} is satisfied with $m=8$, \autoref{thm:highcumF1234} implies this term is of order $O(1)$. summarizing, we find ${F}_{i_1,i_1}, {F}_{i_2,i_2}, {F}_{i_1,i_2}$ and ${F}_{i_2,i_1}$ are asymptotically unbiased and jointly convergence in probability. The continuous mapping theorem establishes then that $\hat{\A}_{i_1,i_2}$ is a $\sqrt{T}$-consistent estimator of ${\A}_{i_1,i_2}$ for any $i_1, i_2,\in \{ 1 \ldots, d\}$.
\end{proof}

\begin{proof}[\bf Proof of \autoref{thm:AN}]
 If \autoref{cumglsp} holds for all moments,  then \autoref{thm:highcumF1234}, yields that for $m>2$
\[ T^{m/2}\operatorname{cum}_{m_0,m_1,m_2,m_3}(F_{i_1,i_2},F_{i_3,i_4}, F_{i_5,i_6}, F_{i_7,i_8}) \to 0 \quad \text{as } T \to \infty,\]
from which asymptotic joint normality of $F_{i_1,i_2}, F_{i_3,i_4},  F_{i_5,i_6}, F_{i_7,i_8}$ follows, i.e.,  we have
\begin{eqnarray*}
\hspace*{-20pt}
\sqrt{T}
\begin{pmatrix}
4 \pi F_{i_1,i_1}-\E(F_{i_1,i_1})\\
4 \pi F_{i_2,i_2}-\E(F_{i_2,i_2})\\
4 \pi F_{i_1,i_2}-\E(F_{i_1,i_2})\\
4 \pi F_{i_2,i_1}-\E(F_{i_1,i_1}) 
\end{pmatrix} \longrightarrow  \mathcal{N} (\mathbf{0}, \boldsymbol{\Sigma}) ,
\end{eqnarray*}
where
\begin{eqnarray} \label{eq:bigSig}
\boldsymbol{\Sigma}=
\left(\begin{array}{cccc}
 \V({F}_{i_1,i_1}) &\cv({F}_{i_1,i_1},{F}_{i_2,i_2}) &\cv({F}_{i_1,i_1},{F}_{i_1,i_2}) & \cv({F}_{i_1,i_1},{F}_{i_2,i_1})\\
\cv({F}_{i_1,i_1},{F}_{i_2,i_2})&  \V({F}_{i_2,i_2}) &\cv({F}_{i_2,i_2},{F}_{i_1,i_2}) & \cv({F}_{i_2,i_2},{F}_{i_2,i_1})\\
\cv({F}_{i_1,i_1},{F}_{i_1,i_2}) &\cv({F}_{i_2,i_2},{F}_{i_1,i_2}) &  \V({F}_{i_1,i_2}) & \cv({F}_{i_1,i_2},{F}_{i_2,i_1})\\
 \cv({F}_{i_1,i_1},{F}_{i_2,i_1})& \cv({F}_{i_2,i_2},{F}_{i_2,i_1}) &\cv({F}_{i_1,i_2},{F}_{i_2,i_1}) & \V({F}_{i_2,i_1})
\end{array}\right).
\end{eqnarray}
To derive from this the distribution of $\hat{\A}_{i_1,i_2}$, consider the function $g: \mathbb{R}^4 \to \mathbb{R}$
\[g(x_1,x_2,x_3,x_4) =1-\frac{x_3}{(x_1+x_2)}-\frac{x_4}{(x_1+x_2)} \]
of which the gradient is given by
\[
\nabla g^{\top}(\boldsymbol{x})
 = \left(\begin{array}{c}
x_3 (x_1+x_2)^{-2}+x_4 (x_1+x_2)^{-2}\\
x_3 (x_1+x_2)^{-2}+x_4 (x_1+x_2)^{-2}\\
- (x_1+x_2)^{-1}\\
- (x_1+x_2)^{-1}
\end{array}\right)
=
\frac{1}{(x_1+x_2)} \left(\begin{array}{c}
\frac{x_3+x_4}{ (x_1+x_2)}\\
\frac{x_3+x_4}{ (x_1+x_2)}\\
-1\\
-1
\end{array}\right).
\]
Then since we can write
\[\hat{\A}_{i_1,i_2}=g\big(F_{i_1,i_1}, F_{i_2,i_2}, F_{i_1,i_2}, F_{i_2,i_1}\big) =1-\frac{F_{i_1,i_2}}{(F_{i_1,i_1}+F_{i_2,i_2})}-\frac{F_{i_2,i_1}}{(F_{i_1,i_1}+F_{i_2,i_2})},\]
the Delta method implies, that as $T \to \infty$,
\begin{align} \label{eq:ANalt}
\left\{\sqrt{T}\left(\hat{\A}_{i_1,i_2} - \A_{i_1,i_2}\right)\right\}_{\{i_1,i_2 \in [d]\}} \to  \mathcal{N} \big(\boldsymbol{0}, \nabla g^{\top}(\boldsymbol{x}) \boldsymbol{\Sigma} \nabla g(\boldsymbol{x}) \big),
\end{align}
where for fixed $i_1, i_2 \in  \{1,\ldots, d\}$,
\[ \boldsymbol{x}=\left(\begin{array}{c}
{x}_{1}\\
x_2 \\
x_3\\
x_4\\
\end{array}\right)
=\left(\begin{array}{c}
 \frac1{4\pi}\int_{-\pi}^\pi\int_0^1 \innprod{\mathcal{F}_{u,\omega}^{i_1}}{\mathcal{F}_{u,\omega}^{i_1}}_{HS}  \mathrm du\mathrm d\omega
 \\
 \frac1{4\pi}\int_{-\pi}^\pi\int_0^1 \innprod{\mathcal{F}_{u,\omega}^{i_2}}{\mathcal{F}_{u,\omega}^{i_2}}_{HS}  \mathrm du\mathrm d\omega\\
  \frac1{4\pi}\int_{-\pi}^\pi\int_0^1 \innprod{\mathcal{F}_{u,\omega}^{i_1}}{\mathcal{F}_{u,\omega}^{i_2}}_{HS}  \mathrm du\mathrm d\omega
  \\
   \frac1{4\pi}\int_{-\pi}^\pi\int_0^1 \innprod{\mathcal{F}_{u,\omega}^{i_2}}{\mathcal{F}_{u,\omega}^{i_1}}_{HS}  \mathrm du\mathrm d\omega,
\end{array} \right) \tageq \label{eq:x}
\]  
and  $\boldsymbol{\Sigma} $ is defined in \eqref{eq:bigSig}.
Next, consider the covariance structure. We have,
\begin{align*}
& T \cv({F}_{i_1,i_2}, {F}_{i_3,i_4}) 
= T \cm\left(\frac{1}{T}\sum_{k_1=1}^{\lfloor N/2 \rfloor} \sum_{j_1=1}^M \innprod{ I_{i_1}^{u_{j_1},\omega_{k_1}}}{I_{i_2}^{u_{j_1},\omega_{k_1-1}}}_{HS},\frac{1}{T}\sum_{k_2=1}^{\lfloor N/2 \rfloor} \sum_{j_2=1}^M \overline{\innprod{I_{i_3}^{u_{j_2},\omega_{k_2}}}{I_{i_4}^{u_{j_2},\omega_{k_2-1}}}}_{HS}\right) 
\end{align*}
By \autoref{thm:highcumF1234}, this satisfies
\begin{align*}
T \cv({F}_{i_1,i_2}, {F}_{i_3,i_4}) & 
=\frac{1}{T}\sum_{k_1,k_2=1}^{\lfloor N/2 \rfloor} \sum_{j_1,j_2=1}^M \cm\Big(\Tr \big(\fdft{1}{1}{i_1}\otimes\fdftc{1}{1}{i_1} \otimes \fdftcl{1}{1}{i_2} \otimes \fdftl{1}{1}{i_2})\big),\\&  
\phantom{\frac{1}{T}\sum_{k_1,k_2=1}^{\lfloor N/2 \rfloor} \sum_{j_1,\ldots,j_n=1}^M \cm\Big(\Tr}\Tr \big(\fdftc{2}{2}{i_3}\otimes\fdft{2}{2}{i_3} \otimes \fdftl{2}{2}{i_4} \otimes \fdftcl{2}{2}{i_4}\big)\Big) 
 \\& =\frac{1}{T}\sum_{k_1,k_2=1}^{\lfloor N/2 \rfloor} \sum_{j_1,j_2=1}^M\Tr\Big(\sum_{\boldsymbol{P} = P_1 \cup \ldots \cup P_G}S_{\boldsymbol{P}}\Big(\otimes_{g=1}^{G}
 \operatorname{cum}\big(D_{i_{p}}^{u_{j_p},\omega_{k_p}}\ | p \in P_g\Big)  \Big)\end{align*}
where $p=(l,m)$ with $k_{p} = (-1)^{l-m} k_l-\delta_{\{m \in \{3,4\}\}}$, $j_p= j_l$ and $i_{p} = i_{2l - \delta\{m \in \{1,2\}\}}$ for $l \in\{1, 2\}$ and $m \in \{1,2,3,4\}$.
 In particular, we are interested in all indecomposable partitions of the array where the summation is overall indecomposable partitions of the array
\[\begin{matrix}
\underbrace{D_{i_1}^{u_{j_1},\omega_{k_1}}}_{1}&\underbrace{D_{i_1}^{u_{j_1},-\omega_{k_1}}}_2 & \underbrace{D_{i_2}^{u_{j_1},-\omega_{k_1-1}}}_3& \underbrace{D_{i_2}^{u_{j_1},\omega_{k_1-1}}}_4\\ 
\underbrace{D_{i_3}^{u_{j_2},-\omega_{k_2}}}_5 &\underbrace{D_{i_3}^{u_{j_2},\omega_{k_2}}}_6&\underbrace{D_{i_4}^{u_{j_2},\omega_{k_2-1}}}_7& \underbrace{D_{i_4}^{u_{j_2},-\omega_{k_2-1}}}_8
\end{matrix}\]
The significant partitions of the form $\cm_4 \cm_2 \cm2$ are
\begin{align*}
& \Tr\Big(S_{(1256)(34)(78)}\Big( \delta_{j_1,j_2}\left[(\frac{2 \pi}{N}\F^{i_1,i_3}_{u_{j_1},\omega_{k_1},-\omega_{k_1},-\omega_{k_2}}+\Eps_4) \otimes (\F^{i_2}_{u_{j_1},-\omega_{k_1-1}}+\Eps_2) \otimes (\F^{i_4}_{u_{j_2},\omega_{k_2-1}}+\Eps_2)\right] \Big) \Big)\\ 
& \Tr\Big(S_{(1278)(34)(56)}\Big(\delta_{j_1,j_2}\left[ (\frac{2 \pi}{N}\F^{i_1,i_4}_{u_{j_1},\omega_{k_1},-\omega_{k_1},\omega_{k_2-1}}+\Eps_4) \otimes (\F^{i_2}_{u_{j_1},-\omega_{k_1-1}}+\Eps_2) \otimes ( \F^{i_3}_{u_{j_2},-\omega_{k_2}}+\Eps_2) \right]\Big) \Big)\\ 
&\Tr\Big(S_{(3456)(12)(78)}\Big(\delta_{j_1,j_2}\left[(\frac{2 \pi}{N}\F^{i_2,i_3}_{u_{j_1},-\omega_{k_1-1},\omega_{k_1-1},-\omega_{k_2}}+\Eps_4)\otimes (\F^{i_1}_{u_{j_1},\omega_{k_1}}+\Eps_2) \otimes(\F^{i_4}_{u_{j_2},\omega_{k_2-1}}+\Eps_2)\right] \Big) \Big)  \\ 
& \Tr\Big(S_{(3478)(12)(56)} \Big(\delta_{j_1,j_2}\left[(\frac{2 \pi}{N}\F^{i_2,i_4}_{u_{j_1},-\omega_{k_1-1},\omega_{k_1-1},\omega_{k_2-1}}+\Eps_4) \otimes ( \F^{i_1}_{u_{j_1},\omega_{k_1}}+\Eps_2) \otimes (\F^{i_3}_{u_{j_2},-\omega_{k_2}}+\Eps_2\right]\Big) \Big).  
\end{align*}

The significant partitions of the form $\cm_2 \cm_2 \cm_2 \cm_2$ are
\begin{align*}
\Tr\Big(S_{(12)(37)(56)(48)}&\Big(\delta_{j_1,j_2} \delta_{k_1,k_2} \big[\F^{i_1}_{u_{j_1},\omega_{k_1}}\otimes \F^{i_2,i_4}_{u_{j_1},-\omega_{k_1-1}}\otimes \F^{i_3}_{u_{j_2},-\omega_{k_2}} \otimes \F^{i_2,i_4}_{u_{j_1},\omega_{k_1-1}}+\Eps_2\big]\Big)\\ 
\Tr\Big(S_{(12)(36)(78)(45)}& \Big(\delta_{j_1,j_2} \delta_{k_1-1,k_2} \big[\F^{i_1}_{u_{j_1},\omega_{k_1}}\otimes\F^{i_2,i_3}_{u_{j_1},-\omega_{k_1-1}}\otimes \F^{i_4}_{u_{j_2},\omega_{k_2-1}} \otimes \F^{i_2,i_3}_{u_{j_1},\omega_{k_1-1}}+\Eps_2\big]\Big)\\  
\Tr\Big(S_{(15)(26)(37)(48)}&\Big(\delta_{j_1,j_2} \delta_{k_1,k_2} \big[\F^{i_1,i_3}_{u_{j_1},\omega_{k_1}}\otimes \F^{i_1,i_3}_{u_{j_1},-\omega_{k_1}}\otimes \F^{i_2,i_4}_{u_{j_1},-\omega_{k_1-1}} \otimes \F^{i_2,i_4}_{u_{j_1},\omega_{k_1-1}}+\Eps_2\big]\Big) \\ 
\Tr\Big(S_{(15)(26)(34)(78)}&\Big(\delta_{j_1,j_2} \delta_{k_1,k_2} \big[\F^{i_1,i_3}_{u_{j_1},\omega_{k_1}}\otimes \F^{i_1,i_3}_{u_{j_1},-\omega_{k_1}}\otimes \F^{i_2}_{u_{j_1},-\omega_{k_1-1}}\otimes \F^{i_4}_{u_{j_2},\omega_{k_2-1}}+\Eps_2\big]\Big)  \\ 
\Tr\Big(S_{(18)(27)(34)(56)}&\Big(\delta_{j_1,j_2} \delta_{k_1,k_2-1} \big[\F^{i_1,i_4}_{u_{j_1},\omega_{k_1}}\otimes \F^{i_1,i_4}_{u_{j_1},-\omega_{k_1}}\otimes \F^{i_2}_{u_{j_1},-\omega_{k_1-1}}\otimes \F^{i_3}_{u_{j_2},-\omega_{k_2}}+\epsilon_2\big]\Big).
\end{align*}
The covariance structure are derived in more detail in the proof of \autoref{dist_H0}. 
\end{proof}

\begin{proof}[\bf Proof of \autoref{dist_H0}]
If the series are independent, possible cross spectral terms drop out. The significant partitions of the form  $\cm_4 \cm_2 \cm2$ become therefore
\begin{align*}
& \Tr\Big(S_{(1256)(34)(78)}\Big(\delta_{i_1,i_3} \delta_{j_1,j_2}\left[(\frac{2 \pi}{N}\F^{i_1}_{u_{j_1},\omega_{k_1},-\omega_{k_1},-\omega_{k_2}}+\Eps_4) \otimes (\F^{i_2}_{u_{j_1},-\omega_{k_1-1}}+\Eps_2) \otimes (\F^{i_4}_{u_{j_2},\omega_{k_2-1}}+\Eps_2)\right] \Big) \Big)\\ 
& \Tr\Big(S_{(1278)(34)(56)}\Big(\delta_{i_1,i_4}\delta_{j_1,j_2}\left[ (\frac{2 \pi}{N}\F^{i_1}_{u_{j_1},\omega_{k_1},-\omega_{k_1},\omega_{k_2-1}}+\Eps_4) \otimes (\F^{i_2}_{u_{j_1},-\omega_{k_1-1}}+\Eps_2) \otimes ( \F^{i_3}_{u_{j_2},-\omega_{k_2}}+\Eps_2) \right]\Big) \Big)\\ 
&\Tr\Big(S_{(3456)(12)(78)}\Big(\delta_{i_2,i_3}\delta_{j_1,j_2}\left[(\frac{2 \pi}{N}\F^{i_2}_{u_{j_1},-\omega_{k_1-1},\omega_{k_1-1},-\omega_{k_2}}+\Eps_4)\otimes (\F^{i_1}_{u_{j_1},\omega_{k_1}}+\Eps_2) \otimes(\F^{i_4}_{u_{j_2},\omega_{k_2-1}}+\Eps_2)\right] \Big) \Big)  \\ 
& \Tr\Big(S_{(3478)(12)(56)} \Big(\delta_{i_2,i_4}\delta_{j_1,j_2}\left[(\frac{2 \pi}{N}\F^{i_2}_{u_{j_1},-\omega_{k_1-1},\omega_{k_1-1},\omega_{k_2-1}}+\Eps_4) \otimes ( \F^{i_1}_{u_{j_1},\omega_{k_1}}+\Eps_2) \otimes (\F^{i_3}_{u_{j_2},-\omega_{k_2}}+\Eps_2\right]\Big) \Big).  
\end{align*}
The significant partitions of the form $\cm_2 \cm_2 \cm_2 \cm_2$ are
\begin{align*}
\Tr\Big(S_{(12)(37)(56)(48)}&\Big(\delta_{i_2,i_4}\delta_{j_1,j_2} \delta_{k_1,k_2} \big[\F^{i_1}_{u_{j_1},\omega_{k_1}}\otimes \F^{i_2}_{u_{j_1},-\omega_{k_1-1}}\otimes \F^{i_3}_{u_{j_2},-\omega_{k_2}} \otimes \F^{i_2}_{u_{j_1},\omega_{k_1-1}}+\Eps_2\big]\Big)\\ 
\Tr\Big(S_{(12)(36)(78)(45)}& \Big(\delta_{i_2,i_3}\delta_{j_1,j_2} \delta_{k_1-1,k_2} \big[\F^{i_1}_{u_{j_1},\omega_{k_1}}\otimes\F^{i_2}_{u_{j_1},-\omega_{k_1-1}}\otimes \F^{i_4}_{u_{j_2},\omega_{k_2-1}} \otimes \F^{i_2}_{u_{j_1},\omega_{k_1-1}}+\Eps_2\big]\Big)\\  
\Tr\Big(S_{(15)(26)(37)(48)}&\Big(\delta_{i_1,i_3}\delta_{i_2,i_4}\delta_{j_1,j_2} \delta_{k_1,k_2} \big[\F^{i_1}_{u_{j_1},\omega_{k_1}}\otimes \F^{i_1}_{u_{j_1},-\omega_{k_1}}\otimes \F^{i_2}_{u_{j_1},-\omega_{k_1-1}} \otimes \F^{i_2}_{u_{j_1},\omega_{k_1-1}}+\Eps_2\big]\Big) \\ 
\Tr\Big(S_{(15)(26)(34)(78)}&\Big(\delta_{i_1,i_3}\delta_{j_1,j_2} \delta_{k_1,k_2} \big[\F^{i_1}_{u_{j_1},\omega_{k_1}}\otimes \F^{i_1}_{u_{j_1},-\omega_{k_1}}\otimes \F^{i_2}_{u_{j_1},-\omega_{k_1-1}}\otimes \F^{i_4}_{u_{j_2},\omega_{k_2-1}}+\Eps_2\big]\Big)  \\ 
\Tr\Big(S_{(18)(27)(34)(56)}&\Big(\delta_{i_1,i_4}\delta_{j_1,j_2} \delta_{k_1,k_2-1} \big[\F^{i_1}_{u_{j_1},\omega_{k_1}}\otimes \F^{i_1}_{u_{j_1},-\omega_{k_1}}\otimes \F^{i_2}_{u_{j_1},-\omega_{k_1-1}}\otimes \F^{i_3}_{u_{j_2},-\omega_{k_2}}+\epsilon_2\big]\Big).
\end{align*}

In order to give meaning to the covariance structure, we need to investigate how it `operates' as a result of the permutation that occurs due to the cumulant operation. For the second order cumulant structure, Theorem \ref{thm:cumHSinprod}] implies that the original order of the simple tensors has structure $\Tr(S_{1234}\cdot \otimes \cdot \otimes \cdot \otimes \cdot \widetilde{\bigotimes} S_{5678} \cdot \otimes \cdot \otimes \cdot \otimes \cdot )$. Using the properties listed in  \autoref{tensorprop} we obtain the following.
\begin{prop}[Trace of permutations of order 8]\label{prop:2ndstruc}
Let $S_{\nu}$ be the permutation operator as defined in \eqref{prodcumthm} that acts on a tensor Hilbert space of appropriate order. The object $\Tr(S_{1234}\cdot \otimes \cdot \otimes \cdot \otimes \cdot \widetilde{\bigotimes} S_{5678} \cdot \otimes \cdot \otimes \cdot \otimes \cdot )$ leads to the following correspondence of simple tensors; $1 \leftrightarrow3$, $2 \leftrightarrow 4$, $5 \leftrightarrow 7$, $6 \leftrightarrow 8$. Then for $X \in  \mathcal{H}^{{\otimes^4}}$ and $Y, Z,  X \in  \mathcal{H}^{\otimes^2}$, properties \autoref{tensorprop} imply
\begin{align*}
 \Tr\Big(S_{(1256)(12)(56)} X \otimes Y \otimes Z)&= \Tr\big( X (\overline{Y} \bigotimes {Z}) ^{\dagger}\big) =\innprod{X}{\overline{Y} \bigotimes {Z}}_{HS}\\
 \Tr\Big(S_{(1256)(12)(56)} X \otimes Y \otimes Z)&= \Tr\big( X (\overline{Y} \bigotimes {Z}) ^{\dagger}\big) =\innprod{X}{\overline{Y} \bigotimes {Z}}_{HS}\\
\Tr\Big(S_{(1256)(12)(56)} X \otimes Y \otimes Z)&=  \Tr\big( X (\overline{Y} \bigotimes {Z}) ^{\dagger}\big) =\innprod{X}{\overline{Y} \bigotimes {Z}}_{HS}\\
 \Tr\Big(S_{(1256)(12)(56)}X \otimes Y \otimes Z)&= \Tr\big( X (\overline{Y} \bigotimes {Z}) ^{\dagger}\big) =\innprod{X}{\overline{Y} \bigotimes {Z}}_{HS},\tageq \label{eq:2ndcum4}
\end{align*}
and for $W, X, Y, Z \in X \in \mathcal{H}^{\otimes^2}$, these imply
\begin{align*}
\Tr\Big(S_{(12)(15)(56)(26)} W \otimes  X \otimes Y \otimes Z\Big) &
=\innprod{\overline{W}^{\dagger}X}{\overline{Z}Y^{\dagger}}_{HS}\\ 
\Tr\Big(S_{(12)(16)(56)(25)} W \otimes  X \otimes Y \otimes Z\Big)&=
 \innprod{\overline{W}^{\dagger}X}{\overline{Z}\overline{Y}}_{HS}\\
\Tr\Big(S_{(15)(26)(15)(26)} W \otimes  X \otimes Y \otimes Z\Big)&= 
\innprod{W}{\overline{Y}}_{HS}\innprod{X}{\overline{Z}}_{HS}\\ 
\Tr\Big(S_{(15)(26)(12)(56)} W \otimes  X \otimes Y \otimes Z\Big) &=\innprod{W \widetilde{\bigotimes} \overline{X}}{(\overline{Y} \bigotimes Z)}_{HS} \\ 
\Tr\Big(S_{(16)(25)(12)(56)}  W \otimes  X \otimes Y \otimes Z\Big)&=\innprod{W \widetilde{\bigotimes}_{\top} {X}}{(\overline{Y} \bigotimes Z)}_{HS}. \tageq \label{eq:2ndcum2}
\end{align*}
\end{prop}
Therefore using Proposition \ref{prop:2ndstruc}and in particular \eqref{eq:2ndcum4}\eqref{eq:2ndcum2}, the corresponding terms of the covariance equal
	\begin{align*}
		T \cv({F}_{i_1,i_2}, &{F}_{i_3,i_4})
	=\frac{1}{T}\sum_{k_1,k_2=1}^{\lfloor N/2 \rfloor} \sum_{j_1,j_2=1}^M\delta_{i_1,i_3} \delta_{j_1,j_2}\left[ \biginnprod{\frac{2 \pi}{N}\F^{i_1}_{u_{j_1},\omega_{k_1},-\omega_{k_1},-\omega_{k_2}}}{\F^{i_2}_{u_{j_1},\omega_{k_1-1}} \bigotimes \F^{i_4}_{u_{j_2},\omega_{k_2-1}} }_{HS} +O(\frac{1}{T})\right]\\
	&+\frac{1}{T}\sum_{k_1,k_2=1}^{\lfloor N/2 \rfloor} \sum_{j_1,j_2=1}^M\delta_{i_1,i_4}\delta_{j_1,j_2}\left[ \biginnprod{ \frac{2 \pi}{N}\F^{i_1}_{u_{j_1},\omega_{k_1},-\omega_{k_1},\omega_{k_2-1}}}{ \F^{i_2}_{u_{j_1},\omega_{k_1-1}}\bigotimes \F^{i_3}_{u_{j_2},-\omega_{k_2}}}_{HS} +O(\frac{1}{T})\right]\\
	&+\frac{1}{T}\sum_{k_1,k_2=1}^{\lfloor N/2 \rfloor} \sum_{j_1,j_2=1}^M\delta_{i_2,i_3}\delta_{j_1,j_2}\left[ \biginnprod{\frac{2 \pi}{N}\F^{i_2}_{u_{j_1},-\omega_{k_1-1},\omega_{k_1-1},-\omega_{k_2}}}{\F^{i_1}_{u_{j_1},-\omega_{k_1}} \bigotimes \F^{i_4}_{u_{j_2},\omega_{k_2-1}}}_{HS} +O(\frac{1}{T})\right]\\
	&+\frac{1}{T}\sum_{k_1,k_2=1}^{\lfloor N/2 \rfloor} \sum_{j_1,j_2=1}^M\delta_{i_2,i_4}\delta_{j_1,j_2}\left[ \biginnprod{(\frac{2 \pi}{N}\F^{i_2}_{u_{j_1},-\omega_{k_1-1},\omega_{k_1-1},\omega_{k_2-1}}}{\F^{i_1}_{u_{j_1},-\omega_{k_1}} \bigotimes \F^{i_3}_{u_{j_2},-\omega_{k_2}}}_{HS} +O(\frac{1}{T})\right]\\
	&+\frac{1}{T}\sum_{k_1,k_2=1}^{\lfloor N/2 \rfloor} \sum_{j_1,j_2=1}^M\delta_{i_2,i_4}\delta_{j_1,j_2} \delta_{k_1,k_2}  \left[ 
\innprod{{\F^{\dagger,i_1}_{u_{j_1},-\omega_{k_1}}} \F^{i_2}_{u_{j_1},-\omega_{k_1-1}}}{{\F^{i_2}_{u_{j_1},-\omega_{k_1-1}}}\F^{\dagger,i_3}_{u_{j_2},-\omega_{k_2}}} +O(\frac{1}{M^2})\right]\\
	&+\frac{1}{T}\sum_{k_1,k_2=1}^{\lfloor N/2 \rfloor} \sum_{j_1,j_2=1}^M\delta_{i_2,i_3}\delta_{j_1,j_2} \delta_{k_1-1,k_2} \left[ \innprod{{\F^{\dagger,i_1}_{u_{j_1},-\omega_{k_1}}}\F^{i_2}_{u_{j_1},-\omega_{k_1-1}}}{\F^{i_2}_{u_{j_1},-\omega_{k_1-1}}\F^{i_4}_{u_{j_2},-\omega_{k_2-1}}}_{HS}
 +O(\frac{1}{M^2})\right]
 \\& +\frac{1}{T}\sum_{k_1,k_2=1}^{\lfloor N/2 \rfloor} \sum_{j_1,j_2=1}^M\delta_{i_1,i_3}\delta_{i_2,i_4}\delta_{j_1,j_2} \delta_{k_1,k_2} \ \left[ \innprod{\F^{i_1}_{u_{j_1},\omega_{k_1}}}{\F^{i_2}_{u_{j_1},\omega_{k_1-1}}}_{HS}\innprod{ \F^{i_1}_{u_{j_1},-\omega_{k_1}}}{ \F^{i_2}_{u_{j_1},-\omega_{k_1-1}}}_{HS}
 +O(\frac{1}{M^2})\right]
  \\& +\frac{1}{T}\sum_{k_1,k_2=1}^{\lfloor N/2 \rfloor} \sum_{j_1,j_2=1}^M\delta_{i_1,i_3}\delta_{j_1,j_2} \delta_{k_1,k_2}  \left[\innprod{\F^{i_1}_{u_{j_1},\omega_{k_1}} \widetilde{\bigotimes} \F^{i_1}_{u_{j_1},\omega_{k_1}}}{\F^{i_2}_{u_{j_1},\omega_{k_1-1}} \bigotimes \F^{i_4}_{u_{j_2},\omega_{k_2-1}}}_{HS}
 +O(\frac{1}{M^2})\right]
   \\& +\frac{1}{T}\sum_{k_1,k_2=1}^{\lfloor N/2 \rfloor} \sum_{j_1,j_2=1}^M
 \delta_{i_1,i_4}\delta_{j_1,j_2} \delta_{k_1,k_2-1}\left[\innprod{\F^{i_1}_{u_{j_1},\omega_{k_1}} \widetilde{\bigotimes}_{\top}  \F^{i_1}_{u_{j_1},-\omega_{k_1}}}{ \F^{i_2}_{u_{j_1},\omega_{k_1-1}} \bigotimes \F^{i_3}_{u_{j_2},-\omega_{k_2}}}_{HS}+O(\frac{1}{M^2})\right]
\end{align*}

So that, as $N, M \to \infty$,
\begin{align*}
		T \cv({F}_{i_1,i_2}, {F}_{i_3,i_4})
	&\to   \delta_{i_1,i_3} \frac{1}{8\pi} \int_{-\pi}^{\pi} \int_{-\pi}^{\pi}\int_0^{1}  \biginnprod{\F^{i_1}_{u,\omega_{1},-\omega_{1},-\omega_{2}}}{\F^{i_2}_{u,\omega_{1}} \bigotimes \F^{i_4}_{u,\omega_{2}} }_{HS}du d\omega_1 d\omega_2\\
	&+\delta_{i_1,i_4}  \frac{1}{8\pi} \int_{-\pi}^{\pi} \int_{-\pi}^{\pi}\int_0^{1} \biginnprod{\F^{i_1}_{u,\omega_{1},-\omega_{1},\omega_{2}}}{ \F^{i_2}_{u,\omega_{1}}\bigotimes \F^{i_3}_{u,-\omega_{2}}}_{HS} du d\omega_1 d\omega_2\\
	&+\delta_{i_2,i_3}  \frac{1}{8\pi} \int_{-\pi}^{\pi} \int_{-\pi}^{\pi}\int_0^{1}  \biginnprod{\F^{i_2}_{u,-\omega_{1},\omega_{1},-\omega_{2}}}{\F^{i_1}_{u,-\omega_{1}} \bigotimes \F^{i_4}_{u,\omega_{2}}}_{HS}  du d\omega_1 d\omega_2\\
	&+\delta_{i_2,i_4}  \frac{1}{8\pi} \int_{-\pi}^{\pi} \int_{-\pi}^{\pi}\int_0^{1} \biginnprod{\F^{i_2}_{u,-\omega_{1},\omega_{1},\omega_{2}}}{\F^{i_1}_{u,-\omega_{1}} \bigotimes \F^{i_3}_{u,-\omega_{2}}}_{HS}   du d\omega_1 d\omega_2\\
	&+\delta_{i_2,i_4}  \frac{1}{4\pi} \int_{-\pi}^{\pi} \int_0^{1} 
\innprod{\F^{\dagger,i_1}_{u,-\omega} \F^{i_2}_{u,-\omega}}{\F^{i_2}_{u,-\omega}\F^{\dagger,i_3}_{u,-\omega}}du d\omega\\
	&+\delta_{i_2,i_3}  \frac{1}{4\pi} \int_{-\pi}^{\pi} \int_0^{1} \innprod{\F^{\dagger,i_1}_{u,-\omega}\F^{i_2}_{u,-\omega}}{\F^{i_2}_{u,-\omega}\F^{i_4}_{u,-\omega}}_{HS}du d\omega
 \\& +\delta_{i_1,i_3}\delta_{i_2,i_4} \frac{1}{4\pi} \int_{-\pi}^{\pi} \int_0^{1} \innprod{\F^{i_1}_{u,\omega}}{\F^{i_2}_{u,\omega}}_{HS}\innprod{ \F^{i_1}_{u,-\omega}}{ \F^{i_2}_{u,-\omega}}_{HS}
du d\omega
  \\& +\delta_{i_1,i_3}  \frac{1}{4\pi}  \int_{-\pi}^{\pi} \int_0^{1} \innprod{\F^{i_1}_{u,\omega} \widetilde{\bigotimes} \F^{i_1}_{u,\omega}}{\F^{i_2}_{u,\omega} \bigotimes \F^{i_4}_{u,\omega}}_{HS}du d\omega
   \\& + \delta_{i_1,i_4}  \frac{1}{4\pi} \int_{-\pi}^{\pi} \int_0^{1}\innprod{\F^{i_1}_{u,\omega} \widetilde{\bigotimes}_{\top}  \F^{i_1}_{u,-\omega}}{ \F^{i_2}_{u,\omega} \bigotimes \F^{i_3}_{u,-\omega}}_{HS}du d\omega.
\end{align*}
From this we can now derive the structures of the four distinct components of the asymptotic variance of $\A^{T}_{(i_1,i_2)}$.
\begin{enumerate}
\item Setting $i_1=i_3=i_2=i_4$
\begin{align*}T\text{Var}(F_{i_1})~ \to &  \frac{2}{8\pi} \int_{-\pi}^{\pi} \int_{-\pi}^{\pi}\int_0^{1} 
\biginnprod{\F^{i_1}_{u,\omega_{1},-\omega_{1},-\omega_{2}}}{\F^{i_1}_{u,\omega_{1}} \bigotimes \F^{i_1}_{u,\omega_{2}} }_{HS}du d\omega_1 d\omega_2 \\ 
&+   \frac{2}{8\pi} \int_{-\pi}^{\pi} \int_{-\pi}^{\pi}\int_0^{1} \biginnprod{ \F^{i_1}_{u,\omega_{1},-\omega_{1},\omega_{2}}}{ \F^{i_1}_{u,\omega_{1}}\bigotimes \F^{i_1}_{u,-\omega_{2}}}_{HS}  du d\omega_1 d\omega_2 \\ 
 &+  \frac{2}{4\pi} \int_{-\pi}^{\pi} \int_0^{1} \snorm{{(\F^{i_1}_{u,\omega})}^2}^2_2 du d\omega \\ 
  &+  \frac{1}{4\pi} \int_{-\pi}^{\pi} \int_0^{1}\snorm{\F^{i_1}_{u,\omega}}^4_2 du d\omega \\ 
 &+  \frac{1}{4\pi} \int_{-\pi}^{\pi} \int_0^{1}   \innprod{\F^{i_1}_{u,\omega} \widetilde{\bigotimes} \F^{i_1}_{u,\omega}}{\F^{i_1}_{u,\omega} \bigotimes \F^{i_1}_{u,\omega}}_{HS}du d\omega \\ 
 &+  \frac{1}{4\pi} \int_{-\pi}^{\pi} \int_0^{1} \innprod{\F^{i_1}_{u,\omega} \widetilde{\bigotimes}_{\top} \F^{i_1}_{u,-\omega}}{ \F^{i_1}_{u,\omega}\bigotimes\F^{i_1}_{u,-\omega}}_{HS} du d\omega \tageq \label{eq:VarF1}
\end{align*}
where we used the self-adjointness of the spectral density operator and that, for any function $g: \mathbb{R} \to \mathbb{K}$, we have $\int_{-\pi}^{\pi} g(\omega) d\omega=\int_{-\pi}^{\pi} g(-\omega) d\omega $. From this it follows that the terms 1 and 4, 2 and 3 and 5 and 6 are respectively equal in the limit. 
\item Setting $i_1=i_3$ and $i_2=i_4$, $i_1\ne i_2 \to i_3\ne i_4, i_1 \ne i_4, i_3 \ne i_2$
\begin{align*}
T\text{Var}(F_{i_1i_2}) ~ & \to ~  \frac{1}{8\pi} \int_{-\pi}^{\pi} \int_{-\pi}^{\pi}\int_0^{1}  \biginnprod{\F^{i_1}_{u,\omega_{1},-\omega_{1},-\omega_{2}}}{\F^{i_2}_{u,\omega_{1}} \bigotimes \F^{i_2}_{u,\omega_{2}} }_{HS}du d\omega_1 d\omega_2\\
	&+  \frac{1}{8\pi} \int_{-\pi}^{\pi} \int_{-\pi}^{\pi}\int_0^{1} \biginnprod{\F^{i_2}_{u,-\omega_{1},\omega_{1},\omega_{2}}}{\F^{i_1}_{u,-\omega_{1}} \bigotimes \F^{i_1}_{u,-\omega_{2}}}_{HS}   du d\omega_1 d\omega_2\\
	&+ \frac{1}{4\pi} \int_{-\pi}^{\pi} \int_0^{1} 
\innprod{\F^{i_1}_{u,-\omega} \F^{i_2}_{u,-\omega}}{\F^{i_2}_{u,-\omega}\F^{i_1}_{u,-\omega}}du d\omega\\
& + \frac{1}{4\pi} \int_{-\pi}^{\pi} \int_0^{1} \innprod{\F^{i_1}_{u,\omega}}{\F^{i_2}_{u,\omega}}_{HS}\innprod{ \F^{i_1}_{u,-\omega}}{ \F^{i_2}_{u,-\omega}}_{HS}
du d\omega
  \\& +  \frac{1}{4\pi}  \int_{-\pi}^{\pi} \int_0^{1} \innprod{\F^{i_1}_{u,\omega} \widetilde{\bigotimes} \F^{i_1}_{u,\omega}}{\F^{i_2}_{u,\omega} \bigotimes \F^{i_2}_{u,\omega}}_{HS}du d\omega.
  \tageq \label{eq:VarF1F2}
\end{align*}
\item Setting $i_1=i_2$ and $i_3=i_4$, $i_1\ne i_3 $, we have to due independence
\begin{align}
T \cv({F}_{i_1,i_1}, {F}_{i_3,i_3}) =0.   \tageq \label{eq:cvF11F22}
\end{align}
\item Setting $i_3=i_2, i_4=i_1$, $i_1\ne i_2 \to i_3 \ne i_4, i_1 \ne i_3, i_2 \ne i_4$
\begin{align*}
T \cv({F}_{i_1,i_2}, {F}_{i_2,i_1}) ~~&  \to
   \frac{1}{8\pi} \int_{-\pi}^{\pi} \int_{-\pi}^{\pi}\int_0^{1} \biginnprod{\F^{i_1}_{u,\omega_{1},-\omega_{1},\omega_{2}}}{ \F^{i_2}_{u,\omega_{1}}\bigotimes \F^{i_2}_{u,-\omega_{2}}}_{HS} du d\omega_1 d\omega_2\\
	&+ \frac{1}{8\pi} \int_{-\pi}^{\pi} \int_{-\pi}^{\pi}\int_0^{1}  \biginnprod{\F^{i_2}_{u,-\omega_{1},\omega_{1},-\omega_{2}}}{\F^{i_1}_{u,-\omega_{1}} \bigotimes \F^{i_1}_{u,\omega_{2}}}_{HS}  du d\omega_1 d\omega_2\\
	&+  \frac{1}{4\pi} \int_{-\pi}^{\pi} \int_0^{1} \innprod{\F^{i_1}_{u,-\omega}\F^{i_2}_{u,-\omega}}{\F^{i_2}_{u,-\omega}\F^{i_1}_{u,-\omega}}_{HS}du d\omega
   \\& + \frac{1}{4\pi} \int_{-\pi}^{\pi} \int_0^{1}\innprod{\F^{i_1}_{u,\omega} \widetilde{\bigotimes}_{\top}  \F^{i_1}_{u,-\omega}}{ \F^{i_2}_{u,\omega} \bigotimes \F^{i_2}_{u,-\omega}}_{HS}du d\omega. \tageq \label{eq:cvF1221}
\end{align*}
\item Setting $i_2=i_3=i_1$ with $i_1 \neq i_4$,
\begin{align*}
T\cv(F_{i_1i_1},F_{i_1 i_4})~  \to
&   \frac{1}{8\pi} \int_{-\pi}^{\pi} \int_{-\pi}^{\pi}\int_0^{1}  \biginnprod{\F^{i_1}_{u,\omega_{1},-\omega_{1},-\omega_{2}}}{\F^{i_1}_{u,\omega_{1}} \bigotimes \F^{i_4}_{u,\omega_{2}} }_{HS}du d\omega_1 d\omega_2\\
	&+ \frac{1}{8\pi} \int_{-\pi}^{\pi} \int_{-\pi}^{\pi}\int_0^{1}  \biginnprod{\F^{i_1}_{u,-\omega_{1},\omega_{1},-\omega_{2}}}{\F^{i_1}_{u,-\omega_{1}} \bigotimes \F^{i_4}_{u,\omega_{2}}}_{HS}  du d\omega_1 d\omega_2\\
	&+ \frac{1}{4\pi} \int_{-\pi}^{\pi} \int_0^{1} \innprod{\F^{\dagger,i_1}_{u,-\omega}\F^{i_1}_{u,-\omega}}{\F^{i_1}_{u,-\omega}\F^{i_4}_{u,-\omega}}_{HS}du d\omega
  \\& +  \frac{1}{4\pi}  \int_{-\pi}^{\pi} \int_0^{1} \innprod{\F^{i_1}_{u,\omega} \widetilde{\bigotimes} \F^{i_1}_{u,\omega}}{\F^{i_1}_{u,\omega} \bigotimes \F^{i_4}_{u,\omega}}_{HS}du d\omega. 
 \tageq \label{eq:cvF1112}
\end{align*}
\item Setting $i_2=i_4=i_1$ with $i_1 \neq i_3$,
\begin{align*}
T\cv(F_{i_1i_1},F_{i_3 i_1})~  \to
	&+ \frac{1}{8\pi} \int_{-\pi}^{\pi} \int_{-\pi}^{\pi}\int_0^{1} \biginnprod{\F^{i_1}_{u,\omega_{1},-\omega_{1},\omega_{2}}}{ \F^{i_1}_{u,\omega_{1}}\bigotimes \F^{i_3}_{u,-\omega_{2}}}_{HS} du d\omega_1 d\omega_2\\
	&+  \frac{1}{8\pi} \int_{-\pi}^{\pi} \int_{-\pi}^{\pi}\int_0^{1} \biginnprod{\F^{i_1}_{u,-\omega_{1},\omega_{1},\omega_{2}}}{\F^{i_1}_{u,-\omega_{1}} \bigotimes \F^{i_3}_{u,-\omega_{2}}}_{HS}   du d\omega_1 d\omega_2\\
	&+  \frac{1}{4\pi} \int_{-\pi}^{\pi} \int_0^{1} 
\innprod{\F^{i_1}_{u,-\omega} \F^{i_2}_{u,-\omega}}{\F^{i_1}_{u,-\omega}\F^{i_3}_{u,-\omega}}du d\omega\\
& +\frac{1}{4\pi} \int_{-\pi}^{\pi} \int_0^{1}\innprod{\F^{i_1}_{u,\omega} \widetilde{\bigotimes}_{\top}  \F^{i_1}_{u,-\omega}}{ \F^{i_1}_{u,\omega} \bigotimes \F^{i_3}_{u,-\omega}}_{HS}du d\omega
\tageq \label{eq:cvF1121}
\end{align*}
\end{enumerate}
 
Let $\boldsymbol{e} $ be the identity vector in $\mathbb{R}^4$. We note that we the variance structure can be written in the form
\begin{align*}
\boldsymbol{e} ^\top\big[ \boldsymbol{\Sigma} \nabla g(\boldsymbol{x})\nabla g^{\top}(\boldsymbol{x}) \big]\boldsymbol{e} 
 =\boldsymbol{e} ^\top &\frac{  \boldsymbol{\Sigma}}{(x_1+x_2)^2} \left(\begin{array}{cccc}
\frac{(x_3+x_4)^2}{{(x_1+x_2)}^2}&\frac{(x_3+x_4)^2}{{(x_1+x_2)}^2}&-\frac{x_3+x_4}{ (x_1+x_2)}&-\frac{x_3+x_4}{ (x_1+x_2)}\\
\frac{(x_3+x_4)^2}{{(x_1+x_2)}^2}&\frac{(x_3+x_4)^2}{{(x_1+x_2)}^2}&-\frac{x_3+x_4}{ (x_1+x_2)}&-\frac{x_3+x_4}{ (x_1+x_2)}\\
-\frac{x_3+x_4}{ (x_1+x_2)}&-\frac{x_3+x_4}{ (x_1+x_2)}&1&1\\
-\frac{x_3+x_4}{ (x_1+x_2)}&-\frac{x_3+x_4}{ (x_1+x_2)}&1&1,\\
\end{array}\right) \boldsymbol{e}
\end{align*}
where $\boldsymbol{\Sigma}$ is defined in \eqref{eq:bigSig} where the expression of the individiual entries are given by
\eqref{eq:VarF1}-\eqref{eq:cvF1121} and where the vector is evaluated as in \eqref{eq:x}. Under the null hypothesis $H_0$, the matrix becomes
\begin{align*}
 \boldsymbol{\Sigma} \nabla g(\boldsymbol{x})\nabla g^{\top}(\boldsymbol{x}) &=\frac{1}{\big(\frac{2}{4 \pi}\int_{-\pi}^{\pi} \int_0^{1}\snorm{\F^{i_1}_{u,\omega}}^2_2 du d\omega\big)^2} \left(\begin{array}{cccc}
1&1&-1&-1\\
1&1&-1&-1\\
-1&-1&1&1\\
-1&-1&1&1\\
\end{array}\right)\\  &
\times 
\left(\begin{array}{cccc}
 \V({F}_{i_1,i_1}) &0 &\cv({F}_{i_1,i_1},{F}_{i_1,i_2}) & \cv({F}_{i_1,i_1},{F}_{i_2,i_1})\\
0 &  \V({F}_{i_2,i_2}) &\cv({F}_{i_2,i_2},{F}_{i_1,i_2}) & \cv({F}_{i_2,i_2},{F}_{i_2,i_1})\\
\cv({F}_{i_1,i_1},{F}_{i_1,i_2}) &\cv({F}_{i_2,i_2},{F}_{i_1,i_2}) &  \V({F}_{i_1,i_2}) & \cv({F}_{i_1,i_2},{F}_{i_2,i_1})\\
 \cv({F}_{i_1,i_1},{F}_{i_2,i_1})& \cv({F}_{i_2,i_2},{F}_{i_2,i_1}) &\cv({F}_{i_1,i_2},{F}_{i_2,i_1}) & \V({F}_{i_2,i_1}). 
\end{array}\right) \tageq  \label{bigcov}
\end{align*} 
where we used that $\cv({F}_{i_1,i_1},{F}_{i_2,i_2})=0$ because of the independence assumption. Using the expression of the entries of the asympototic variance \eqref{eq:VarF1}-\eqref{eq:cvF1121} we exploit that all of these components are restricted forms of \eqref{eq:VarF1} and that this equation consists of 6 distinct terms. For the structure of \eqref{bigcov}, a tedious derivation yields
\begin{enumerate}
\item the first and second term of (the fourth order terms) \eqref{eq:VarF1}: $2$ times in $\V({F}_{i_1,i_1})$, $\V({F}_{i_2,i_2})$, once in $\V({F}_{i_2,i_1})$ and $\V({F}_{i_2,i_1}),  \cv({F}_{i_1,i_2},{F}_{i_2,i_1})$, and once in $\cv({F}_{i_1,i_1},{F}_{i_1,i_2}),  \cv({F}_{i_1,i_1},{F}_{i_2,i_1})$. Hence using the number of occurrence in the covariance matrix of these entries, we find these terms to arise $4+2+2-4-4=0$ times and hence cancels.
\item the third term of \eqref{eq:VarF1}: two times in $\V({F}_{i_1,i_1})$, $\V({F}_{i_2,i_2})$, once in $\V({F}_{i_2,i_1})$ and $\V({F}_{i_2,i_1})$,  $\cv({F}_{i_1,i_2},{F}_{i_2,i_1})$, and once in  $\cv({F}_{i_1,i_1},{F}_{i_1,i_2}),  \cv({F}_{i_1,i_1},{F}_{i_2,i_1})$. Hence using the number of occurrence in the matrix, we find this term to arise $4+2+2-4-4=0$ times and hence cancels.
\item  the fourth term of \eqref{eq:VarF1}: once in $\V({F}_{i_1,i_1})$, $\V({F}_{i_2,i_2})$ and in once in $\V({F}_{i_2,i_1})$ and $\V({F}_{i_2,i_1})$ and does not aris in the other components. Therefore, we find a term $\frac{4}{4\pi} \int_{-\pi}^{\pi} \int_0^{1}\snorm{\F^{i_1}_{u,\omega}}^4_2 du d\omega$ to remain.
\item  the fifth term of \eqref{eq:VarF1}: once in $\V({F}_{i_1,i_1})$, $\V({F}_{i_2,i_2})$ and in once in $\V({F}_{i_2,i_1})$ and $\V({F}_{i_2,i_1})$ and off diagonal it occurs in the structures of the form $\cv({F}_{i_1,i_1},{F}_{i_1,i_2})$, which arises 4 times. Hence $4-4=0$ and cancels.
\item the sixth term of \eqref{eq:VarF1}: once in $\V({F}_{i_1,i_1})$, $\V({F}_{i_2,i_2})$, and from the off-diagonal terms once in $\cv({F}_{i_1,i_2},{F}_{i_2,i_1})$ and once in  $\cv({F}_{i_1,i_1},{F}_{i_2,i_1})$. As $\cv({F}_{i_1,i_2},{F}_{i_2,i_1})$ occurs twice off-diagonal in \eqref{bigcov} and $\cv({F}_{i_1,i_1},{F}_{i_2,i_1})$ four times, we find the sixth term of \eqref{eq:VarF1} to occur $2+2-4=0$ and cancels as well.
\end{enumerate}
Thus, under the null hypothesis $H_0$
\begin{align*}
\V(\hat{\A}_{(i_1,i_2)})
\overset{p}{\to} \frac{\frac{4}{4\pi} \int_{-\pi}^{\pi} \int_0^{1}\snorm{\F^{i_1}_{u,\omega}}^4_2 du d\omega}{\big(\frac{2}{4 \pi}\int_{-\pi}^{\pi} \int_0^{1}\snorm{\F^{i_1}_{u,\omega}}^2_2 du d\omega\big)^2} = 4\pi\frac{ \int_{-\pi}^{\pi} \int_0^{1}\snorm{\F^{i_1}_{u,\omega}}^4_2 du d\omega}{\big(\int_{-\pi}^{\pi} \int_0^{1}\snorm{\F^{i_1}_{u,\omega}}^2_2 du d\omega\big)^2}.
\end{align*}
\end{proof}
\begin{proof}[\bf  Proof of \autoref{lem:est_sig}]
We have,
\[
\innprod{I_{p}^{u_j,\omega_k}}{I_{p}^{u_j,\omega_{k-1}}}_{HS} 
=\biginnprod{ \frac{1}{2}\big(I_{i_1}^{u_j,\omega_k}+ I_{i_2}^{u_j,\omega_k}\big)}{\frac{1}{2} \big(I_{i_1}^{u_j,\omega_{k-1}}+ I_{i_2}^{u_j,\omega_{k-1}}\big)}_{HS}.
\]
Hence the proof of theorem \autoref{thm:con} shows that 
\[
\Big( \frac{2 }{T}\sum_{j=1}^M \sum_{k=1}^{\lfloor N/2 \rfloor} \langle I_{p}^{u_j,\omega_k}, I_{p}^{u_j,\omega_{k-1}}\rangle_{HS}\Big)^2 \overset{p}{\to} \big(\frac{2}{4 \pi}\int_{-\pi}^{\pi} \int_0^{1}\snorm{\F^{i_1}_{u,\omega}}^2_2 du d\omega\big)^2,\]
and it remains to show that under $H_0$
\[ \frac{2}{3T}\sum_{j=1}^M\sum_{k=1}^{\lfloor N/2 \rfloor}  \big(\innprod{I_{p}^{u_j,\omega_k}}{I_{p}^{u_j,\omega_{k-1}}}_{HS}\big)^2 \overset{p}{\to} \frac{1}{\pi} \int_{-\pi}^{\pi} \int_0^{1}\snorm{\F^{i_1}_{u,\omega}}^4_2 du d\omega. \]
To this purpose note that
\begin{align*}
\E  \frac{1}{T}\sum_{j=1}^M\sum_{k=1}^{\lfloor N/2 \rfloor}  \big(\innprod{I_{p}^{u_j,\omega_k}}{I_{p}^{u_j,\omega_{k-1}}}_{HS}\big)^2 & =  \frac{1}{T}\sum_{j=1}^M\sum_{k=1}^{\lfloor N/2 \rfloor}  \cv \big(\innprod{I_{p}^{u_j,\omega_k}}{I_{p}^{u_j,\omega_{k-1}}}_{HS}\big) \\& 
  + \frac{1}{T}\sum_{j=1}^M\sum_{k=1}^{\lfloor N/2 \rfloor} \big(\E \innprod{I_{p}^{u_j,\omega_k}}{I_{p}^{u_j,\omega_{k-1}}}_{HS}\big)^2,  \tageq \label{eq:varnum}
\end{align*} 
and that the second term converges to $ \frac{1}{4\pi}\int_{-\pi}^{\pi} \int_0^{1}\snorm{\F^{i_1}_{u,\omega}}^4_2 du d\omega$. For the first term we write
\begin{align*}
\cv & \big(\innprod{I_{p}^{u_j,\omega_k}}{I_{p}^{u_j,\omega_{k-1}}}_{HS}\big) \\ =&\frac{1}{4}\Big( \cv \big(\innprod{I_{i_1}^{u_j,\omega_k}}{I_{i_1}^{u_j,\omega_{k-1}}}_{HS}+\innprod{I_{i_2}^{u_j,\omega_k}}{I_{i_2}^{u_j,\omega_{k-1}}}_{HS} +\innprod{I_{i_1}^{u_j,\omega_k}}{I_{i_2}^{u_j,\omega_{k-1}}}_{HS}+\innprod{I_{2}^{u_j,\omega_k}}{I_{1}^{u_j,\omega_{k-1}}}_{HS}\Big)
 \\ 
=&\frac{1}{4}\Bigg[\V \big(\innprod{I_{i_1}^{u_j,\omega_k}}{I_{i_1}^{u_j,\omega_{k-1}}}_{HS}\big)+\V\big(\innprod{I_{i_2}^{u_j,\omega_k}}{I_{_2}^{u_j,\omega_{k-1}}}_{HS} \big)+\V\big(\innprod{I_{i_1}^{u_j,\omega_k}}{I_{i_2}^{u_j,\omega_{k-1}}}_{HS}\big)+\V\big(\innprod{I_{i_2}^{u_j,\omega_k}}{I_{i_1}^{u_j,\omega_{k-1}}}_{HS}\big) \\& +2\cv(\innprod{I_{i_1}^{u_j,\omega_k}}{I_{i_1}^{u_j,\omega_{k-1}}}_{HS},\innprod{I_{i_1}^{u_j,\omega_k}}{I_{2}^{u_j,\omega_{k-1}}}_{HS}) +2 \cv(\innprod{I_{i_1}^{u_j,\omega_k}}{I_{i_1}^{u_j,\omega_{k-1}}}_{HS},\innprod{I_{i_2}^{u_j,\omega_k}}{I_{i_1}^{u_j,\omega_{k-1}}}_{HS})
\\& +2\cv(\innprod{I_{i_2}^{u_j,\omega_k}}{I_{i_2}^{u_j,\omega_{k-1}}}_{HS},\innprod{I_{1}^{u_j,\omega_k}}{I_{i_2}^{u_j,\omega_{k-1}}}_{HS}) +2 \cv(\innprod{I_{i_2}^{u_j,\omega_k}}{I_{2}^{u_j,\omega_{k-1}}}_{HS},\innprod{I_{i_2}^{u_j,\omega_k}}{I_{1}^{u_j,\omega_{k-1}}}_{HS})\\&
+2\cv(\innprod{I_{i_1}^{u_j,\omega_k}}{I_{i_2}^{u_j,\omega_{k-1}}}_{HS},\innprod{I_{i_2}^{u_j,\omega_k}}{I_{i_1}^{u_j,\omega_{k-1}}}_{HS}) +2\cv(\innprod{I_{i_1}^{u_j,\omega_k}}{I_{i_1}^{u_j,\omega_{k-1}}}_{HS},\innprod{I_{i_2}^{u_j,\omega_k}}{I_{i_2}^{u_j,\omega_{k-1}}}_{HS}) \Bigg]\tageq \label{eq:bigvar}
\end{align*}
Consider the structure of the first two terms of \eqref{eq:bigvar}. By \autoref{thm:highcumF1234}, we are interested in all indecomposable partitions of the array where the summation is overall indecomposable partitions of the array
\[\begin{matrix}
\underbrace{D_{i_1}^{u_{j},\omega_{k}}}_{1}&\underbrace{D_{i_1}^{u_{j},-\omega_{k}}}_2 & \underbrace{D_{i_1}^{u_{j},-\omega_{k-1}}}_3& \underbrace{D_{i_1}^{u_{j},\omega_{k-1}}}_4\\ 
\underbrace{D_{i_1}^{u_{j},-\omega_{k}}}_5 &\underbrace{D_{i_1}^{u_{j},\omega_{k}}}_6&\underbrace{D_{i_1}^{u_{j},\omega_{k-1}}}_7& \underbrace{D_{i_1}^{u_{j},-\omega_{k-1}}}_8
\end{matrix}\]
It is immediate from \autoref{cor:cumbound} that all terms not consisting of second order cumulants will be of lower order. Additionally, certain partitions will be of lower order when it involves a cumulant component that is off the frequency manifold. Indecomposability moreover requires the first row to hook with the second. Those partitions that remain are
\begin{align*}
& \Tr\Big(S_{(15)(26)(34)(78)}\Big( \F^{i_1}_{u_{j},\omega_{k}}\otimes \F^{i_1}_{u_{j},-\omega_{k}}\otimes \F^{i_1}_{u_{j},-\omega_{k-1}} \otimes \F^{i_1}_{u_{j},\omega_{k-1}}+\Eps_2 \Big) \Big)\\
&\Tr\Big(S_{(15)(26)(37)(48)}\Big(  \F^{i_1}_{u_{j},\omega_{k}}\otimes \F^{i_1}_{u_{j},-\omega_{k}}\otimes \F^{i_1}_{u_{j},-\omega_{k-1}} \otimes \F^{i_1}_{u_{j},\omega_{k-1}}+\Eps_2\Big) \Big)\\
& \Tr\Big(S_{(12)(56)(37)(48)}\Big(\F^{i_1}_{u_{j},\omega_{k}}\otimes \F^{i_1}_{u_{j},-\omega_{k}}\otimes \F^{i_1}_{u_{j},-\omega_{k-1}} \otimes \F^{i_1}_{u_{j},\omega_{k-1}}+\Eps_2\Big)\Big)
\end{align*}
Hence, a similar argument as provided in Proposition \autoref{prop:2ndstruc} demonstrates that the orignal order has correspondence $1 \to 3, 2 \to 4, 5 \to 7, 6 \to 8$, if we plug these in, we find 
\begin{align*}
& \Tr\Big(S_{(15)(26)(12)(56)}\Big( \F^{i_1}_{u_{j},\omega_{k}}\otimes \F^{i_1}_{u_{j},-\omega_{k}}\otimes \F^{i_1}_{u_{j},-\omega_{k-1}} \otimes \F^{i_1}_{u_{j},\omega_{k-1}}+\Eps_2 \Big) \Big)\\
&\Tr\Big(S_{(15)(26)(15)(26)}\Big(  \F^{i_1}_{u_{j},\omega_{k}}\otimes \F^{i_1}_{u_{j},-\omega_{k}}\otimes \F^{i_1}_{u_{j},-\omega_{k-1}} \otimes \F^{i_1}_{u_{j},\omega_{k-1}}+\Eps_2\Big) \Big)\\
& \Tr\Big(S_{(12)(56)(15)(26)}\Big(\F^{i_1}_{u_{j},\omega_{k}}\otimes \F^{i_1}_{u_{j},-\omega_{k}}\otimes \F^{i_1}_{u_{j},-\omega_{k-1}} \otimes \F^{i_1}_{u_{j},\omega_{k-1}}+\Eps_2\Big)\Big)
\end{align*}
But because switching the tensors at position $2$ and $5$ has no effect (being the same object), we find using Properties \autoref{tensorprop}
\begin{align*}
 \frac{1}{T}\sum_{j=1}^M\sum_{k=1}^{\lfloor N/2 \rfloor} \text{Var} \big(\innprod{I_{i_1}^{u_j,\omega_k}}{I_{i_1}^{u_j,\omega_{k-1}}}_{HS}\big) & 
= \frac{3}{T}\sum_{j=1}^M\sum_{k=1}^{\lfloor N/2 \rfloor} 
 \innprod{\F^{i_1}_{u_{j},\omega_{k}}\bigotimes \F^{i_1}_{u_{j},\omega_{k}}}{\F^{i_1}_{u_{j},\omega_{k-1}}\bigotimes \F^{i_1}_{u_{j},\omega_{k-1}}}_{HS} 
+O(\frac{1}{M^2})
  \\ & = \frac{3}{T}\sum_{j=1}^M\sum_{k=1}^{\lfloor N/2 \rfloor} \innprod{\F^{i_1}_{u_{j},\omega_{k}}}{\F^{i_1}_{u_{j},\omega_{k-1}}}_{HS} \innprod{\F^{i_1}_{u_{j},\omega_{k}}}{ \F^{i_1}_{u_{j},\omega_{k-1}}}_{HS} 
 +O(\frac{1}{M^2})
    \\ & \to \frac{3}{4\pi} \int_{-\pi}^{\pi} \int_0^{1} \snorm{\F^{i_1}_{u_{j},\omega_{k}}}^2_{2}  \snorm{\F^{i_1}_{u_{j},\omega_{k}}}^2_{2}
  du d\omega  
  \\& =\frac{3}{4\pi} \int_{-\pi}^{\pi} \int_0^{1} \snorm{\F^{i_1}_{u_{j},\omega_{k}}}^4_{2} du d\omega, \tageq \label{eq:VARpooled11}
  \end{align*}
Similarly
\begin{align*}
 \frac{1}{T}\sum_{j=1}^M\sum_{k=1}^{\lfloor N/2 \rfloor} \text{Var} \big(\innprod{I_{i_2}^{u_j,\omega_k}}{I_{i_2}^{u_j,\omega_{k-1}}}_{HS}\big)    =\frac{3}{4\pi} \int_{-\pi}^{\pi} \int_0^{1} \snorm{\F^{i_1}_{u_{j},\omega_{k}}}^4_{2} du d\omega, \tageq \label{eq:VARpooled22}
  \end{align*}
under $H_0$. Consider then the structure of the third and fourth term of  \eqref{eq:bigvar}.  We are again interested in all indecomposable partitions of the array where the summation is overall indecomposable partitions of the array
\[\begin{matrix}
\underbrace{D_{i_1}^{u_{j},\omega_{k}}}_{1}&\underbrace{D_{i_1}^{u_{j},-\omega_{k}}}_2 & \underbrace{D_{i_2}^{u_{j},-\omega_{k-1}}}_3& \underbrace{D_{i_2}^{u_{j},\omega_{k-1}}}_4\\ 
\underbrace{D_{i_1}^{u_{j},-\omega_{k}}}_5 &\underbrace{D_{i_1}^{u_{j},\omega_{k}}}_6&\underbrace{D_{i_2}^{u_{j},\omega_{k-1}}}_7& \underbrace{D_{i_2}^{u_{j},-\omega_{k-1}}}_8
\end{matrix}\]
Because of uncorrelatedness under $H_0$ between series $i_1$ and $i_2$ and that the rows of the array must hook results  in three permutations remainging; $S_{(15)(26)(34)(78)},S_{(15)(26)(37)(48)} and S_{(12)(56)(37)(48)}$. A similar reasoning shows that, for the fifth to eight terms of  \eqref{eq:bigvar} which involve an array of the form
\[\begin{matrix}
\underbrace{D_{i_1}^{u_{j},\omega_{k}}}_{1}&\underbrace{D_{i_1}^{u_{j},-\omega_{k}}}_2 & \underbrace{D_{i_1}^{u_{j},-\omega_{k-1}}}_3& \underbrace{D_{1}^{u_{j},\omega_{k-1}}}_4\\ 
\underbrace{D_{i_1}^{u_{j},-\omega_{k}}}_5 &\underbrace{D_{i_1}^{u_{j},\omega_{k}}}_6&\underbrace{D_{i_2}^{u_{j},\omega_{k-1}}}_7& \underbrace{D_{i_2}^{u_{j},-\omega_{k-1}}}_8,
\end{matrix}\]
will only have the permutation $S_{(15)(26)(34)(78)}$ remaining. Uncorrelatedness under $H_0$ between series $i_1$ and $i_2$ and the lag Fourier imply that the lasts two terms of \eqref{eq:bigvar} will be of lower order. We thus have that the total sum in \eqref{eq:bigvar} consists of $3 \times2+3 \times 2+4 \times 2=20$ terms and similar to the proof of \eqref{eq:VARpooled11} it can be shown these terms all converge to the same limit under $H_0$. Therefore, we find under $H_0$
\begin{align*}
\E \frac{1}{T}\sum_{j=1}^M\sum_{k=1}^{\lfloor N/2 \rfloor}  \cv \big(\innprod{I_{p}^{u_j,\omega_k}}{I_{p}^{u_j,\omega_{k-1}}}_{HS}\big)  =\frac{20}{4} \frac{1}{4\pi}\int_{-\pi}^{\pi} \int_0^{1} \snorm{\F_{u_{j},\omega_{k}}}^4_{2} du d\omega 
\end{align*}
which together with the second term of \eqref{eq:varnum}, implies
 \begin{align*} \E  \frac{1}{T}\sum_{j=1}^M\sum_{k=1}^{\lfloor N/2 \rfloor} \big(\innprod{I_{p}^{u_j,\omega_k}}{I_{p}^{u_j,\omega_{k-1}}}_{HS}\big)^2
 &=   \frac{1}{T}\sum_{j=1}^M\sum_{k=1}^{\lfloor N/2 \rfloor}  \cv \big(\innprod{I_{p}^{u_j,\omega_k}}{I_{p}^{u_j,\omega_{k-1}}}_{HS}\big) \\& 
    + \frac{1}{T}\sum_{j=1}^M\sum_{k=1}^{\lfloor N/2 \rfloor} \big(\E \innprod{I_{p}^{u_j,\omega_k}}{I_{p}^{u_j,\omega_{k-1}}}_{HS}\big)^2
    \\& \to 
  \frac{20}{4}\frac{1}{4\pi} \int_{-\pi}^{\pi} \int_0^{1} \snorm{\F_{u_{j},\omega_{k}}}^4_{2} du d\omega 
 +\frac{1}{4\pi} \int_{-\pi}^{\pi} \int_0^{1} \snorm{\F_{u_{j},\omega_{k}}}^4_{2} du d\omega 
 \\& 
 = \frac{24}{4}\frac{1}{4\pi} \int_{-\pi}^{\pi} \int_0^{1} \snorm{\F_{u_{j},\omega_{k}}}^4_{2} du d\omega 
 \end{align*}
It can be shown along the lines of the proof of \autoref{thm:con} that this is in fact a $\sqrt{T}$ consistent estimator. The joint convergence in probability therefore immediately follows and the result follows from an application of the continuous mapping theorem. 
\end{proof}

\section{Analysis of the spectral clustering algorithm}
\subsection{Consistency of $\hat{L}$ for $L$}

\begin{proof}[Proof of \autoref{lem:Lcons}]
From \autoref{thm:con} we have that $\hat{\A} \in \rnum^{d \times d}$ is a $\sqrt{T}$-consistent estimator of the distance measure $\A$. The continuous mapping theorem therefore implies that $\hat{W}^T$ is consistent, i.e., 
A simple calculation shows that, as $T \to \infty$,
\begin{align} \label{eq:boundW}
 \mathbb{P}( \opnorm{\hat{W}^T-W} \ge \varepsilon) & \le \mathbb{P}( d \max_{i,j}|\hat{W}_{i,j}^T-W_{i,j}| \ge \varepsilon) \to 0.
\end{align}
Similarly, 
\begin{align}  \label{eq:boundD}
 \mathbb{P}( \max_{i} |D_i-\hat{D}_{i}| \ge \varepsilon) = \mathbb{P}( \max_{i} |\sum_{j} \hat{W}^T_{i,j}-\sum_{j} \hat{W}^T_{i,j}| \ge \varepsilon)  \le \mathbb{P}( d \max_{i,j}|\hat{W}_{i,j}^T-W_{i,j}| \ge \varepsilon) \to 0.
\end{align}
Similar to \citet{CR2011}, we 
use the decomposition 
\begin{align*}
\hat{L}-L 
 & = \hat{D}^{-1/2}\hat{W}^T\hat{D}^{-1/2} - {D}^{-1/2}\hat{W}^T {D}^{-1/2} +{D}^{-1/2}\hat{W}^T {D}^{-1/2} -{D}^{-1/2}{W} {D}^{-1/2}  \\
 &
= \big(\hat{D}^{-1/2}-{D}^{-1/2}\big)\hat{W}^T\hat{D}^{-1/2} +
   {D}^{-1/2}\hat{W}^T\big(\hat{D}^{-1/2}- {D}^{-1/2}\big) +{D}^{-1/2}\big(\hat{W}^T -{W}\big) {D}^{-1/2}\\&
=\big(I  -{D}^{-1/2} \hat{D}^{1/2}\big) \hat{D}^{-1/2}\hat{W}^T\hat{D}^{-1/2} +
  \big({D}^{-1/2} \hat{D}^{1/2}\big) \hat{D}^{-1/2}\hat{W}^T\hat{D}^{-1/2} \big(I- \hat{D}^{1/2}{D}^{-1/2}\big) 
  \\
  & ~+{D}^{-1/2}\big(\hat{W}^T -{W}\big) {D}^{-1/2}
\end{align*}
and  bound these terms separately. Note that as $D$ and $\hat{D}$ are degree matrices, they are diagonal with nonnegative entries. We therefore have 
\begin{align*}
\snorm{I  -{D}^{-1/2} \hat{D}^{1/2}}_{\infty} =\max_{i} \Big{|}1-\sqrt{\frac{\hat{D}_{i}}{D_i}}\Big{|} \le \max_{i}  \Big{|}1-{\frac{\hat{D}_{i}}{D_i}} \Big{|} \le  \max_{i} \frac{|D_i-\hat{D}_{i}|}{\min_{i} D_i} .
\end{align*}
The triangle inequality gives
\begin{align*}
\snorm{{D}^{-1/2} \hat{D}^{1/2}}_{\infty} =\snorm{I-\big(I-{D}^{-1/2} \hat{D}^{1/2})}_{\infty} \le 1+\max_{i} \frac{|D_i-\hat{D}_{i}|}{\min_{i} D_i} 
\end{align*}
Additionally, since $\hat{D}_{i}= \sum_{j} {\hat{W}^T_{i,j}}$ it follows that $
\snorm{\hat{D}^{-1/2}\hat{W}^T\hat{D}^{-1/2}}_{\infty} = 1$. Furthermore, \\$
\snorm{{D}^{-1/2}\big(\hat{W}^T -{W}\big) {D}^{-1/2}}_{\infty} \le \frac{1}{\min_{i} D_i} \snorm{\hat{W}^T -{W}}_{\infty}.
$
Therefore,
\begin{align*}
\snorm{\hat{L}-L}_{\infty} & \le \frac{\max_{i} |D_i-\hat{D}_{i}|}{\min_{i} D_i} \Bigg(2+ \frac{\max_{i} |D_i-\hat{D}_{i}|}{\min_{i} D_i}\Bigg) +\frac{1}{\min_{i} D_i} \snorm{\hat{W}^T -{W}}_{\infty}.
\end{align*}
Consequently, \eqref{eq:boundW} and  \eqref{eq:boundD} imply
\[
\forall \varepsilon>0,\quad \lim_{T \to \infty}  \mathbb{P}\big(\snorm{\hat{L}-L}_{\infty} > \varepsilon\big)=0. 
\]
\end{proof}

\subsection{Concentration of $\Vn$}

We shall use \autoref{lem:Lcons} to analyze the concentration of $\Vn$. We first need the following auxiliary lemma, which follows from the Davis-Kahan theorem \citep{DK70}

\begin{lemma} \label{lem:Rotbound}
Let $S \subset \mathbb{R}$ an interval. Let $A, H \in \mathbb{R}^{d \times d}$ be two symmetric matrices and let $\hat{A}= A+H$ denote a perturbed version of A. Denote $\hat{Q}$ and $Q$ be orthornormal matrices of dimension $\mathbb{R}^{d \times k}$ whose column spaces equal the eigenspace of $\hat{A}$ and $A$ respectively. Then there exists an orthonormal rotation matrix $O \in \mathbb{R}^{k \times k}$ such that 
\[
\|\hat{Q}-QO \|_2 \le  \frac{\sqrt{2 k}\snorm{H}_{\infty}}{\delta}
\]
where
\[
\delta = \min\{ | \lambda- s|: \lambda \text{ eigenvalue of } A,\, \lambda \not\in S,\, s \in S \}
\]
\end{lemma}

\begin{proof}[Proof of \autoref{lem:Rotbound}]
Using the singular value decomposition, we can find orthonormal matrices $P_1$ and $P_2$ such that the singular values of $Q^\top \hat{Q}$ are exactly the cosines of the principal angles $\Theta$, i.e., we can find $P_1$ and $P_2$ such that $Q^\top \hat{Q} = P_1 \Sigma P_2^\top$ where the diagonal of $\Sigma$ contains the principal angles between the column space of $\hat{Q}$ and $Q$. Define the rotation matrix $O$ as $O=P_1 P_2^\top$. Then, by definition of the Frobenius norm, the orthonormality of $\hat{Q}$ and $Q$ 
\begin{align*}
\|\hat{Q}-QO \|^2_2  &= \Tr\big((\hat{Q}-QO)^\top(\hat{Q}-QO)\big) 
\\& = 2k -2\Tr( OQ^\top\hat{Q})
\\& 
=2k -2 \Tr(\cos \Theta) =2k -2\sum_{i=1}^k \cos \theta_i
\\& \le 2k -2\sum_{i=1}^k \cos (\theta)^2_i = 2k -2k + 2\sum_{i=1}^k \sin (\theta)^2_i = 2\|\sin \Theta \|^2_2.
\end{align*}
The classical Davis-Kahan theorem then yields
\begin{align*}
\|\hat{Q}-QO \|^2_2  &\le  2 \|\sin \Theta \|^2_2 \le 2 \frac{\| H \|^2_2}{\delta^2}.
\end{align*}
Finally, since $\| H\|^2_2 \le k \max_j |\lambda^H_j|^2 = k \snorm{H}_{\infty}^2$, we obtain
\[
\|\hat{Q}-QO \|_2 \le  \sqrt{2 k}  \frac{\snorm{ H}_{\infty}}{\delta}.
\]
\end{proof}

\begin{Corollary} \label{cor:rotbound}
There exists an orthonormal rotation matrix $O \in \mathbb{R}^{k \times k}$ such that
\[
\|\hat{U}-UO \|_2 \le  \frac{2\sqrt{k}\snorm{\hat{L}-L}_{\infty}}{\lambda_{k+1}}
\]
where $\lambda_{k+1}$ is the $(k+1)$-th smallest eigenvalue of $L$.
\end{Corollary}

\begin{proof}[Proof of \autoref{cor:rotbound}]
By construction, $\hat{L}$ and $L$ are symmetric and it is clear that we can view $\hat{L}$ as a perturbed version of $L$. Additionally, the columns of $\hat{U}$ and $U$ contain the eigenvectors that correspond to the $k$ smallest  eigenvalues of $\hat{L}$ and $L$, respectively. It follows therefore directly from \autoref{lem:Rotbound} that
\[
\|\hat{U}-UO \|_2 \le  \frac{2\sqrt{k}\snorm{\hat{L}-L}_{\infty}}{\delta} \le \frac{2\sqrt{k}\snorm{\hat{L}-L}_{\infty}}{\lambda_{k+1}}.
\]
The last inequality is a consequence of the following observation. The matrix $L$ has exactly k zero eigenvalues. Hence if we take $S=[0, \epsilon)$ for arbitray small $\epsilon >0$ or actually the singleton $S=\{0\}$, then the first k eigenvalues of $L$ all belong to S. The smallest distance between eigenvalues that belong to S and that do not belong to $S$ is thus given by $|0-\lambda_{k+1}|$. Hence $\delta=\lambda_{k+1}$.
\end{proof}

\begin{proof}[Proof of \autoref{cor:conVn}]

We note that by definition we have  $\Vn_{i,\cdot} = \frac{\hat{U}_{i,\cdot}}{\|\hat{U}_{i,\cdot}\|_2}$ and $\mathcal{U}_{i,\cdot} =\frac{{(UO)}_{i,\cdot}}{\|{U}_{i,\cdot}\|_2}$. Therefore standard linear algebra shows 
\begin{align*}
\|\Vn-\mathcal{U} \|^2_2 &= \sum_{i=1}^d \Bignorm{ \frac{\hat{U}_{i,\cdot}}{\|\hat{U}_{i,\cdot}\|_2}-\frac{{(UO)}_{i,\cdot}}{\|{U}_{i,\cdot}\|_2}}^2_2
\\& \le 2 \sum_{i=1}^d  \Bignorm{ \frac{\hat{U}_{i,\cdot} \big( \|{U}_{i,\cdot}\|_2- \|\hat{U}_{i,\cdot}\|_2\big)}{\|\hat{U}_{i,\cdot}\|_2 \|{U}_{i,\cdot}\|_2}}^2_2+ \Bignorm{\frac{\hat{U}_{i,\cdot}-{(UO)}_{i,\cdot}}{\|{U}_{i,\cdot}\|_2}}^2_2
\\& =2 \sum_{i=1}^d   \frac{\big| \|{U}_{i,\cdot}\|_2- \|\hat{U}_{i,\cdot}\|_2 \big|^2}{\|{U}_{i,\cdot}\|^2_2}+ \frac{\|\hat{U}_{i,\cdot}-{(UO)}_{i,\cdot}\|^2_2}{\|{U}_{i,\cdot}\|^2_2}
\\& \le 4 \sum_{i=1}^d  \frac{\|\hat{U}_{i,\cdot}-{(UO)}_{i,\cdot}\|^2_2}{\|{U}_{i,\cdot}\|^2_2}
 \\& \le \frac{4}{\min_{i}\|{U}_{i,\cdot}\|^2_2} \|\hat{U}-{(UO)}\|^2_2  = \frac{4}{\min_{i} D_i} \|\hat{U}-{(UO)}\|^2_2
\end{align*}

The last equality follows since $U$ collects eigenvectors of the form $\sqrt{D}\mathbb{1}_{C_l}$ for $l=1,\ldots,k$, where $\mathbb{1}_{C_l} \in \mathbb{R}^d$ denotes the indicator vector that equals 1 if point $i$ belongs to component $C_l$. This means in particular that $U$ has exactly one nonzero entry per row. A trivial lower bound on  $\min_{i} D_i$ can be given by
\[
\min_{i}\|{U}_{i,\cdot}\|^2_2 \ge \frac{\min_i D_{i}}{\mathcal{C_{\text{max}}}}
\]
where $\mathcal{C}_{\text{max}} =\max_{i}\sum_{i_1 \in G_i} \sum_{i_2 \in G_i}W_{i_1,i_2}$. 
Hence, using \autoref{cor:rotbound} and \autoref{lem:Lcons}
\begin{align*}
\|\Vn-\mathcal{U} \|_2 \le 4 \sqrt{k} \sqrt{\frac{ \mathcal{C_{\text{max}}}}{\min_i D_i}} \frac{\opnorm{\hat{L}-L}}{\lambda_{k+1}} \to 0  \text{ as } T \to \infty.
\end{align*}
\end{proof}

\subsection{Analyzing the $k$-means step}

Using the properties of the row-normalized eigenvectors of $L$, we proceed by providing a definition of the set of points that are clustered correctly and then derive a bound on the complement set.The technique is therefore similar to, among others, \citet{RChY11} and \cite{LR2015}.

\begin{lemma}\label{lem:sigma}
Assume the graph has $k$ components. Let $C^{\star}$ defined in \eqref{eq:kmeansobj} and $\mathcal{U}$ defined in \eqref{eq:popEmb}. Then, the set of correctly clustered points is defined as the complement of the set 
\begin{align}\label{eq:sigma}
\Sigma=\{i: \|C^{\star}_{i,\cdot} - \mathcal{U}_{i,\cdot}\|_2 \ge \frac{1}{\sqrt{2}}\}\end{align}
\end{lemma}
\begin{proof}[Proof of \autoref{lem:sigma}]
By construction and using the properties of the Laplacian, $\mathcal{U}$ has exactly one 1 per row.  All other entries in that row are zero. In total, there are $k$ distinct rows which are orthonormal. Therefore, $\|\mathcal{U}_{i,\cdot}-\mathcal{U}_{j,\cdot}\|_2=0$ if the embedded points $i$ and $j$ belong to the same component and $\|\mathcal{U}_{i,\cdot}-\mathcal{U}_{j,\cdot}\|_2=\sqrt{2}$ if they belong to different components. At the same time, $\|{C}^{\star}_{i,\cdot}-{C}^\star_{j,\cdot}\|_2=0$ if and only if the algorithm has clustered $i,j$ in the same cluster. So let, $i$ and $j$ belong to $\Sigma^\mathsf{c}$. Minkowski's inequality yields 
\[
\|\mathcal{U}_{i,\cdot}-\mathcal{U}_{j,\cdot}\|_2 \le \| \mathcal{U}_{i,\cdot}-{C}^\star_{i,\cdot}\|_2 +\|C^{\star}_{i,\cdot} -C^{\star}_{j,\cdot}\|_2+ \|C^{\star}_{j,\cdot} - \mathcal{U}_{j,\cdot}\|_2 \le 2 \frac{1}{\sqrt{2}} = \sqrt{2}
\]
if and only if $i$ and $j$ are clustered in the same cluster. Otherwise, we have a contradiction. Additionally, since  $C{^\star} \in \mathcal{M}(d,k)$ clusters cannot be split. Hence, points in $\Sigma^\mathsf{c}$ must be correctly clustered
\end{proof}
\begin{proof}[Proof of \autoref{thm:miscl}]
 First note that $\mathcal{U}\in  \mathcal{M}(d,k)$ since it has exactly $k$ distinct rows. Consequently, 
\[\argmin_{C \in \mathcal{M}(d,k)} \|\Vn-C \|^2_2= \|\Vn-C^\star \|^2_2\le \|\Vn-\mathcal{U}\|^2_2\]
\begin{align*}
|\Sigma|=\sum_{i \in \Sigma} 1 \le & \sum_{i \in \Sigma}2\|{C}^\star_{i,\cdot} - \mathcal{U}_{i,\cdot}\|^2_2 \\
&\le2\|{C}^\star - \mathcal{U}\|^2_2 \\
&\le 4 \big(\|{C}^\star- \Vn\|^2_2+\|\Vn - \mathcal{U}\|^2_2 \big)\\
&= 8\|\Vn-\mathcal{U}\|^2_2
\end{align*}
The result now follows from \autoref{cor:conVn}.
\end{proof}
\end{appendices}


\end{document}